\documentclass[11pt]{article}
\usepackage{fullpage}

\usepackage{amsthm}
\usepackage{appendix}
\usepackage{hyperref}
\usepackage{tikz}

\usepackage{amsmath}
\usepackage{amssymb}
\usepackage{graphicx}
\usepackage{amsfonts}
\usepackage{latexsym}
\usepackage{pdfsync}
\usepackage{subfigure}
\usepackage{color}
\usepackage{appendix}

 \usepackage{multirow}
  \usepackage{tabularx}
 \usepackage[english]{babel}
 \usepackage{array}

\newcommand{\commentout}[1]{}

\newcommand{\nwc}{\newcommand}
\nwc{\ba}{\begin{array}}
\nwc{\bal}{\begin{align}}
\nwc{\bea}{\begin{eqnarray}}
\nwc{\beq}{\begin{eqnarray}}
\nwc{\bean}{\begin{eqnarray*}}
\nwc{\beqn}{\begin{eqnarray*}}
\nwc{\beqast}{\begin{eqnarray*}}

\nwc{\ea}{\end{array}}
\nwc{\eal}{\end{align}}
\nwc{\eea}{\end{eqnarray}}
\nwc{\eeq}{\end{eqnarray}}
\nwc{\eean}{\end{eqnarray*}}
\nwc{\eeqn}{\end{eqnarray*}}
\nwc{\eeqast}{\end{eqnarray*}}

\nwc{\ep}{\varepsilon}
\nwc{\ept}{\epsilon}

\newtheorem{proposition}{Proposition}
\newtheorem{corollary}{Corollary}
\newtheorem{theorem}{Theorem}

\newtheorem{lemma}{Lemma}
\newtheorem{remark}{Remark}

\nwc{\nn}{\nonumber}

\newcommand{\NN}{\mathbb{N}}

\newcommand{\CC}{\mathbb{C}}
\newcommand{\RR}{\mathbb{R}}
\newcommand{\ZZ}{\mathbb{Z}}
\newcommand{\TT}{\mathbb{T}}
\newcommand{\om}{\omega}

\newcommand{\rank}{\text{Rank\,}}
\newcommand{\supp}{\mathcal{S}}
\nwc{\calP}{\mathcal{P}}
\nwc{\calL}{\mathcal{L}}

\nwc{\xmax}{x_{\text{max}}}
\nwc{\xmin}{x_{\text{min}}}
\nwc{\yep}{y^\ep}
\nwc{\range}{{\rm Range}}
\nwc{\calE}{\mathcal{E}}
\nwc{\calN}{\mathcal{N}}
\nwc{\smax}{\sigma_{\rm max}}
\nwc{\smin}{\sigma_{\rm min}}
\nwc{\re}{{\rm Re}}
\nwc{\im}{{\rm Im}}
\nwc{\bv}{{\rm v}}
\nwc{\bw}{{\rm w}}
\nwc{\PO}{\calP_1}
\nwc{\PT}{\calP_2}
\nwc{\PEO}{\calP^\ep_1}
\nwc{\PET}{\calP^\ep_2}
\nwc{\JE}{J^\ep}
\nwc{\RE}{R^\ep}
\nwc{\bc}{{\rm c}}
\nwc{\si}{\sigma}
\nwc{\se}{\sigma^\ep}
\nwc{\UE}{U^\ep}
\nwc{\VE}{V^\ep}
\nwc{\SE}{\Sigma^\ep}
\nwc{\Phit}{ \Phi^{-\frac L 2\rightarrow \frac L 2}}
\nwc{\lan}{\langle}
\nwc{\ran}{\rangle}

\nwc{\HH}{\mathcal{H}}
\nwc{\HHE}{\mathcal{H}^\ep}
\nwc{\EE}{\mathcal{E}}
\nwc{\SI}{\Sigma}
\nwc{\QE}{Q^\ep}

\nwc{\XP}{X(\Phi^{M-L})^T}
\nwc{\plo}{\phi^L(\om)}

\title{MUSIC for Single-Snapshot Spectral Estimation:  Stability and Super-resolution}

\author{Wenjing Liao
\thanks{Statistical and Applied Mathematical Sciences Institute (SAMSI) and Department of Mathematics, Duke University, Durham, NC. Wenjing Liao is grateful to support from NSF DMS 0847388 and SAMSI under grant NSF DMS-1127914. Email: wjliao@math.duke.edu. }
\and
Albert Fannjiang
\thanks{Department of Mathematics, University of California, Davis, CA. Email:  fannjiang@math.ucdavis.edu. }
}

\begin{document}

\maketitle

\begin{abstract}
This paper studies  the  problem of  line spectral estimation in the continuum of a bounded interval with {\em one snapshot} of array measurement. The single-snapshot measurement data is turned into a Hankel data matrix which admits the Vandermonde decomposition and is suitable for the MUSIC algorithm. The MUSIC algorithm amounts to finding the null space (the noise space) of the Hankel matrix, forming the noise-space correlation function and identifying the $s$ smallest local minima  of the noise-space correlation as the frequency set.  

In the noise-free case {\em exact} reconstruction is guaranteed for any arbitrary set of  frequencies as long as the number of measurement data is at least twice the number of distinct frequencies to be recovered. In the presence of noise the stability analysis shows that  the perturbation of the noise-space correlation is proportional to the spectral norm of the noise matrix as long as the latter is smaller than the smallest (nonzero) singular value of the {\em noiseless} Hankel data matrix.  Under the assumption that the true frequencies are separated by at least twice the Rayleigh Length (RL), the stability of the noise-space correlation is proved by means of novel discrete Ingham inequalities which provide bounds on  the largest and smallest nonzero singular values of the { noiseless} Hankel data matrix. 

The numerical performance of MUSIC is tested in comparison with
other algorithms such as BLO-OMP and SDP (TV-min). While BLO-OMP is the stablest algorithm for frequencies separated above 4 RL, MUSIC becomes the best performing one for frequencies separated between 2 RL and 3 RL. Also, MUSIC is more efficient than  other methods. MUSIC truly shines when the frequency separation drops to 1 RL or below when all other methods fail. Indeed, the resolution length of MUSIC  decreases to zero as noise decreases to zero as a power 
law with an exponent much smaller than an upper bound established  by Donoho. 

\end{abstract}


\section{Introduction}

The field of Compressive Sensing (CS) \cite{DonCS} has provided us with a new technology of reconstructing a signal from a small number of linear measurements. With a new exceptions, signals considered in the compressive sensing community are assumed to be sparse under a discrete, finite-dimensional dictionary. 

However, signals arising in applications such as radar \cite{CheneyRadar}, sonar and remote sensing \cite{FSY} are represented by few parameters on a continuous domain. These signals are usually not sparse under any discrete dictionary but can be approximately sparsely represented by indicator functions on a discrete domain. An approximation error, called gridding error \cite{FL,DB} or basis mismatch \cite{FL1,stro3,Chi} exists, manifesting the gap between the continuous world and the discrete world. This issue is well illustrated by the spectral estimation problem \cite{SAS} as follows.

Suppose a signal $y(t)$ consists of linear combinations of $s$ time-harmonic components from the set 
$$\{e^{-2\pi i \om_j t} : \om_j \in \RR, \ j = 1,\ldots,s\}.$$

Consider the noisy signal model
\beq
\label{model}
y^\ep(t) = y(t) + \ep(t), \quad y(t) = \sum_{j=1}^s x_j e^{-2\pi i \om_j t}
\eeq
where $\ep(t)$ is the external noise. 

The task of spectral estimation is to find out the frequency support set $\supp =	\{\om_1, . . . , \om_s\}$ 	and	the	corresponding	amplitudes	$x	=	[x_1, . . . , x_s]^T$	from	a	finite data sampled at, say, $t=0,1,2, \cdots, M\in \NN$. Because the signal $y(t)$ depends nonlinearly
on $\supp$, the main difficulty  of spectral estimation lies in identifying $\supp$. The amplitudes $x$ can be recovered by solving least squares once $\supp$ is found.

More explicitly, denote (with a slight abuse of notation)  $y = [y_k]_{k=0}^{M},$ $\ep = [\ep_k]_{k=0}^{M}$
and $y^\ep = y + \ep \in \CC^{M+1}$,  with $y_k = y(k)$, $\yep_k = \yep(k)$ and $\ep_k = \ep(k)$.
Let 
\beq
\label{imagingvector}
\phi^{M}(\om) = [1  \ e^{-2\pi i \om} \ e^{-2\pi i 2\om} \ \ldots \ \ e^{-2\pi i M\om}]^T \in \CC^{M+1}
\eeq 
be  the imaging vector of size $M+1$ at the frequency $\om$ and
define 
$$\Phi^{M}= [\phi^{M}(\om_1)\  \phi^{M}(\om_2) \ \ldots \ \phi^{M}(\om_s)] \in \CC^{(M+1)\times s }.$$ 
The single-snapshot  formulation of spectral estimation takes the form 
\beq
y^\ep = \Phi^{M}  x + \ep. 
\label{linearsystem}
\eeq
Again the main difficulty is in the (nonlinear) dependence of $\Phi^M$ on the unknown frequencies
in $\supp$.  
In addition, with the sampling times $t=0,1,2, \cdots, M\in \NN$, one can only hope to determine frequencies
on  the torus $\TT=[0,1)$ with the natural metric 
\[
d(\om_j,\om_l) = \min_{n\in \ZZ} |\om_j+n-\om_l|.
\]

One can attempt to linearize (\ref{linearsystem}) by expanding the matrix $\Phi^M$
via  setting up a grid 
\beq
\label{eq4}
\mathcal{G} = \left\{\frac 0  N , \frac 1 N, \ldots, \frac{N-1}{N}\right\} 
\subset [0,1),
\eeq
 where $N$ is some large integer, and 
writing the spectral estimation problem  in the form  a linear inversion problem
\beq
y^\ep = A  x +\ep
\label{linearsystem1}
\eeq
where
$$A := \left[\phi^{M}\left(\frac 0 N\right)\ \  \phi^{M}\left(\frac 1 N\right) \ \ \ldots \ \ \phi^{M}\left(\frac{N-1}{N}\right)\right] \in \CC^{(M+1)\times N}$$

Discretizing $[0,1)$ as in (\ref{eq4}) amounts to  rounding frequencies on the continuum  to the nearest grid points in $\mathcal{G}$, giving rise to a gridding error which is roughly proportional to the grid spacing. On the other hand, as $N$ increases, correlation among adjacent columns of $A$ also increases dramatically \cite{FL}.

A key unit of frequency separation  is the Rayleigh Length, roughly the minimum resolvable separation of two objects with equal intensities in classical resolution theory \cite{resolutionsurvey,Donoho92}.  
Mathematically, the Rayleigh Length (RL) is the distance between the center and the first zero of the Dirichlet kernel
$$D(\om) = \int_{ -M/2}^{M/2} e^{2\pi i t\om}dt = \frac{\sin{\pi \om M}}{\pi \om}.$$
Hence 1 RL $= 1/M$. 

The ratio $F=N/M$ between RL and the grid spacing  is called  the refinement factor in \cite{FL} and super-resolution factor in \cite{Csr1}. The higher $F$ is, the more coherent the measurement matrix $A$ becomes. 

\subsection{Single-snapshot MUSIC}
In this paper, to circumvent the gridding problem, we  reformulate the spectral estimation problem (\ref{linearsystem})  in the form of {\em multiple measurement vectors} that is suitable for the application of
the MUltiple Signal Classification (MUSIC) algorithm \cite{Sch,SchD}, widely used in signal processing\cite{T82,KJR81,KV96} and array imaging \cite{Cheney,Devaney, Kirsch}. 
\commentout{The MUSIC algorithm was introduced by Schmidt \cite{Sch,SchD} and many extensions  exist including S-MUSIC\cite{SMUSIC}, IES-MUSIC\cite{IESMUSIC}, R-MUSIC\cite{RMUSIC} and RAP-MUSIC\cite{RAPMUSIC}. According to Wikipedia, in a detailed evaluation based on thousands of simulations, M.I.T.'s Lincoln Laboratory concluded that, among currently accepted high-resolution algorithms, MUSIC was the most promising and a leading candidate for further study and actual hardware implementation.} 

Most state-of-the-art spectral estimation methods (\cite{SAS} and references therein)
 assume many snapshots of array measurement as well as 
 statistical assumptions on  measurement noise. In contrast, we pursue below a deterministic approach to spectral estimation with a single snapshot of array measurement in common with \cite{DNN}. 

Fixing a positive integer $1\leq L < M$, we form the Hankel matrix 
\beq
\label{hankel}
H = {\rm Hankel}(y) =
\begin{bmatrix}
y_0 & y_1 & \ldots & y_{M-L}\\
y_1 & y_2 & \ldots & y_{M-L+1}\\
\vdots & \vdots & \vdots & \vdots\\
y_{L} & y_{L+1} & \ldots & y_{M}\\
\end{bmatrix}.
\eeq
Since its first appearance in Prony's method \cite{Prony} the Hankel data matrix (\ref{hankel}) plays an important role in modern methods such as  the state space method \cite{Rao1,Rao2} and the matrix pencil method \cite{MatrixPencil2}. 
 
It is straightforward to verify that  ${\rm Hankel}(y)$ with $y=\Phi^M x$  admits  the Vandermonde decomposition  
\beq
\label{van}
H = \Phi^L X (\Phi^{M-L})^T, \quad X = {\rm diag}(x_1,\ldots,x_s)
\eeq
with the Vandermonde matrix
$$\Phi^L = \begin{bmatrix}
1 & 1 & \ldots & 1 \\
e^{-2\pi i\om_1} & e^{-2\pi i\om_2}& \ldots & e^{-2\pi i\om_s}\\
(e^{-2\pi i\om_1})^2 & (e^{-2\pi i\om_2})^2& \ldots & (e^{-2\pi i\om_s})^2\\
\vdots &\vdots& \vdots & \vdots\\
(e^{-2\pi i\om_1})^{L} & (e^{-2\pi i\om_2})^{L}& \ldots & (e^{-2\pi i\om_s})^{L}\\
\end{bmatrix}.$$
Here we use a  special property of Fourier measurements: a time translation corresponds to a frequency phase modulation. 

Let $H^{\ep} = {\rm Hankel}(\yep)$ and $E = {\rm Hankel}(\ep)$.
The multiple measurement vector formulation of spectral estimation takes the form
\beq
H^{\ep} = H + E = \Phi^L X (\Phi^{M-L})^T + E.
\label{eq5}
\eeq

The crux of MUSIC is this:  In the noiseless case with $L\geq s$ and $M-L+1\geq s$ 
the ranges of $H$ and $\Phi^L$
 coincide and are a proper subspace (the signal space) of $\CC^{L+1}$. 
Let the noise space be the orthogonal complement of the signal space in $\CC^{L+1}$. 
Then $\supp$ can be identified as the zero set of the orthogonal projection of the imaging
vector $\phi^L(\om)$ of size $L+1$ onto the noise space. 

More specifically, let  the Singular Value Decomposition (SVD) of $H$ be written as
$$ H = [\underbrace{U_1}_{(L+1) \times s} \ \underbrace{U_2}_{(L+1) \times (L+1-s)}] \ \underbrace{\text{diag}(\sigma_1,\sigma_2,\ldots,\sigma_s,0,\ldots,0)}_{(L+1) \times (M-L+1)} \ [\underbrace{V_1}_{(M-L+1) \times s} \ \underbrace{V_2}_{(M-L+1) \times (M-L+1-s)}]^\star$$
with the singular values $\sigma_1\geq \sigma_2\geq \sigma_3 \geq \cdots\sigma_s>0.$ 
The signal and noise spaces are exactly the column spaces of  $U_1$ and $U_2$ respectively.  

The orthogonal projection $\calP_2$ onto the noise space is given by $\calP_2 \bw = U_2 (U_2^\star \bw), \,\,\forall \bw \in \CC^{L+1}$. 
Under mild assumptions one can prove that $\om \in \supp$ if and only if $\calP_2\phi^L(\om) = \mathbf{0}$. Hence  $\supp$  can be identified as the zeros of the noise-space correlation function
$$
R(\om) = \frac{\|\PT\phi^L(\om)\|_2}{   \|\phi^L(\om)\|_2} =  \frac{\|U_2^\star\phi^L(\om)\|_2}{\|\phi^L(\om)\|_2},
$$
or the peaks of the imaging function
$$
J(\om) = \frac{\|\phi^L(\om)\|_2}{\|\calP_2\phi^L(\om)\|_2} =  \frac{\|\phi^L(\om)\|_2}{\|U_2^\star\phi^L(\om)\|_2} .
$$

The following fact is the basis  for noiseless MUSIC (See Appendix \ref{secmusic} for proof).
\begin{theorem}
\label{thm1}
Suppose $\om_k \neq \om_l\ \forall k \neq l$. If 
\beq
\label{exact}
L \ge s,\quad M- L+1 \ge s, 
\eeq
then
\begin{center}
$\om \in \supp \Longleftrightarrow R(\om) = 0 \Longleftrightarrow J(\om) = \infty$.
\end{center}
\end{theorem}

\begin{remark}
Condition (\ref{exact}) says that the number of measurement data $(M+1) \ge 2s$ 
suffices to guarantee exact reconstruction by the MUSIC algorithm. \end{remark}

For the noisy data matrix 
$H^\ep$ let  the SVD  be written as   
$$H^\ep = [\underbrace{U^\ep_1}_{(L+1) \times s} \ \underbrace{U^\ep_2}_{(L+1)\times (L+1-s)}] \ \underbrace{\text{diag}(\sigma^\ep_1,\sigma^\ep_2,\ldots,\sigma^\ep_s,\sigma^\ep_{s+1},\ldots)}_{(L+1) \times (M-L+1)} \ [\underbrace{V^\ep_1}_{(M-L+1) \times s} \ \underbrace{V^\ep_2}_{(M-L+1) \times (M-L+1-s)}]^\star$$
with the singular values $\sigma^\ep_1\geq \sigma^\ep_2\geq \sigma^\ep_3 \geq \cdots.$ 
The noise-space correlation function and imaging function become 
$$
\RE(\om) = \frac{\|\PET\phi^L(\om)\|_2}{   \|\phi^L(\om)\|_2} =  \frac{\|{U_2^\ep}^\star\phi^L(\om)\|_2}{\|\phi^L(\om)\|_2} 
$$
and
$$ \JE(\om) = \frac{1}{{\RE}(\om)} =  \frac{   \|\phi^L(\om)\|_2}{\|\PET\phi^L(\om)\|_2} =  \frac{\|\phi^L(\om)\|_2} {\|{U_2^\ep}^\star\phi^L(\om)\|_2},
$$
respectively with $\PET=U^\ep_2(U^\ep_2)^\star$.

The MUSIC algorithm is given by
\begin{center}
   \begin{tabular}{|l|}\hline
    { \centerline{\bf MUSIC for Spectral Estimation}} \\ \hline
    {\bf Input:} $y^\ep \in \CC^{M+1}, s, L$. \\
     1) Form matrix $H^\ep = {\rm Hankel}(y^\ep) \in \CC^{(L+1)\times(M-L+1)}$.
     \\
     2) SVD: $H^\ep = [U_1^\ep\  U_2^\ep] {\rm diag}(\si_1^\ep , \ldots , \si_s^\ep ,\ldots) [V_1^\ep\ V_2^\ep]^\star $, where $U_1^\ep \in \CC^{(L+1)\times s}$.\\
     3) Compute imaging function $J^\ep(\om) = \|\phi^{L}(\om)\|_2 /\|{U_2^\ep}^\star \phi^L(\om)\|_2$. \\
   {\bf Output:} $\hat \supp =\{ \om \text{ corresponding to } s \text{ largest local maxima of } J^\ep(\om) \} $.\\
    \hline
   \end{tabular}
\end{center}

Figure \ref{figm1} shows a noise-space correlation function and an imaging function in the noise-free case. True frequencies are exactly located where the noise-space correlation function vanishes and the imaging function peaks.

\begin{figure}[hthp]
        \centering
       \subfigure[Noise-space correlation $R(\om)$]{\includegraphics[width=6cm]{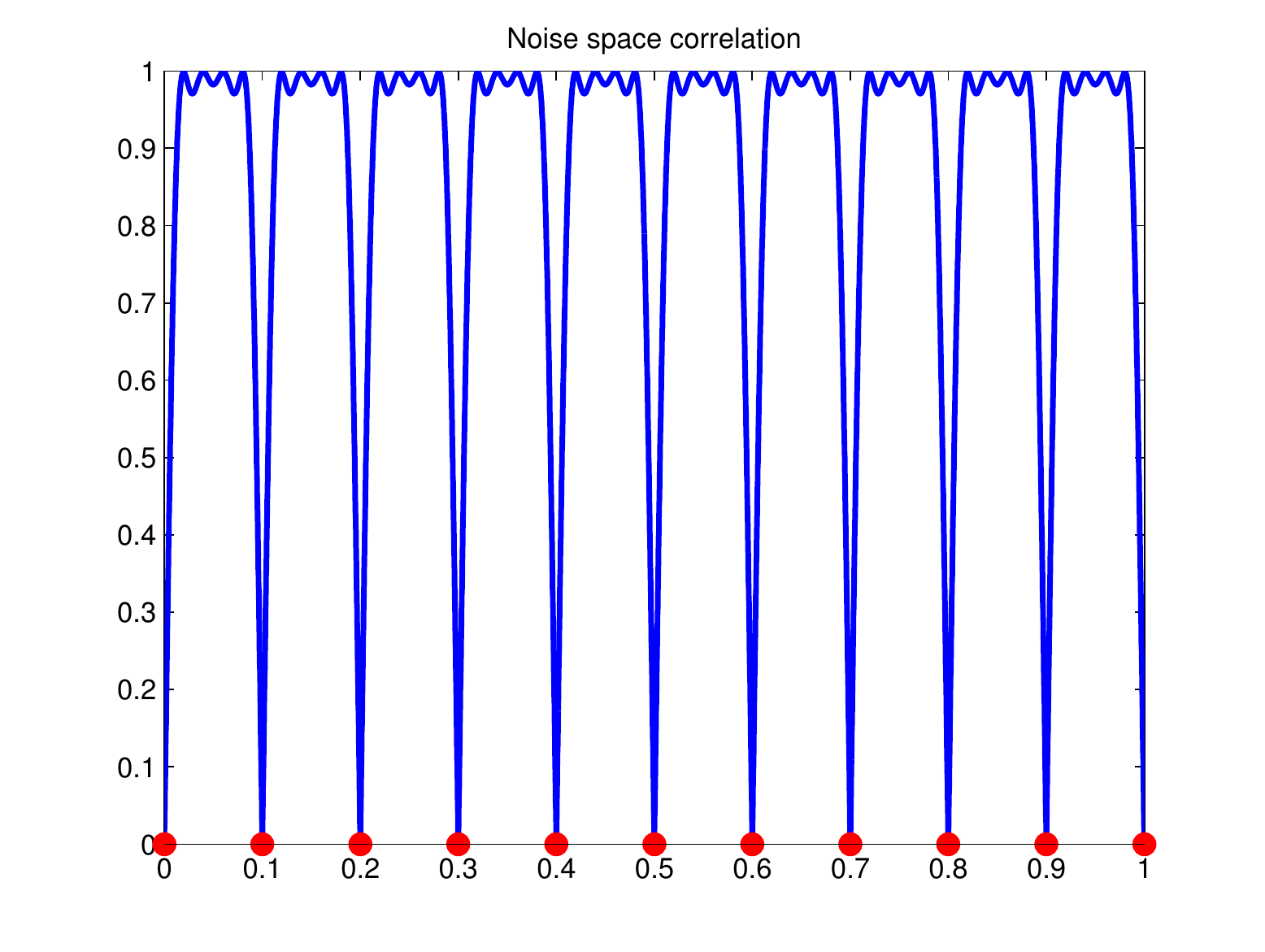}}
        \subfigure[Imaging function $J(\om)$]{\includegraphics[width=6cm]{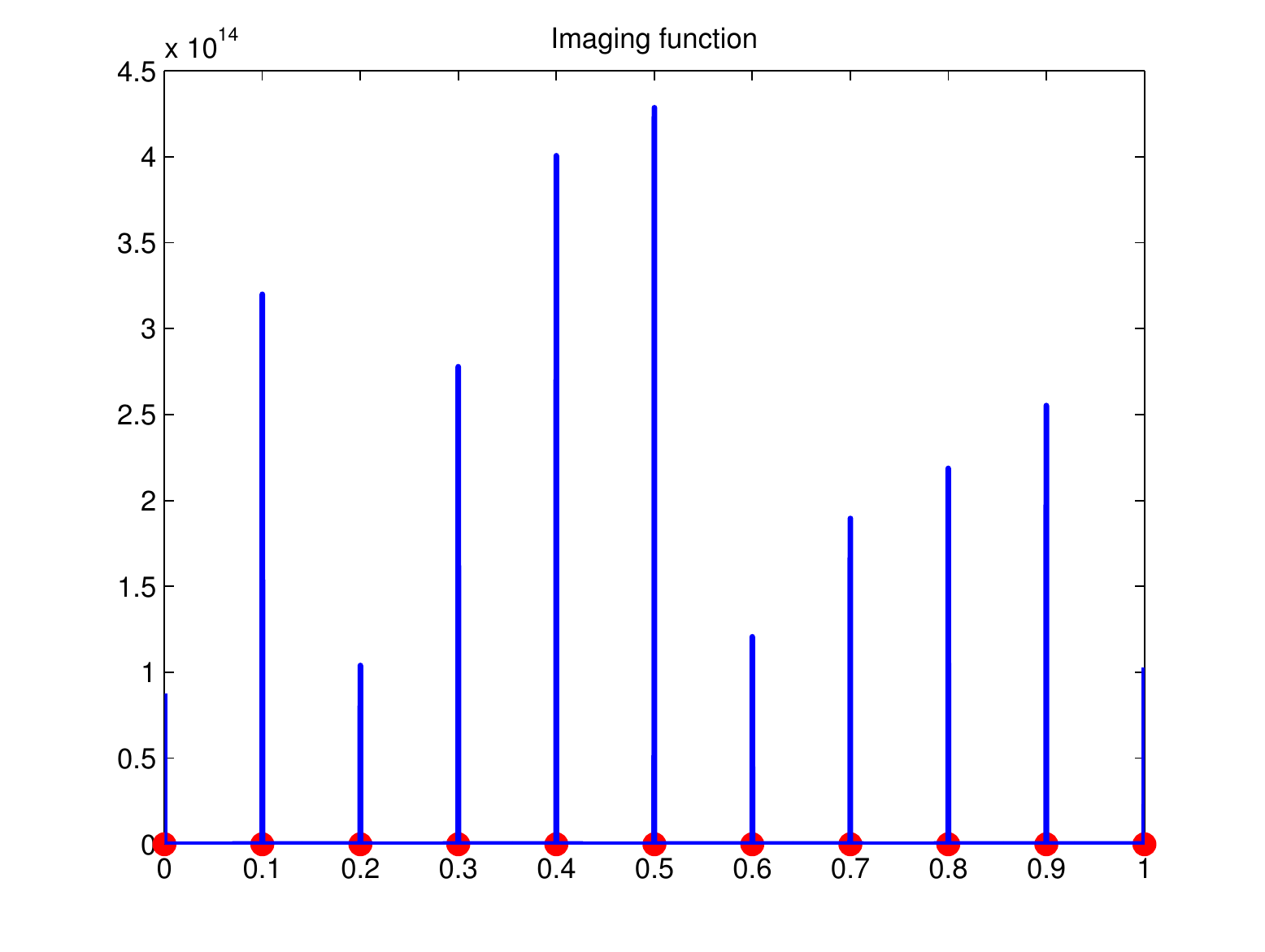}}
\caption{Plots of $R(\om)$ and $J(\om)$ when $M = 100, L = 50$ and there are $10$ equally spaced objects on $\TT$ represented by red dots.}
              \label{figm1}
\end{figure}

The MUSIC algorithm as formulated above requires the number of frequencies $s$ as an input. There are some techniques \cite{ests1,ests2} for evaluating how many objects are present in the event that such information is not available.   
When $\si_s \gg 2\|E\|_2$, $s$ can be easily estimated based on the singular value distribution of $H^\ep$ due to Weyl's theorem \cite{Weyl}. 
\begin{proposition} [Weyl's Theorem]
\label{propweyl}
$|\si^{\ep}_j -\si_j| \le \|E\|_2, \ j = 1,2,\ldots$
\end{proposition}

As a result, $\si^\ep_j \le \|E\|_2, \ \forall  j\ge s+1$ and $\si^\ep_s \ge \si_s-\|E\|_2$. Hence $\se_s \gg \se_{s+1}$, creating a gap between $\sigma_s^\ep$ and
 $\{\si^\ep_j: j\geq s+1\}$. An example is shown in Figure \ref{fig3}.

Before describing our main results, we pause to define notations to be used in the subsequent sections.
For  an $m\times n$ matrix $A$, let $\smax(A)$ and $\smin(A)$ denote  the maximum and minimum nonzero singular values of $A$, respectively. Denote the spectral norm, Frobenius norm and nuclear norm of $A$ by $\|A\|_2$, $\|A\|_F$ and $\|A\|_\star$. 
Let $\xmax = \max_{j=1,\ldots,s} |x_j|$ and $\xmin = \min_{j=1,\ldots,s} |x_j|$. The dynamic range of $x$ is defined as $\xmax/\xmin$. Fixing $\supp = \{\om_1,\ldots,\om_s\}\subset \TT$, we define the matrix $\Phi^{N_1 \rightarrow N_2}$ such that $$\Phi^{N_1 \rightarrow N_2}_{kj} = e^{-2\pi i k \om_j}, \ k = N_1,\ldots,N_2,\ j = 1,\ldots,s.$$
For simplicity, we denote $\Phi^M = \Phi^{0 \rightarrow M}.$
\subsection{Contribution of the present work}\label{sec1.2}
The main contribution of the paper is a stability analysis for the MUSIC algorithm with respect to general support set $\supp$ and external noise. 

In the MUSIC algorithm frequency candidates are identified at the $s$ smallest local minima of the noise-space correlation which measures how much an imaging vector is correlated with the noise space. In noise-free case, the noise-space correlation function $R(\om)$ vanishes exactly on $\supp$. For the noisy case we prove
\beq
\label{eqp'}
|R^\ep(\om) -R(\om)| \le \alpha \|E\|_2,\quad \alpha=\frac{4\si_1+2\|E\|_2}{(\si_s-\|E\|_2)^2}
\eeq
which holds for any support set  $\supp \subset\TT$.


\commentout{
sufficiently small $\|E\|_2$ there exist local minimizers $\{\hat\om_j: j=1,\cdots, s\}$  of
$R^\ep$
 such that  
\beq
\label{eqfre'} 
|\hat\om_j-\om_j| \le
C\alpha \|E\|_2
\eeq
for some $C$ depending on $L$ and the shape of $R^2(\om)$ around $\om_j$. 
}

To make the bounds (\ref{eqp'}) explicit and more meaningful, we
prove  the discrete
Ingham inequalities (Corollary  \ref{cor2}) which implies 
\beq
\label{smin'}
{\sigma_s^2\over {L(M-L)}}&\ge & \xmin^2 
\left( \frac 2 \pi - \frac{2}{\pi L^2 q^2} - \frac 4 L\right)
\left( \frac 2 \pi - \frac{2}{\pi (M-L)^2 q^2} - \frac {4}{M-L}\right)\\
{  \sigma_1^2\over {L(M-L)}}& \le  & \xmax^2 
\left(\frac{4\sqrt 2}{\pi } + \frac{\sqrt 2}{\pi L^2 q^2} + \frac{3\sqrt 2}{L}\right)
\left(\frac{4\sqrt 2}{\pi } + \frac{\sqrt 2}{\pi (M-L)^2 q^2} + \frac{3\sqrt 2}{M-L}\right) \label{smax'}  \eeq
 under the gap assumption 
\beq
\label{eq9}
q=\min_{j\neq l}d(\om_j,\om_l)
 >\max\left(
\frac{1}{L}\sqrt{\frac{2}{\pi} }\left(\frac 2 \pi - \frac 4 L\right)^{-\frac 1 2},
\frac{1}{M-L}\sqrt{\frac{2}{\pi}} \left(\frac 2 \pi - \frac {4} {M-L}\right)^{-\frac 1 2}
\right) . 
  \eeq

Furthermore, we prove that for every $\om_j\in\supp$, there exists a local minimizer $\hat\om_j$
of $R^\ep$ such that $\hat\om_j\to \om_j$ as noise decreases to $0$.

To relax the restriction on the minimum separation between adjacent frequencies, condition (\ref{eq9}) suggests that $L$ should be about $M/2$ and then the resolving power of the present form of MUSIC is as good as  $2/M=$ 2 RL. 

By the results of  \cite{Adam}, the spectral norm of the random Hankel
 matrix $E$ constructed from a zero mean,  independently and identically distributed (i.i.d.)  sequence of a finite variance is on the order of $\sqrt{M \log M}$ for $M\gg 1$ while $\sigma_s$ is on the order
 of $M$ (with $L\approx M/2$). 
In this case the factor $\alpha$ in (\ref{eqp'}) is almost always positive
for sufficiently large $M$ regardless of the variance of noise and 
\[
|\RE(\om)-R(\om)| =\mathcal{O}(M^{-1/2})
\]
up to a logarithmic factor.

Also the super-resolution effect of MUSIC is studied. When the minimum separation between frequencies drops below 1 RL, we show that the noise level that MUSIC can tolerate obeys a power law with respect to the minimum separation with an exponent smaller than an estimate established by Donoho.

Our analysis can be easily extended to other settings where the MUSIC algorithm can be applied, such as the estimation of Directions of Arrivals (DOA) \cite{KV96} and inverse scattering \cite{Cheney,Devaney, Kirsch}.

\subsection{Comparison with other works}\label{sec:1.2}

Among existing works, \cite{MUSIC} is  most closely related to the present work. Central to the results of \cite{MUSIC} is a stability criterion expressed in terms of the Noise-to-Signal Ratio (NSR) $\mathbb{E}(\|\ep\|_2)/\|y\|_2$,
the dynamic range and, when the objects are located exactly on a grid
of spacing $\geq$ 1 RL,  the restricted isometry constants from the theory of compressed sensing \cite{RIP}. The emphasis there is on {\em sparse}  (i.e. undersampling), and typically random,  measurement. 
For the gridless setting considered in the present work, the implications of the analysis in \cite{MUSIC} are not explicit  due to lack of the restricted isometry property for a well-separated set in the continuum. This barrier is overcome in the present work by the discrete Ingham inequalities and the resulting bounds on singular values.  

Other closely related work includes \cite{DNN} and \cite{YY}  where Vandermonde decomposition of the Hankel matrix \eqref{hankel} are used to design different algorithms. 

In \cite{DNN} Demanet {\em et al.} proposed an approach to spectral estimation with a selection step of the support set followed by a pruning step. In the selection step, any $\om$ satisfying $\sin\measuredangle(\plo,\range H^\ep)$ $ \le \eta$ for some judicious choice of $\eta>0$ is kept as a frequency candidate based on their estimate  
$$\sin\measuredangle(\phi^L(\om_j),\range H^\ep) \le C {s\|E\|_2\over \xmin  \smin(\Phi^{M-L})\|\phi^L(\om_j)\|_2},\quad \forall \om_j \in \supp$$
for some constant $C>0$. In comparison, our estimates (\ref{eqp'})-(\ref{smax'}) are more comprehensive as they apply to $\TT$ including $\supp$ and more explicit due to discrete Ingham inequalities. In addition,  the choice of the thresholding parameter $\eta$ can affect the performance of the algorithm in \cite{DNN} while MUSIC  does not contain any thresholding parameter. 

In \cite{YY} Chen and Chi exploited the low-rank property of the Hankel matrix $H$ and applied the matrix completion technique to recover a spectrally  sparse signal from its partial time-domain samples. The focus of \cite{YY} is on the completion and denoising of of data from the partial noisy samples while MUSIC is designed for frequency recovery. A combination of \cite{YY} and our work constitutes a new framework for single-snapshot spectral estimation with compressive noisy measurements which is to be discussed in Section \ref{secconc}.
 
As for frequency recovery, recent progresses center around greedy algorithms and Total Variation (TV) minimization. 

The challenge of applying greedy algorithms to \eqref{linearsystem1} while $N \gg M$ lies in the high coherence and ill conditioning of the sensing matrix $A$. In order to mitigate this effect, we exploited the coherence pattern of $A$ and introduced the techniques of Band exclusion and Local Optimization (BLO) to enhance  standard compressive sensing algorithms. The performance guarantee in \cite{FL} assumes $q\geq 3$ RL and ensures reconstruction of $\supp$ to the accuracy of $1$ RL.    

In \cite{Csr1,Csr2}, Cand\`es and Fernandez-Granda proposed TV minimization and showed that, under the assumption of $q\geq 4$ RL, the TV minimizer yields an $L_1$ reconstruction error linearly proportional to noise with a magnification factor proportional to $F^2$ where $F$ is the refinement/super-resolution factor. 
Inspired by this approach, Tang {\em et al.} \cite{Tang} developed an atomic norm (equivalent to the TV norm in 1D) minimization for the completion of $y$ from its partial samples and showed exact reconstruction in the noise-free case. Like \cite{MUSIC}, a main
emphasis in \cite{Tang} is on sparse measurements. Unfortunately, the effect of noise is not considered in \cite{Tang}. 
For numerical implementation  a SemiDefinite Programming (SDP) on the dual problem is solved in \cite{Tang,Csr2} where  numerical efficiency and stable retrieval of primal solutions may become a problem.

Historically, Prony was the first to address the problem of spectral estimation \cite{Prony}. Unfortunately, Prony's method is numerically unstable and numerous modifications were attempted to improve its numerical behavior. Approximate Prony Method (APM) proposed by Beylkin and Monz\'on in \cite{Beylkin} is a major breakthrough for function approximation by exponential sums. Specifically, Beylking and Monz\'on considered the following problem: given $2N+1$ values of function $f(t)$ on a uniform grid on $[0,1]$ and a target accuracy $\ep > 0$, they find the minimal number $s$ of complex weights $w_j$ and complex nodes $\gamma_j$ such that 
$$\left|
f\left(\frac{k}{2N} \right) - \sum_{j=1}^s w_j \gamma_j^k
\right|
\le \ep, \ \forall k, 0\le k \le 2N.$$
Many interesting examples were provided in \cite{Beylkin}. For  instance, the Bessel function $J_0(100\pi t)$ in $[0,1]$ is approximated by exponential sums of $28$ complex nodes with accuracy $\ep =10^{-10}$ by APM. 

In comparison  the spectral estimation problem  \eqref{model} is the  identification of  $\{\gamma_j\}$ from noisy data, instead of approximation of the signal. 
For spectral estimation with noisy data, APM's  stability may be questionable. The numerical examples of spectral estimation by APM in \cite{Potts} all have low  NSR $=\mathcal{O}(10^{-\delta})$ where $\delta\geq 4$. 
 In contrast  our simulations in Section \ref{secnum} are performed with NSR as large as $0.5$. Furthermore, the super-resolution effect of MUSIC is quantitatively documented in Section \ref{secsup} while it has not been reported in literature whether APM has the capability of localizing closely spaced frequencies.     

In terms of discrete Ingham inequalities, Theorem \ref{thm3} is the first result in which both the gap condition and the upper/lower bounds are explicitly given. In comparison semi-discrete Ingham inequalities in \cite[Lemma 3.1]{Semi,Potts} give
the correct scaling with respect to the size of the Vandermonde matrix $\Phi^L$ without  an explicit estimate for  the constants (cf. (\ref{sv}) and (\ref{sv2}) in Section \ref{seccon}). In other words  the previous Ingham inequalities affirm only that the matrix $\Phi^L$ has a finite condition number under certain gap condition of $\supp$ without an explicit estimate on the magnitude of the condition number.




\commentout{ 
 It is impossible to indefinitely reduce the gridding error while keeping
the side-lobes down. 
On the one hand, to keep the side-lobes down, 
 the grid spacing needs to be at least
one Rayleigh length (RL), which is the reciprocal of the time window $ 1/M $.
On the other hand, to reduce the gridding error, a fine grid 
of spacing $\ell = {\rm RL}/F$ is needed, where $F>1$ is called  the refinement factor in \cite{FL} and super-resolution factor in \cite{Csr1}.  A large $F$ gives rise to an  underdetermined and highly coherent system. 
}

Detailed numerical comparisons of the MUSIC algorithm with Band-excluded Locally Optimized Orthogonal Matching Pursuit (BLOOMP) of \cite{FL}, SemiDefinite Programming (SDP) of \cite{Csr1,Csr2,Tang} and Matched filtering using prolates enhanced by the Band-excluded and Locally Optimized technique \cite{Armin} are presented in Section \ref{secnum}.

Since the SVD step is its primary computational  cost, MUSIC  has low computational complexity compared to other existing methods.
As we will also see, MUSIC is also among the most accurate algorithms.  Finally, MUSIC is the only algorithm that can resolve frequencies with complex amplitudes closely spaced below
1 RL. Indeed, the resolution of MUSIC  can be  arbitrarily small for sufficiently small noise.

\commentout{For the detection of well-separated objects, MUSIC and BLOOMP combine the advantages of strong stability and low computation complexity. SDP is also stable but suffers from long running time even when the problem size is small. In order to pick the right candidates from the SDP solution the band exclusion technique proposed in \cite{FL} is appropriate. BLO-based DPSS is simple but does not work for complex-valued objects. Furthermore, only the MUSIC algorithm has the perfect theory and numerical performance in the noise-free case and the capability of localizing closely spaced objects. 
}
 
The paper is organized as follows. We  estimate nonzero singular values of rectangular Vandermonde matrices with nodes on the unit disk in Section \ref{seccon}. Perturbation theory for MUSIC is presented in in Section \ref{secper} and  super-resolution effect of MUSIC is studied in Section \ref{secsup}. Numerical experiments are provided in Section \ref{secnum}. We finally conclude and discuss extensions of our current work in Section \ref{secconc}.

\section{Vandermonde matrices with nodes on the unit circle}
\label{seccon}

Performance of the MUSIC algorithm in the presence of noise is crucially dependent on $\si_1$ and $\si_s$, the maximum and minimum nonzero singular values of the noiseless Hankel data matrix. To pave a way for the stability analysis, we discuss singular values of the rectangular Vandermonde matrix $\Phi^L$ in this  section.
 
\commentout{
\subsection{Coherence pattern of $\Phi^L$}
The concept of mutual incoherence introduced by Donoho and Huo \cite{DonHuo} has been extensively used in compressive sensing as a measure of the capability of many CS algorithms to correctly detect sparse objects. Mutual coherence represents the maximum pairwise correlation between columns of the sensing matrix. Pairwise coherence between the $j$-th column and the $l$-th column of $\Phi^L$ is defined as
\beq
\label{co1}
\mu(\om_j,\om_l) = \frac{|\lan \phi^L(\om_j),\phi^L(\om_l)\ran |}{\|\phi^L(\om_j)\|_2 \cdot \|\phi^L(\om_l)\|_2}.
\eeq
With deterministically sampled data we can compute the pairwise coherence explicitly, i.e.,
\beq
\label{co2}
\mu(\om_j,\om_l) = \Big|\frac{\sin[\pi(L+1)d(\om_j,\om_l)]}{(L+1)\sin[\pi d(\om_j,\om_l)]}\Big|.
\eeq
We plot $\mu(\om_j,\om_l)$ as a function of $d(\om_j,\om_l)$ in Fig. \ref{fig0}. In fact the pairwise coherence is only dependent on the distance between two objects and it has a clear pattern: 
\begin{enumerate}
\item $\mu(\om_j,\om_l) = 0$ if $d(\om_j,\om_l) = k/(L+1), k = 1,2,\ldots$.
\item $\mu(\om_j,\om_l) \le 1/[2(L+1)d(\om_j,\om_l)]$ as $0 \le d(\om_j,\om_l) \le 1/2$.
\item $\mu(\om_j,\om_l) \rightarrow 1$ as $d(\om_j,\om_l) \rightarrow 0$.
\end{enumerate}
Overall if $d(\om_j,\om_l)$ is below $1/(L+1)$, $\phi^L(\om_j)$ and $\phi^L(\om_l)$ become highly correlated and the pairwise coherence $\mu(\om_j,\om_l)$ is large, yielding a bad conditioning  of $\Phi^L$. By utilizing the coherence pattern of $\Phi^L$, the authors proposed techniques of band exclusion and local optimization embedded in standard CS algorithms to approximately locate objects separated by $3$ RL. 

\begin{figure}[hthp]
        \centering
\includegraphics[width=7cm]{fig/PairwiseCoherenceL16.pdf}
\caption{Pairwise coherence $\mu(\om_j,\om_l)$ versus $d(\om_j,\om_l)$ when $L = 16$. $d(\om_j,\om_l)$ on x-axis is counted in the unit of $1/(L+1)$.}
\label{fig0}
\end{figure}

Coherence of matrix $\Phi^L$ is defined as the maximum pairwise coherence, i.e.,
$$\mu(\Phi^L) = \max_{j \neq l} \mu(\om_j,\om_l).$$
$\mu(\Phi^L)$ is small if and only if objects in $\{\om_j, j = 1, \ldots, s\}$ are pairwise separated. 

According to the Gershgorin circle theorem, we can derive a coarse estimation of nonzero singular values of $\Phi^L$ through the coherence $\mu(\Phi^L)$:  
\beq
1-(s-1)\mu(\Phi^L)\le \frac{1}{L+1}\si^2(\Phi^L)\le 1+(s-1)\mu(\Phi^L).
\label{sinvalue1}
\eeq
However, the bound in \eqref{sinvalue1} is suboptimal and too coarse to be used in practice. A sharper bound is provided as follows.
}


Our estimate is motivated by the classical Ingham inequalities \cite[(pp.162-164)]{Ingham,Young} for non-harmonic Fourier series whose exponents satisfy a gap condition. Specifically Ingham inequalities address the stability problem of complex exponential sums in the system $\{e^{2\pi i \om_j t}, \ t \in [-T/2,T/2],\ \om_j \in \RR,\ j = 1,\ldots,s\}.$ 

\begin{proposition}
\label{propingham}
Let $s \in \NN$ and $T>0$ be given. If the ordered frequencies $\{\om_1,\ldots,\om_s\}$ fulfill the gap condition 
\beq
\om_{j+1}-\om_j \ge q > 1/T, \ j = 1,\ldots, s-1,
\label{seping}
\eeq
then the system of complex exponentials $\{e^{-2\pi i\om_j t}, \ t \in [-T/2,T/2], \ j = 1,\ldots,s\}$ form a Riesz basis of its span in $L^2[-T/2,T/2]$, i.e., 
\beq
\label{ing}
\frac{2}{\pi}\Big(1-\frac{1}{T^2 q^2}\Big)\|{\bc}\|^2_2 
\le  \frac{1}{T}\int_{-T/2}^{T/2} \Big|\sum_{j=1}^s {\bc}_j e^{-2\pi i \om_j t}\Big|^2 dt
 \le 
 \frac{4\sqrt 2}{\pi}\Big(1+\frac{1}{4T^2 q^2}\Big)\|{\bc}\|_2^2
 \eeq
for all complex vectors ${\bc} = ({\bc}_j)_{j=1}^s \in \CC^s$. 
\end{proposition}

\begin{remark}
Ingham inequalities can be considered as a generalization of the Parseval's identity for non-harmonic Fourier series. The gap condition is necessary for a positive lower bound in \eqref{ing} but the upper bound always holds. 
\end{remark}


We prove a discrete version of Ingham inequalities. 

\begin{theorem}
\label{thm3}
 Suppose $\supp$ satisfies  the gap condition 
\beq
\label{sep}
q=\min_{j\neq l} d(\om_j,\om_l) > \frac 1 L\sqrt{\frac{2}{\pi }}\Big(\frac{2}{\pi} - \frac 4 L \Big)^{-\frac 1 2}.
\eeq 
When $L$ is an even integer,
\beq
 \Big(\frac{2}{\pi} - \frac{2}{\pi L^2 q^2}-\frac 4 L\Big)
 \|{\bc}\|_2^2
 \le \frac 1 L \|\Phi^L {\bc}\|_2^2 
 \le 
 \Big(\frac{4\sqrt 2}{\pi}  + \frac{\sqrt 2}{\pi L^2 q^2} + \frac{3\sqrt 2}{L}\Big)\|{\bc}\|_2^2, \ \forall {\bc} \in \CC^s.
 \label{sv}
\eeq
 In other words,
\beq
\label{smax}
\frac 1 L \smax^2(\Phi^L) \le   \frac{4\sqrt 2}{\pi}  + \frac{\sqrt 2}{\pi L^2 q^2} + \frac{3\sqrt 2}{L}
\eeq
and 
\beq
\label{smin}
\frac 1 L \smin^2(\Phi^L) \ge \frac{2}{\pi}- \frac{2}{\pi L^2 q^2}-\frac 4 L .
\eeq
When $L$ is an odd integer,
\beq
\label{sv2}
\left(\frac{2}{\pi}- \frac{2}{\pi L^2 q^2} -\frac 4 L\right) \|\bc\|_2^2
\le
\frac{1}{L}\|\Phi^L {\bc}\|_2^2
\le
\left(1+\frac 1 L\right)\left(\frac{4\sqrt 2}{\pi}  + \frac{\sqrt 2}{\pi (L+1)^2 q^2}+ \frac{3\sqrt 2}{L+1} \right)\|{\bc}\|_2^2, \ \forall {\bc} \in \CC^s.
\eeq
\end{theorem}

Proof of Theorem \ref{thm3} is provided in Appendix \ref{app2}.

\begin{remark}
The difference between the bounds of  the  discrete and the continuous Ingham inequalities  is $\mathcal{O}(1/L)$ which is negligible when $L$ is large. The upper bound in \eqref{sv} holds even when the gap condition \eqref{sep} is violated; however, \eqref{sep} is necessary for the positivity of the lower bound.
\end{remark}
\begin{remark}
Some form of discrete Ingham inequalities are developed in \cite{NZ1,NZ2} for the analysis of the control/observation properties of numerical schemes of the 1-d wave equation. The main result therein is that when time integrals in \eqref{ing} are replaced by discrete sums on a discrete mesh, discrete Ingham inequalities converge to the continuous one as the mesh becomes infinitely fine. Their asymptotic analysis, however, do not provide the non-asymptotic
results stated in Theorem \ref{thm3}. 
\end{remark}


\section{Perturbation of noise-space correlation}
\label{secper}

In this section we use tools in classical matrix perturbation theory \cite{Stewart, Ilse} and develop  a perturbation estimate on the noise-space correlation function, the key ingredient of the MUSIC algorithm. 
Our main results are presented in Theorem \ref{thmp1}, Corollary \ref{cor2}  and Theorem \ref{thmp2} and proofs are provided in Appendix \ref{app1} and \ref{app4}.

\begin{theorem}
\label{thmp1}
Suppose $L \ge s$, $M-L+1\ge s$ and $\|E\|_2 < \si_s$. Then
\beq
\label{eqp}
|R^\ep(\om) -R(\om)| \le \| \PET -\PT \|_2 :=
\sup_{\phi\in \CC^{L+1}} \frac{\|\PET\phi-\PT\phi\|_2}{\|\phi\|_2}
\le \frac{4\si_1+2\|E\|_2}{(\si_s -\|E\|_2)^2} \|E\|_2.
\eeq
In particular, for $\om_j \in \supp$, $R(\om_j) = 0$ and
\beq
\label{eqp2}
|\RE(\om_j)| \le
\frac{2\|E\|_2}{\xmin\smin((\Phi^{M-L})^T)\|\phi^L(\om_j)\|_2}, \  j = 1,2,\ldots,s.
\eeq
\end{theorem}

\begin{remark}
While \eqref{eqp} is a general perturbation estimate valid on $\TT$, including $\supp$, \eqref{eqp2} is a sharper estimate for $\om_j \in \supp$. 
\end{remark}

\begin{remark}
Suppose noise vector $\ep$ contains {\em i.i.d.} random variables of variance $\sigma^2$, $\|E\|_2 \le \|E\|_F = \mathcal{O}(\sigma)$. For fixed $M$ and $\supp$,
$$|\RE(\om)-R(\om)| = \mathcal{O} (\sigma).$$ 
\end{remark}

Theorem \ref{thmp1} holds for  all signal models with {\em any} support set $\supp$. In view of  the Vandermonde decomposition (\ref{van}) 
we next derive explicit bounds for the perturbation of noise-space correlation
 by combining Theorem \ref{thm3} and \ref{thmp1}.  
\begin{corollary}
\label{cor2}
Let $L$ and $M-L$ be even integers. Suppose $\supp$ satisfies the following gap condition 
\beq
\label{eq9'}
q=\min_{j\neq l}d(\om_j,\om_l)
 >\max\left(
\frac{1}{L}\sqrt{\frac{2}{\pi} }\left(\frac 2 \pi - \frac 4 L\right)^{-\frac 1 2},
\frac{1}{M-L}\sqrt{\frac{2}{\pi}} \left(\frac 2 \pi - \frac {4} {M-L}\right)^{-\frac 1 2}
\right) . 
  \eeq
Then 
\beq
\label{ce1}
|\RE(\om)-R(\om)| \le \frac{4\alpha_1+2\frac{\|E\|_2}{\sqrt{L(M-L)}}}{\left(\alpha_2 -\frac{\|E\|_2}{\sqrt{L(M-L)}}\right)^2} \cdot \frac{\|E\|_2}{\sqrt{L(M-L)}},
\eeq
where 
\beq
\label{eqalpha1}
\alpha_1 = \xmax \sqrt{
\left(\frac{4\sqrt 2}{\pi } + \frac{\sqrt 2}{\pi L^2 q^2} + \frac{3\sqrt 2}{L}\right)
\left(\frac{4\sqrt 2}{\pi } + \frac{\sqrt 2}{\pi (M-L)^2 q^2} + \frac{3\sqrt 2}{M-L}\right)}
\eeq
and 
\beq
\label{eqalpha2}
\alpha_2 = \xmin \sqrt{
\left( \frac 2 \pi - \frac{2}{\pi L^2 q^2} - \frac 4 L\right)
\left( \frac 2 \pi - \frac{2}{\pi (M-L)^2 q^2} - \frac {4}{M-L}\right)}.
\eeq
\commentout{
If external noise is bounded such that $\ep_j \le \sigma, \ j = 0,\ldots, M$,
\begin{eqnarray}
|\RE(\om)-R(\om)| 
& \le& \frac{4\alpha_1+2\sigma \sqrt{\frac{(L+1)(M-L+1)}{L(M-L)}}}{\alpha_2^2 - \sigma^2 \frac{(L+1)(M-L+1)}{L(M-L)} }\cdot \sigma \sqrt{\frac{(L+1)(M-L+1)}{L(M-L)}}
\label{eq18}\\
& \approx&  \frac{4\alpha_1 + 2\sigma}{\alpha_2^2 -\sigma^2} \sigma \quad (L,M\gg 1). \nn
\end{eqnarray}
}
\end{corollary}
\commentout{For bounded noises $\ep_j \le \sigma, \ j =0,\ldots, M,$
we can use the bound
$$\|E\|_2 \le \|E\|_F \le \sigma\sqrt{(L+1)(M-L+1)}$$
to derive (\ref{eq18}).
}

\begin{remark}
Corollary \ref{cor2} is stated in the case that both $L$ and $M-L$ are even integers. However, the results hold in other cases with a slightly different $\alpha_1$ based on \eqref{sv} and \eqref{sv2}.
\end{remark}

\begin{remark}
As noted in Section \ref{sec1.2}, for i.i.d. noise $\|E\|_2$ grows like $\sqrt{M\log M}$ which is  much smaller
than $\sqrt{L(M-L)}$ with $L\sim M$ as $M\to \infty$. As a consequence, 
\[
|\RE(\om)-R(\om)| =\mathcal{O}\left(\frac{\sqrt{\log M}}{\sqrt M}\right)
\]
under the gap condition \eqref{eq9'}.
\end{remark}

How close are the MUSIC estimates, namely the $s$ lowest local minimizers of $R^\ep(\om)$, to the true frequencies which are  the  zeros of $R(\om)$? While we can not at the moment answer
this question directly, the following asymptotic result says that 
every true frequency has a near-by strict local minimizer of $R^\ep$ in its vicinity converging to it. Denote $[\plo]' = {d\plo}/{d\om}$. 

\begin{theorem}
\label{thmp2}
Suppose 
\beq
\label{generic}
\left. \PT[\phi^L(\om)]' \right|_{\om = \om_j} \neq \mathbf{0}, \quad \forall \om_j \in \supp.
\eeq
When $\|E\|_2$ is sufficiently small (details in \eqref{p4e3} and \eqref{p4e4}), 
there exists a strict local minimizer $\hat\om_j$ of $\RE(\om)$ near $\om_j$
such that
\beq
\label{eq23}
|\hat\om_j-\om_j|\displaystyle\min_{\xi \in (\om_j,\hat\om_j)} |Q''(\xi)|
\le
4\alpha  \eta(L){\|E\|_2  }
\eeq
where  $Q(\om) = R^2(\om)$, $\eta(L) = 2\pi \sqrt{1^2+2^2+\ldots+L^2}/\sqrt{L+1}$ and
\[ \alpha=\frac{4\si_1+2\|E\|_2}{(\si_s-\|E\|_2)^2}.\]

\end{theorem}

\begin{remark}
\label{rm41}
For  {\em i.i.d.} random variables of variance $\sigma^2$, $\|E\|_2= \mathcal{O}(\sigma)$ for fixed $M$.  Hence $$|\hat\om_j - \om_j| = \mathcal{O} (\sigma)$$ 
as $\sigma\to 0$. 
\end{remark}

\begin{remark}
\label{rm42}
In view of the identity 
\beqn
Q''(\om_j) ={2\over L+1} \|\mathcal{P}_2 [\phi^L(\om)]'\|_2^2\Big|_{\om=\om_j}\neq 0, \quad \forall j
\eeqn
cf. (\ref{Qpp}),
assumption (\ref{generic}) is a generic condition that says the true frequencies are not degenerate minimizers of $R^2$. Figure \ref{fig30} shows $\|\PT [\plo]' \|_2/\|[\plo]'\|_2$  for various $M$ and suggests that for  $q \ge $4 RL, \beq
\label{eqremark6}
C_1 \|[\phi^L(\om)]'\|_2^2 \le \|\PT[\phi^L(\om)]'\|_2^2 \le C_2 \|[\phi^L(\om)]'\|_2^2, \ \forall \ \om\in [0,1)
\eeq
for some constants $C_1,C_2 > 0$ independent of $M$.

\begin{figure}[hthp]
        \centering
       \subfigure[$M = 64$]{\includegraphics[width=5cm]{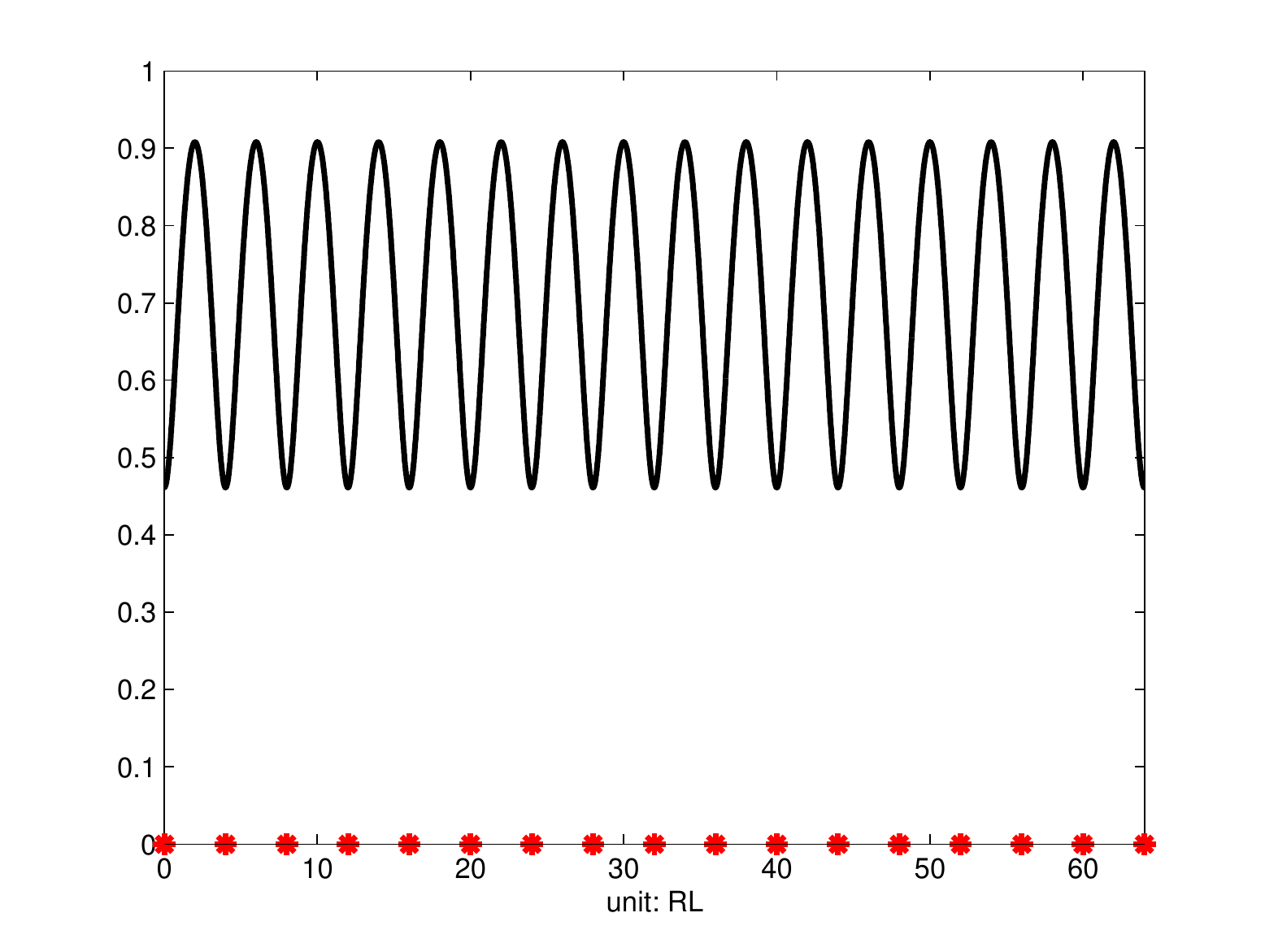}}
       \subfigure[$M = 128$]{\includegraphics[width=5cm]{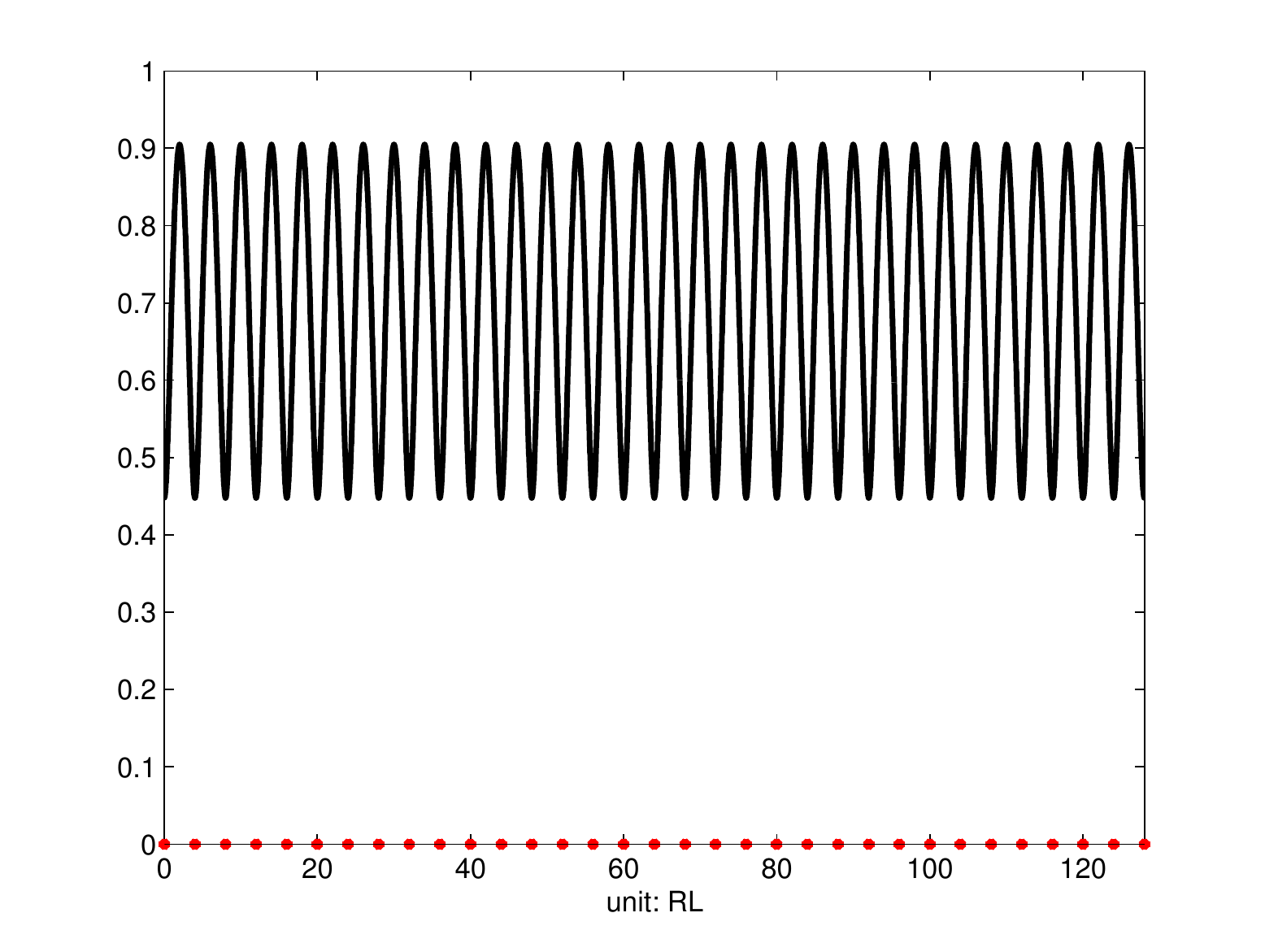}}
        \subfigure[$M = 256$]{\includegraphics[width=5cm]{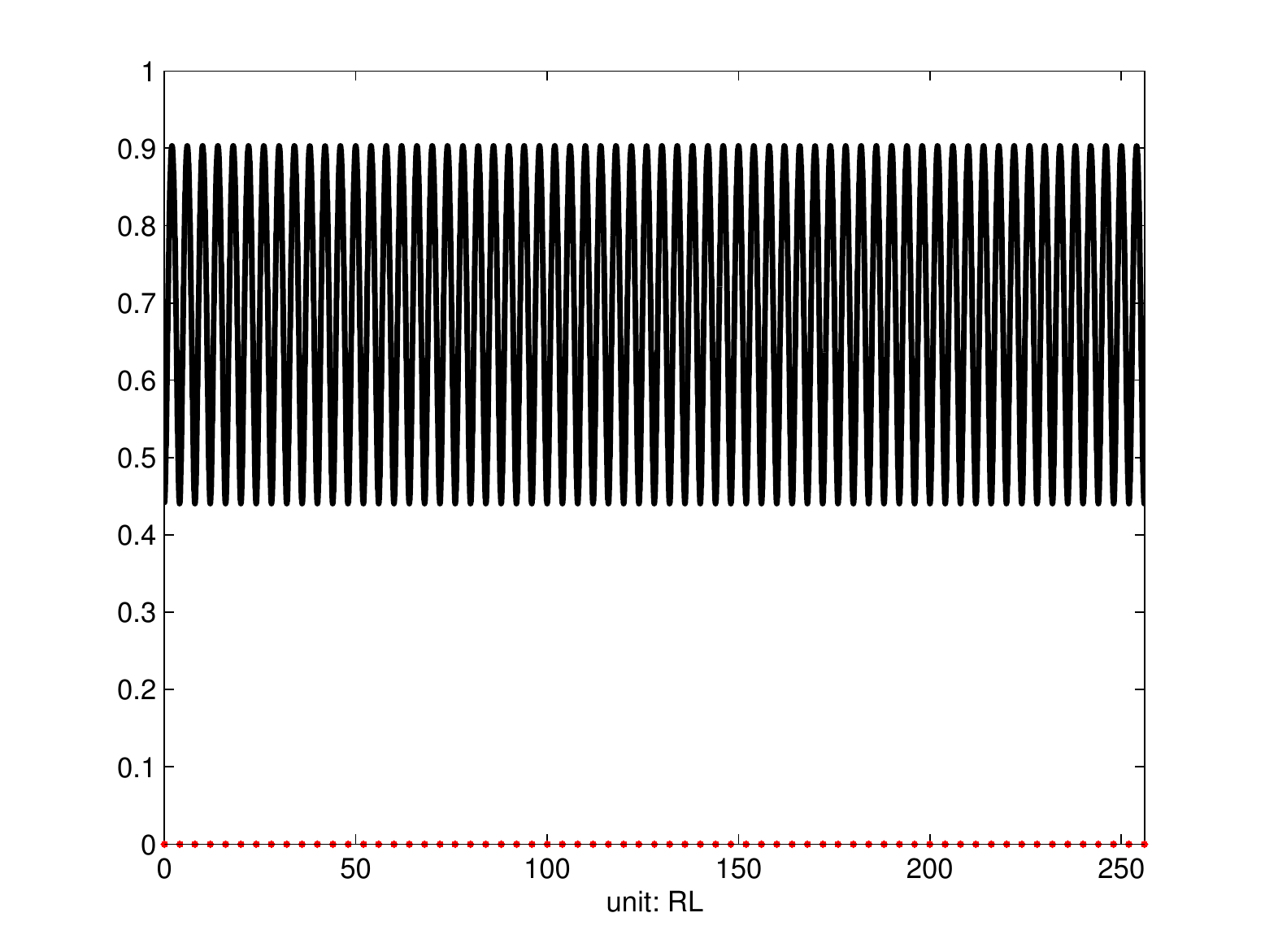}}
       \caption{Function ${\|\PT[\plo]'\|_2}/{\|[\plo]'\|_2}$ with varied $M$ in the case that frequencies in $\supp$ are separated by 4 RL. $L = M/2$ and $\supp$ is marked by red dots.}
         \label{fig30}
\end{figure}

\end{remark}

\begin{remark}
\label{rm43}
Under the assumption (\ref{eq9'}), the right hand side of (\ref{eq23}) scales like
$
\sigma\sqrt{M\log M}
$ with $L = M/2$
for i.i.d. random noise of variance $\sigma^2$. 
Under  the assumption \eqref{eqremark6}, 
$$|Q''(\om_j)| = \frac{2}{L+1} \|\PT[\phi^L(\om_j)]'\|_2^2
\ge \frac{2C_1\|[\phi^L(\om_j)]\|_2^2}{L+1}
= \frac{2C_1(1^2+2^2+\ldots+L^2)}{L+1}
=\mathcal{O}(L^2).$$
As shown in  the proof of Theorem \ref{thmp2} (Step 1), 
\[
\min_{\xi\in(\om_j,\hat\om_j)}|Q''(\xi)| >{1\over 2} Q''(\om_j)=\mathcal{O}(M^2),\quad M=2L\gg 1.
\]
As a result
$$|\hat\om_j - \om_j| = \mathcal{O}\left(\frac{\sqrt{\log M}}{M^{3/2}}\right).$$

\end{remark}

\commentout{
\begin{remark}
We do not provide an explicit estimate  on how small $\|E\|_2$ must be in order to imply the inequality
(\ref{eqfre}) due to the nature of the proof. 
{\color{red} Moreover,  the presence of $\hat\om_j$ on
the right hand side of (\ref{eqfre}) prevents us from inferring that
  $\hat\om_j \rightarrow \om_j$ as $\|E\|_2 \rightarrow 0$. }
  
  The factor $\min_{\xi \in (\om_j,\hat\om_j)} |Q''(\xi)|$ in \eqref{eqfre}, however, does suggest that the larger $ |Q''(\om)|$ in the neighborhood of $\om_j$ are, the more accurate the MUSIC estimates are. Intuitively a pointy minimum of $R^2(\om)$ at $\om_j$ improves  the stability of the MUSIC algorithm. Roughly speaking the same amount of external noise yields a smaller perturbation of the local minimizer in Figure \ref{fig32} (a) than the one in Figure \ref{fig32} (b). 
\begin{figure}[hthp]
        \centering
       \subfigure[]{\includegraphics[width=5cm]{FigureLocalMin/Good.pdf}}
         \subfigure[]{\includegraphics[width=5cm]{FigureLocalMin/Bad.pdf}}
         \caption{}
                     \label{fig32}
\end{figure}
\end{remark}
}

Strong stability in the case of well separated objects is demonstrated in Figure \ref{fig3}. Let $M = 64$, $L = 32$ and then 1 RL = $1/64$. External noise is i.i.d. Gaussian noise, i.e. $\ep \sim N(0,\sigma^2 I)$. Noise-to-Signal Ratio (NSR) $= \mathbb{E}(\|\ep\|_2)/ \|y\|_2 = 20\%$. We display singular values of $H$ and $H^\ep$ in Figure \ref{fig3}(a), noise-space correlation function $R(\om)$ and $\RE(\om)$ in Figure \ref{fig3}(b), and imaging function $\JE(\om)$ in Figure \ref{fig3}(c). With a minimum separation of 4 RL imposed on $\supp$, $\RE(\om)$ is stably perturbed from $R(\om)$. More importantly, every peak of $\JE(\om)$ is near a true object in $\supp$. Simply extracting $s$ largest local maxima of $\JE(\om) $ yields a good reconstruction.

\begin{figure}[hthp]
        \centering
       \subfigure[Singular values of $H$ and $H^\ep$]{\includegraphics[width=5cm]{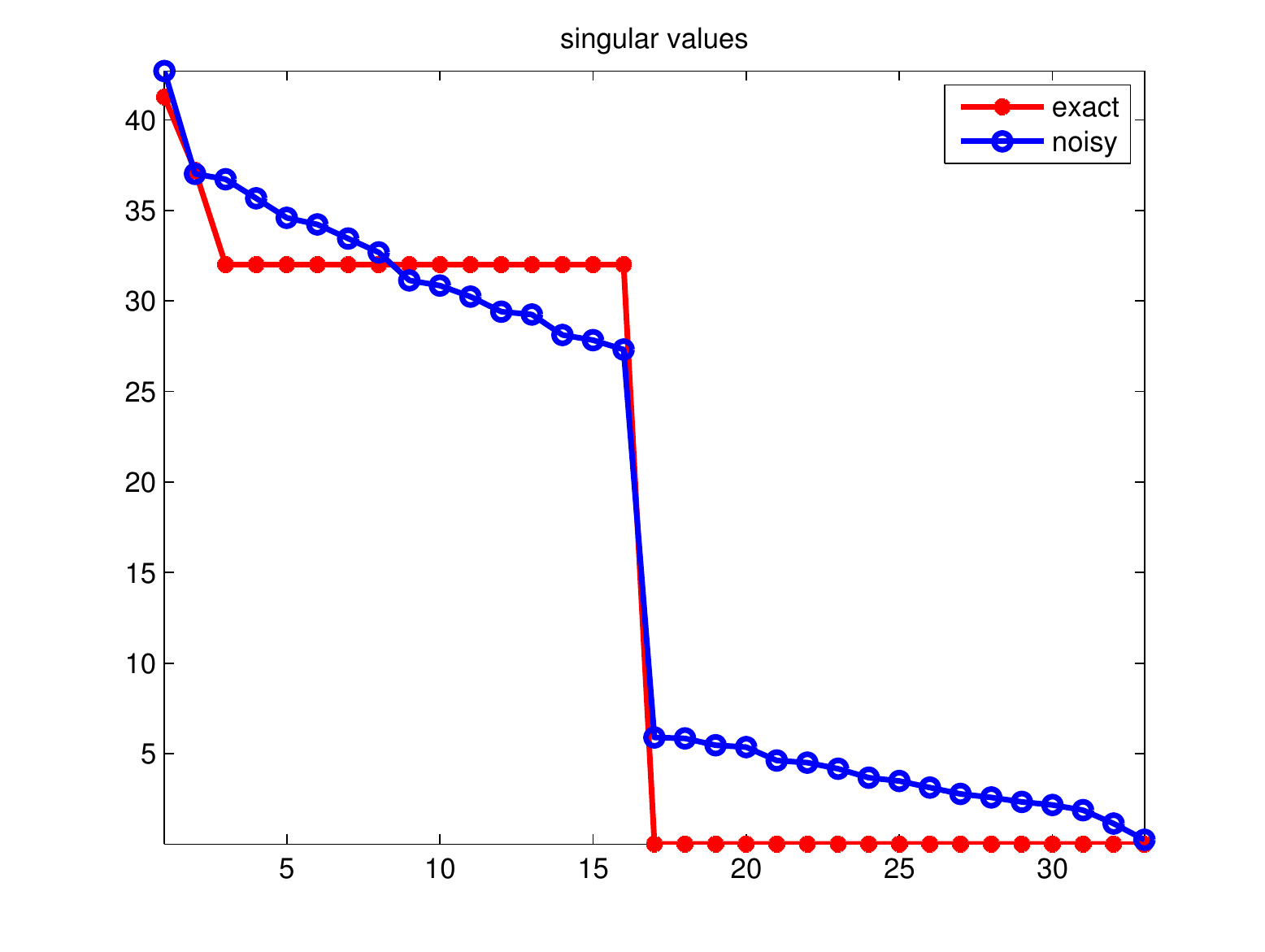}}
                     \subfigure[$R(\om)$ and $\RE(\om)$]{\includegraphics[width=5cm]{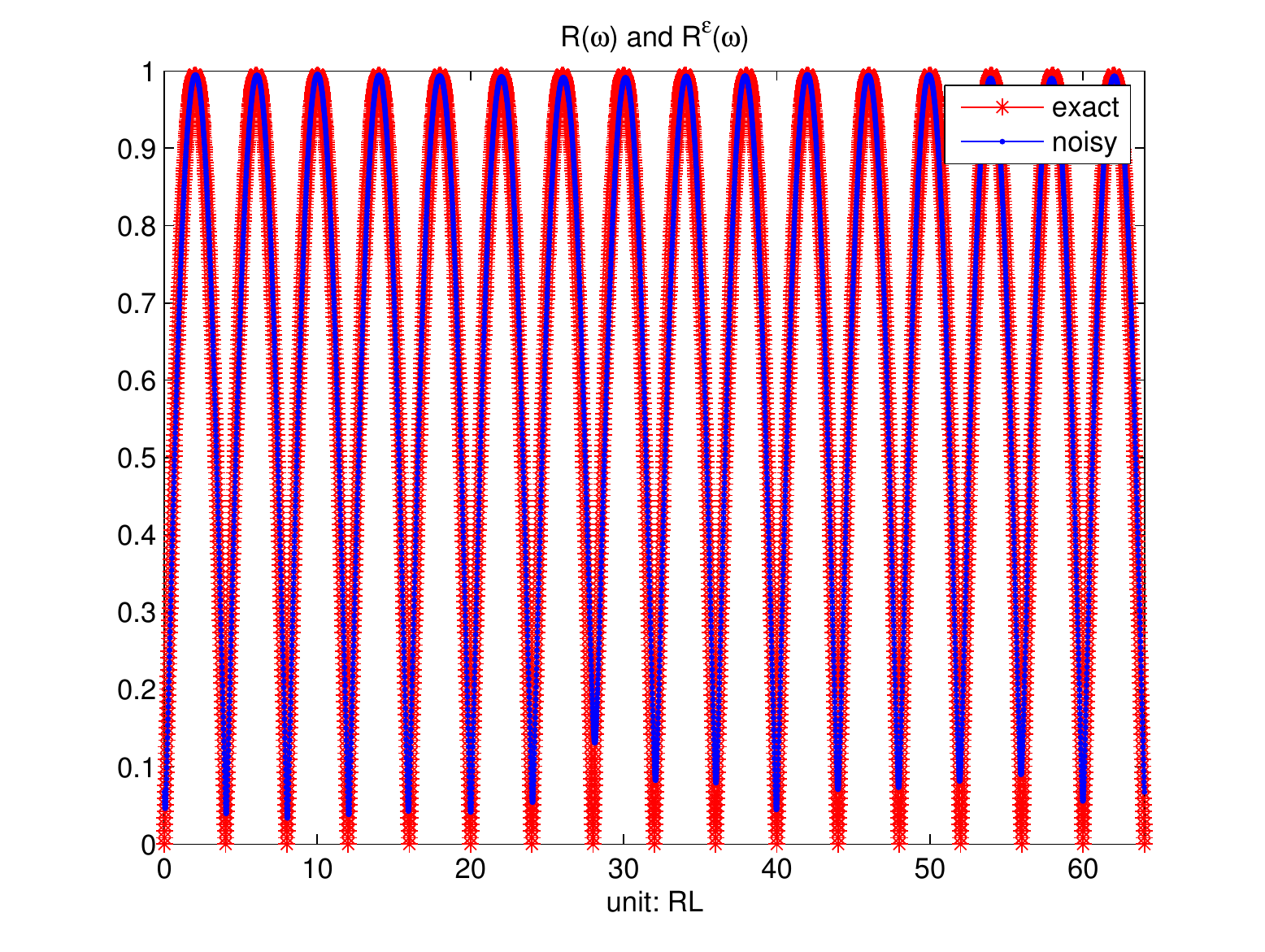}}
        \subfigure[$\JE(\om)$]{\includegraphics[width=5cm]{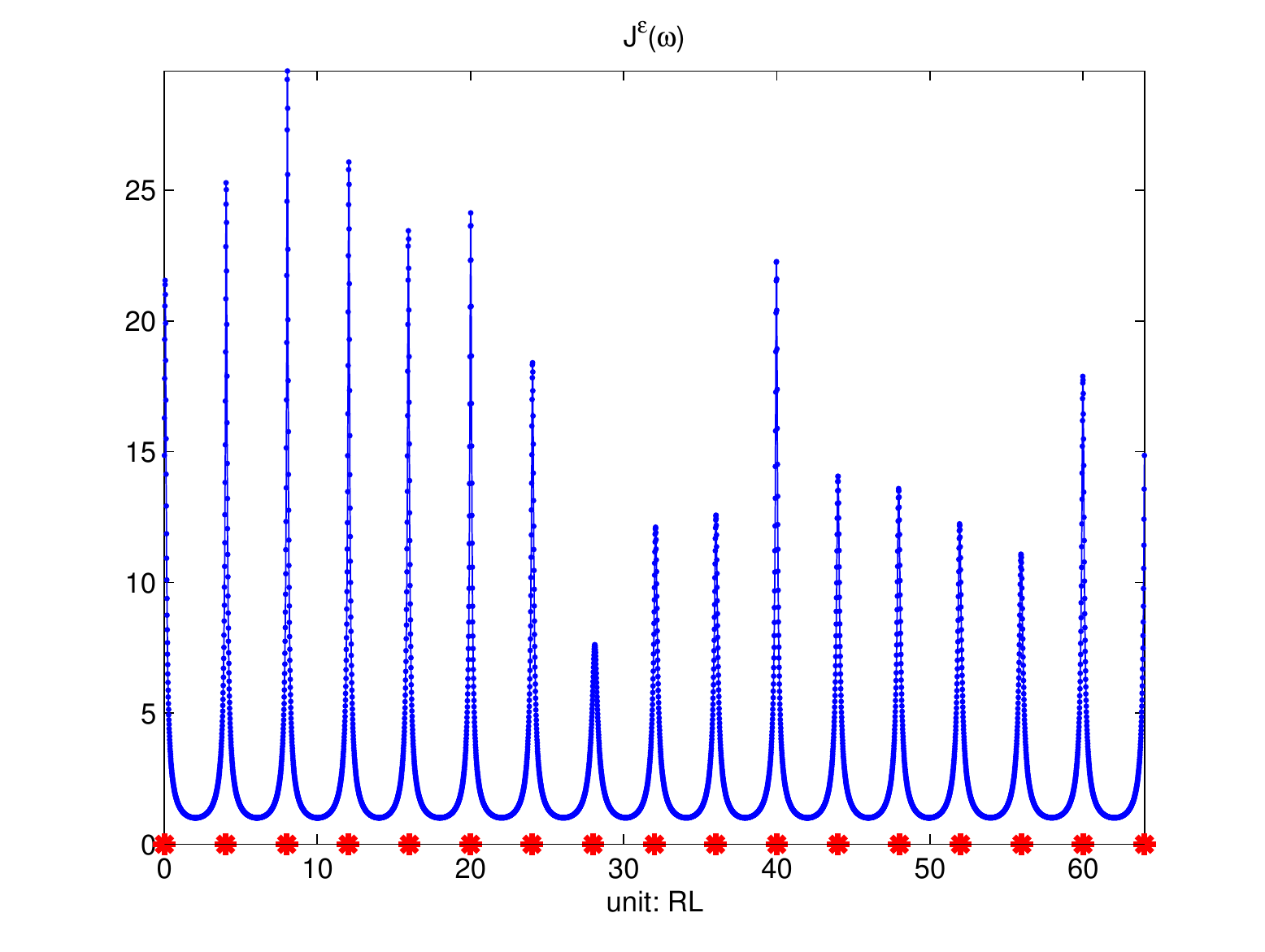}}
\caption{Perturbation of noise-space correlation and imaging function in the case where objects are separated by $4$ RL and NSR = $20\%$. In (c) true objects are located at red dots.}
              \label{fig3}
\end{figure}

\commentout{
Suppose objects in $\supp$ are well separated. We observe that the noise-space correlation $R(\om)$ has a strong pattern: it vanishes exactly on $\supp$ and it is bounded below at all frequencies away from $\supp$. Figure \ref{fig2} shows the noise-space correlation function $R(\om)$ when $M = 64, L=32$ and frequencies in $\supp$ are separated by $2$ RL, $4$ RL and $8$ RL respectively. Notice that $R(\om)$ is not affected by the amplitudes $x$ as the computation of $R(\om)$ only involves the projection of $\plo$ onto the complement of $\range(\Phi^L)$, so the pattern of $R(\om)$ is completely determined by $\supp$.

Next we provide a lower bound for  $R(\om)$ for $\om$ at which the matrix $
\Phi^L_\om = [\Phi^L \ \plo]
$ is well-conditioned (see  Appendix \ref{app3} for proof).

\begin{theorem}
\label{lemma3}
Suppose that
\beq
\label{lemma31}
\smax^2(\Phi^L_\om)-\smin^2(\Phi^L_\om) < 4 \smin^2(\Phi^L)
\eeq
and
\beq
\label{lemma32}
\frac{\smax^2(\Phi^L_\om)-\smin^2(\Phi^L_\om)}{4(L+1)}
<
1-\frac{\smax^2(\Phi^L_\om)-\smin^2(\Phi^L_\om)}{4\smin^2(\Phi^L)},
\eeq
hold at $\om$. 
Then
\beq
\label{lemma33}
R(\om) \ge 
 \sqrt{1-\frac{\smax^2(\Phi^L_\om)-\smin^2(\Phi^L_\om)}{4(L+1)}
\Big(1-\frac{\smax^2(\Phi^L_\om)-\smin^2(\Phi^L_\om)}{4\smin^2(\Phi^L)}\Big)^{-1} 
}.
\eeq
\end{theorem}

\begin{figure}[hthp]
        \centering
       \subfigure[Separation: 2 RL]{\includegraphics[width=5cm]{FigR/RM64Sep2.pdf}}
        \subfigure[Separation: 4 RL]{\includegraphics[width=5cm]{FigR/RM64Sep4.pdf}}
        \subfigure[Separation: 8 RL]{\includegraphics[width=5cm]{FigR/RM64Sep8.pdf}}
\caption{$R(\om)$ when $M = 64, L= 32$ and  frequencies in $\supp$ are separated by (a) $2$ RL, (b) $4$ RL, (c) $8$ RL where RL = 1/64.}
              \label{fig2}
\end{figure}
}

\section{Super-resolution effect of MUSIC}
\label{secsup}

Super-resolution refers to the capability of resolving frequencies separated below 1 RL.
The super-resolution effect of MUSIC has been numerically demonstrated in various applications \cite{Devaney,KV96,MUSIC, SAS}, but theoretical guarantees are lacking. In this section  we aim  to analyze the super-resolution effect in light of the preceding results and
that of \cite{Donoho92}.

With  $2s$ noiseless data MUSIC guarantees to exactly recover $s$ distinct frequencies. With noisy data,  the resulting perturbation of the noise-space correlation function depends crucially on $\si_1$ and $\si_s$ (Theorem \ref{thmp1}). As 
\beq
\si_1 & \le \smax(\Phi^L)\xmax\smax(\Phi^{M-L}),\\
\si_s & \ge \smin(\Phi^L)\xmin\smin(\Phi^{M-L}),\label{29}
\eeq
the larger $ \smin(\Phi^L), \xmin, \smin(\Phi^{M-L})$ and the smaller $ \smax(\Phi^L), \xmax, \smax(\Phi^{M-L})$ are, the less sensitive MUSIC is to noise. It follows that a close-to-unity dynamic range $\xmax/\xmin$  and good conditioning of $\Phi^L$ and $\Phi^{M-L}$
are  conducive  to the stability of MUSIC.  

In particular, the denominator $(\si_s -\|E\|_2)^2$ on the right hand side of (\ref{eqp}) indicates that
the amount of noise that can be tolerated by MUSIC is approximately $\sigma_s$, which by (\ref{29}) is at least
$\xmin\smin(\Phi^L)\smin(\Phi^{M-L})$.
 
Hence, to understand the super-resolution effect of MUSIC 
it is essential to estimate the smallest nonzero singular value of $\Phi^L$ with closely spaced frequencies. Under certain weakened gap conditions proposed in \cite{BKL,Donoho92}, we provide an explicit upper bound on  $\smax(\Phi^L)$ and discuss 
the possible implication of the bounds on  $\smin(\Phi^L)$ 
in \cite{Donoho92} on the super-resolution effect.

\subsection{Maximum singular value of $\Phi^L$} 

Motivated by \cite{BKL,Donoho92}, we consider a sequence $\supp=\{\omega_j: j\in \ZZ\}$ which is {\em $s$-periodic} in the sense that $\supp\cap [0,1)=\{\om_1,\om_2,\dots,\om_s\}$ and  $\om_{j+ks}=k+\om_j,\forall k\in \ZZ,  j=1,\dots, s.$ 
Without loss of generality, we assume $0\leq \om_1<\om_2<\cdots<\om_s<1$.

 Define 
\begin{align*}
B(q,L) &= \left\{
\begin{array}{l l}
\frac{4\sqrt 2}{\pi} + \frac{\sqrt 2}{\pi L^2 q^2} + \frac{3\sqrt 2}{L} 
& \text{if } L \text{ is even}
\\
\\
\left(1+\frac 1 L\right)\left(\frac{4\sqrt 2}{\pi} + \frac{\sqrt 2}{\pi (L+1)^2 q^2} + \frac{3\sqrt 2}{L+1}\right)  
& \text{if } L \text{ is odd}.
\end{array}\right.
\end{align*}

\begin{theorem}
\label{thmd1}
\commentout{
 For a fixed positive integer $R$ and $\rho >0$ satisfying
$$ R\rho > \frac 1 L\sqrt{\frac{2}{\pi }}\Big(\frac{2}{\pi} - \frac 4 L \Big)^{-\frac 1 2},$$
}
Suppose $\supp $ is an $s$-periodic sequence 
and  satisfies the following weakened gap condition
\beq
\label{eqd1}
|\om_{j+R} -\om_j| >  R\rho, \ j=1,\ldots, s,
\eeq
for  some $R\in \ZZ^+$ and  $\rho \in \RR^+$.  
Then
\beq
\label{eqd2}
\frac{1}{L}\sum_{k=0}^L \left|\sum_{j=1}^s  \bc_j e^{-2\pi i k\om_j}\right|^2 \le B(R\rho,L) R\|\bc\|_2^2, \ \forall \bc \in \CC^s.
\eeq
In other words,
\beq
\label{eqd3}
\frac{1}{L}\smax^2(\Phi^L) \le   B(R\rho,L) R.
\eeq

\end{theorem}
Proof of Theorem \ref{thmd1} is in Appendix \ref{appd}.

By choosing $R$ and $\rho$ such that  $R\rho = 2/L$, in which case there are at most $R$ frequencies in any interval of 4 RL (1 RL = 1/(2L) when L = M/2), we can maintain $B(R\rho,L)$ roughly independent of $L\gg 1$ and obtain the (asymptotic) upper bound ${17\sqrt{2}\over 4\pi} R$ for $\smax^2(\Phi^L)/L$. It is noteworthy that  the minimum separation between two consecutive frequencies significantly affect the least nonzero singular value but
not  the largest one.

\subsection{Minimum nonzero singular value of $\Phi^L$}
Presently  we can not prove  an {\em explicit} lower bound for $\smin(\Phi^L)$ when two or more  frequencies are spaced below 2 RL (1 RL = 1/(2L) when L = M/2). 

Now we recall  and reformulate the lower
bound established by Donoho \cite{Donoho92}.
Let $\supp=\{\om_j\}_{j\in \ZZ}$ be a subset of the lattice $\mathcal{L}(\Delta) = \{k\Delta, k \in \ZZ\}$ of spacing $\Delta$.
Let  the Rayleigh index $R_*$ be the {\em least} positive integer such that  
\beq
\label{eqd11}
|\om_{j+R_*}-\om_j| >  {2R_*}/L,\quad \forall j. 
\eeq
In other words, $R_*$ is the size of the largest cluster whose members are separated from each other by less than $4$  RL.
Define 
\beq
\label{eqd5}
\nu(\Delta,L,R_*)=\min_{\|c\|_2=1}  \left(\int_{-L/2}^{L/2}\Big|\sum_{j \in \ZZ}\bc_j e^{-2\pi i \om_j t}\Big|^2 dt\right)^{\frac 1 2} .
\eeq
According to \cite{Donoho92} 
\beq
\nu(\Delta,L,R_*) \ge \Delta^{2R_*+1} \alpha(L,R_*) 
\label{eqd6}
\eeq
where $\alpha(L,R_*)$ is some positive constant depending on $L$ and $R_*$. No algorithm is proposed for support reconstruction in \cite{Donoho92}. 



The definition  (\ref{eqd5}) is the continuous analog of
\beq
\label{eqd4}
\smin(\Phi^L)=
\min_{\|c\|_2=1} 
\left(\sum_{k=0}^L \Big| \sum_{j=1}^s \bc_j e^{-2\pi i k\om_j}\Big|^2
\right)^{\frac 1 2}. 
\eeq
Hence by making the identification $\Delta= \min_{i\neq j} d(\om_i,\om_j)=q$, it is plausible that, even with discrete data 
and without  the lattice substrate,  
the bound
\beq
\smin(\Phi^L) \ge q^{2R_*+1} \tilde\alpha(L, R_*) 
\label{eqd60}
\eeq
holds  
with some positive constant $\tilde{\alpha}(L,R_*)$
under the condition (\ref{eqd11}).

Our preceding analysis in Sections \ref{seccon} and \ref{secper} is for  the case $R_*=1$. 




As commented above, for a support set with Rayleigh index $R_*$, the amount of noise that can be tolerated by MUSIC is approximately $\sigma_s$ which, according to \eqref{eqd60},   decays at worst like $\xmin q^{4R_*+2}$ as $q\rightarrow 0$.
Our numerical experiments in Section \ref{secsupnum}
show that the noise level that MUSIC can tolerate obeys  $q^{e(R_*)}$ for $R_* = 2,3,4,5$ and
$e(2) = 3.6691$, $e(3) = 6.0565$, $e(4) = 8.3861$ and $e(5) = 11.2392$, suggesting that $e(R_*)\approx 2.504 R_*-1.4262$.

 


\commentout{
\textcolor{red}{----------Before Aug 28----------}

\begin{theorem}[{\cite[Theorem 1.3 and 1.4]{Donoho92}}]
\label{thmd2}
Let $\supp=\{\om_j: j\in \ZZ\}$ be any \textcolor{red}{$s$-periodic} subset of the lattice $\mathcal{L}(\Delta) = \{k\Delta, k \in \ZZ\}$ of spacing $\Delta$.
Let  the Rayleigh index $R_*$ be the least positive integer such that  
\beq
\label{eqd11}
|\om_{j+R_*}-\om_j| >  R_*/M,\quad \forall j. 
\eeq
Define 
\beq
\label{eqd5}
\nu(\Delta,T,R_*)=\min_{\|c\|_2=1}  \int_{-T}^T \Big|\sum_{j=1}^s\bc_j e^{-2\pi i \om_j t}\Big|^2 dt .
\eeq
Then the following statements hold.
\begin{enumerate}
\item 
For  $T > 1$, we have 
\beq
\nu(\Delta,T,R_*) \ge \Delta^{2R_*+1} \alpha(T,R_*) 
\label{eqd6}
\eeq
where $\alpha(T,R_*)$ is some positive constant depending on $T$ and $R_*$.

\item Fix any number $\Delta_0\in (0,1).$  We have 
\beq
\nu(\Delta,T,R_*) \le \Delta^{2R_*-1} \beta(T,R_*,\Delta_0),\quad
\forall \Delta<\Delta_0
\label{eqd7}
\eeq
where $\beta(T,R_*,\Delta_0)$ is some positive constant depending on $T$, $R_*$ and $\Delta_0$.

\end{enumerate}
\end{theorem}

\begin{remark} The definition  (\ref{eqd5}) is the continuous analog of
\beq
\label{eqd4}
\smin^2(\Phi^L)=
\min_{\|c\|_2=1} \sum_{k=0}^L \Big| \sum_{j=1}^s \bc_j e^{-2\pi i k\om_j}\Big|^2. 
\eeq
Hence by making the identifications  $L=2T$ and $\Delta= \min_{i\neq j} d(\om_i,\om_j)=q$, it is plausible that, even with discrete data 
and without  the lattice substrate,  
the bounds 
\[
\alpha(L, R_*) q^{2R_*+1} \leq \smin^2(\Phi^L)\leq \beta(L,R_*,\Delta_0) q^{2R_*-1}
\]
hold  for any positive $q<\Delta_0<1$ under the condition (\ref{eqd11}) 
that the largest cluster with consecutive frequencies separated by at most 1 RL has exactly $R_*$ elements. 

Our preceding analysis in Sections \ref{seccon} and \ref{secper} deals with the case $R_*=1$. 

 Donoho \cite{Donoho92} conjectured
 that $\nu(\Delta,T,R_*)\sim \Delta^{p(R_*)}$ for certain power $p(R_*)\in [2R_*-1, 2R_*+1]$ depending on $R_*$ and by analogy $\smin^2(\Phi^L)\sim q^{p(R_*)}$

 \label{rmk12}
\end{remark}

As commented above, 
the amount of noise that can be tolerated by MUSIC is approximately $\sigma_s$ which, according to Remark \ref{rmk12},   is roughly proportional to $\xmin q^{p(R_*)}$
with $p(R_*)\in [2R_*-1, 2R_*+1]$.

No algorithm is proposed for support reconstruction in \cite{Donoho92}. On the other hand, our numerical experiments
show that MUSIC can  tolerate  the noise level with the optimal power  $p(R_*)\approx 2R_*-1$, for 
$R_* = 2, 3, 4, 5$ (Figure \ref{figPT}).

 

}

\section{Numerical experiments}
\label{secnum}

A systematic numerical simulation is performed on MUSIC, BLOOMP, SDP and Matched Filtering using prolates in this section, showing that MUSIC combines the advantages of strong stability and low computation complexity for the detection of well-separated frequencies and furthermore only MUSIC yields an exact reconstruction in the noise-free case regardless of the distribution of true frequencies and processes  the capability of resolving closely spaced frequencies.

\subsection{Algorithms tested.}
We compare the performances of various algorithms on the spectral estimation problem 
(\ref{model})
with $M = 100$ and i.i.d. Gaussian noise, i.e. $\ep \sim N(0,\sigma^2 I) + i N(0, \sigma^2 I)$. 
Define the  
\beq
{\hbox{Noise-to-Signal Ratio (NSR)} } = \mathbb{E}(\|\ep\|_2)/\|y\|_2 = \sigma\sqrt{2(M+1)}/\|y\|_2.
\label{eqnsr}
\eeq

We test and compare the following algorithms.
\begin{enumerate}

\item The MUSIC algorithm: As suggested by \eqref{ce1} in Corollary \ref{cor2}  we set $M=2L$. 

\commentout{
\begin{center}
   \begin{tabular}{|l|}\hline
    { {\bf MUSIC algorithm}} \\ \hline
    {\bf Input:} $y^\ep \in \CC^{M+1}, s, L$. \\
     1) Form matrix $H^\ep = {\rm Hankel}(y^\ep) \in \CC^{(L+1)\times(M-L+2)}$.
     \\
     2) SVD: $H^\ep = [U_1^\ep\  U_2^\ep] {\rm diag}(\si_1^\ep , \ldots , \si_s^\ep ,\ldots) [V_1^\ep\ V_2^\ep]^\star $, where $U_1^\ep \in \CC^{(L+1)\times s}$.\\
     3) Compute imaging function $J^\ep(\om) = \|\phi^{L}(\om)\|_2 /\|{U_2^\ep}^\star \phi^L(\om)\|_2$. \\
   {\bf Output:} $S =\{ s \text{ largest local maximizers of } J^\ep(\om) \} $.\\
    \hline
   \end{tabular}
\end{center}
}

\item Band-excluded and Locally Optimized Orthogonal Matching Pursuit (BLOOMP) \cite{FL}: BLOOMP works with an arbitrarily fine grid with grid spacing $\ell = {\rm RL}/F$ where $F$ is the refinement/super-resolution factor. In the presence of noise it is unnecessary to set an extremely large $F$.  A rule of thumb for the problem of spectral estimation is that $F\sim \hbox{\rm SNR}$ gives rise to  a gridding error comparable to the external noise \cite[Fig. 1]{FL}. For instance  $F = 20$ is adequate when SNR $\geq 5\%$. When frequencies are separated above $3$ RL (i.e., in Fig. \ref{figc3}(b), Fig. \ref{figc1}(b) and Fig. \ref{figc2}(a)(b)), we can set the radii of excluded band and local optimization to be $2$ RL and $1$ RL,  respectively. When frequencies are separated between 2 RL and 3 RL (i.e., in Fig. \ref{figc2}(c)(d), the radii of excluded band and local optimization is set to be 1 RL.

\item SemiDefinite Programming (SDP) \cite{Csr2,Tang}: The code is  from \url{http://www.stanford.edu/~cfgranda/superres_sdp_noisy.m} where SDP is solved through CVX, a package for specifying and solving convex programs \cite{cvx1,cvx2}. Output of SDP is the dual solution of total variation minimization. In the code, frequencies are identified through root findings of a polynomial and amplitudes are solved through least squares. Let $\tilde {\om}=[\tilde\om_j]_{j=1}^n \in \RR^n$ be the frequencies retrieved from the code and $ \tilde x =[\tilde x_j]_{j=1}^n \in \CC^n$ be the corresponding amplitudes.
Usually $n$ is greater than $s$. A straightforward way of extracting $s$ reconstructed frequencies is by Hard Thresholding (HT), i.e., picking the frequencies corresponding to the $s$ largest amplitudes in $\tilde x$.
We also test if the Band Excluded Thresholding (BET) technique introduced in \cite{FL}  can improve on hard thresholding and 
enhance the performance of SDP (Fig. \ref{figc2}).
BET amounts to trimming $\tilde\om\in\RR^n$ to $\hat\om\in \RR^s$ 
as follows. 



\begin{center}
   \begin{tabular}{|l|}\hline
    { {\bf Band Excluded Thresholding (BET)}} \\ \hline
    {\bf Input:} $\tilde\om, \tilde x, s , r$ (radius of excluded band). 
    \\
    {\bf Initialization:} $\hat \om= [\ ]$.
    \\
    {\bf Iteration:}     for $k = 1, \ldots, s$
    \\
     1) Find $j$ such that $|{\tilde x}_{j}| = \max_i |{\tilde x}_i|$.
     \\  \quad \,If ${\tilde x}_{j}=0$, then go to {\bf Output}.\\
     2) Update the support vector: $\hat \om = [\hat\om\ ; \ \tilde\om_j]$.\\
     3) For $i = 1: n$
     \\
     \qquad  If $\tilde \om_i \in (\tilde\om_{j}- r,\tilde\om_{j} + r)$, set $\tilde x_i = 0$.	
     \\
   {\bf Output:} $\hat\om$.\\
    \hline
   \end{tabular}
\end{center}

\commentout{
\begin{center}
   \begin{tabular}{|l|}\hline
    { {\bf Band Excluded Thresholding (BET)}} \\ \hline
    {\bf Input:} $\hat\om, \hat x, s , r$ (radius of excluded band). 
    \\
    {\bf Initialization:} $\tilde \om= \mathbf{0}\in \RR^s$, $\tilde x=\mathbf{0}\in \CC^s$.
    \\
    {\bf Iteration:}     for $k = 1, \ldots, s$
    \\
     1) Find $j$ such that $|{\hat x}_{j}| = \max_i |{\hat x}_i|$.
     \\  \quad \,If ${\hat x}_{\hat i}=0$, then go to {\bf Output}.\\
     2) Update the support vector: $\tilde\om \hookleftarrow \tilde\om+\hat \om_{j} {\rm e}_k$.\\
3) Update the estimate: $\tilde x_{k}=\hat x_{j}$.\\
     4) For $i = 1: n$
     \\
     \qquad  If $\hat \om_i \in (\hat\alpha_{j}- r,\hat\alpha_{j} + r)$, set $\hat x_i = 0$.
     \\
   {\bf Output:} $\tilde\om$, $\tilde x$.\\
    \hline
   \end{tabular}
\end{center}
}
When frequencies are separated by at least 4 RL, we choose $r = $ 1 RL.

\item Matched Filtering (MF) using prolates: In \cite{Armin}, Eftekhari and Wakin use matched filtering windowed by the Discrete Prolate Spheroidal (Slepian) Sequence \cite{DPSS} for the same problem while frequencies are extracted by band-excluded and locally optimized thresholding proposed in \cite{FL}. In its current form,  MF using prolates can not deal with complex-valued amplitudes  so it is tested with real-valued amplitudes only. 
\end{enumerate}

Reconstruction error is measured by Hausdorff distance between the exact ($\supp$) and the recovered ($\hat \supp$) sets of frequencies:
\beq
\label{hausdorff}
d(\hat \supp,\supp) = \max\left\{ \max_{\hat\om\in\hat\supp }\min_{\om\in\supp}d(\hat\om,\om)\ ,\ \max_{\om\in\mathcal{S}}\min_{\hat\om\in\hat\supp }d(\hat\om,\om)\right\}.
\eeq

\subsection{Noise-free case} In the noise-free case only MUSIC processes a theory of exact reconstruction regardless of the distribution of true frequencies. In theory, BLOOMP requires a separation of 3 RL for approximate support recovery while SDP requires a separation of 4 RL for exact recovery. In this test we use the four algorithms to recover  $15$ real-valued amplitudes separated by $1$ RL. Figure \ref{figc3} shows that MUSIC achieves the accuracy of about $0.004$ RL while BLOOMP, SDP and MF using prolates essentially fail, which implies that certain separation condition is necessary for BLOOMP, SDP and MF using prolates.
\begin{figure}[hthp]
        \centering
       \subfigure[MUSIC. Red: exact; Blue: recovered. $d(\hat\supp ,\supp)\approx 0.004$ RL.]{\includegraphics[width=7cm]{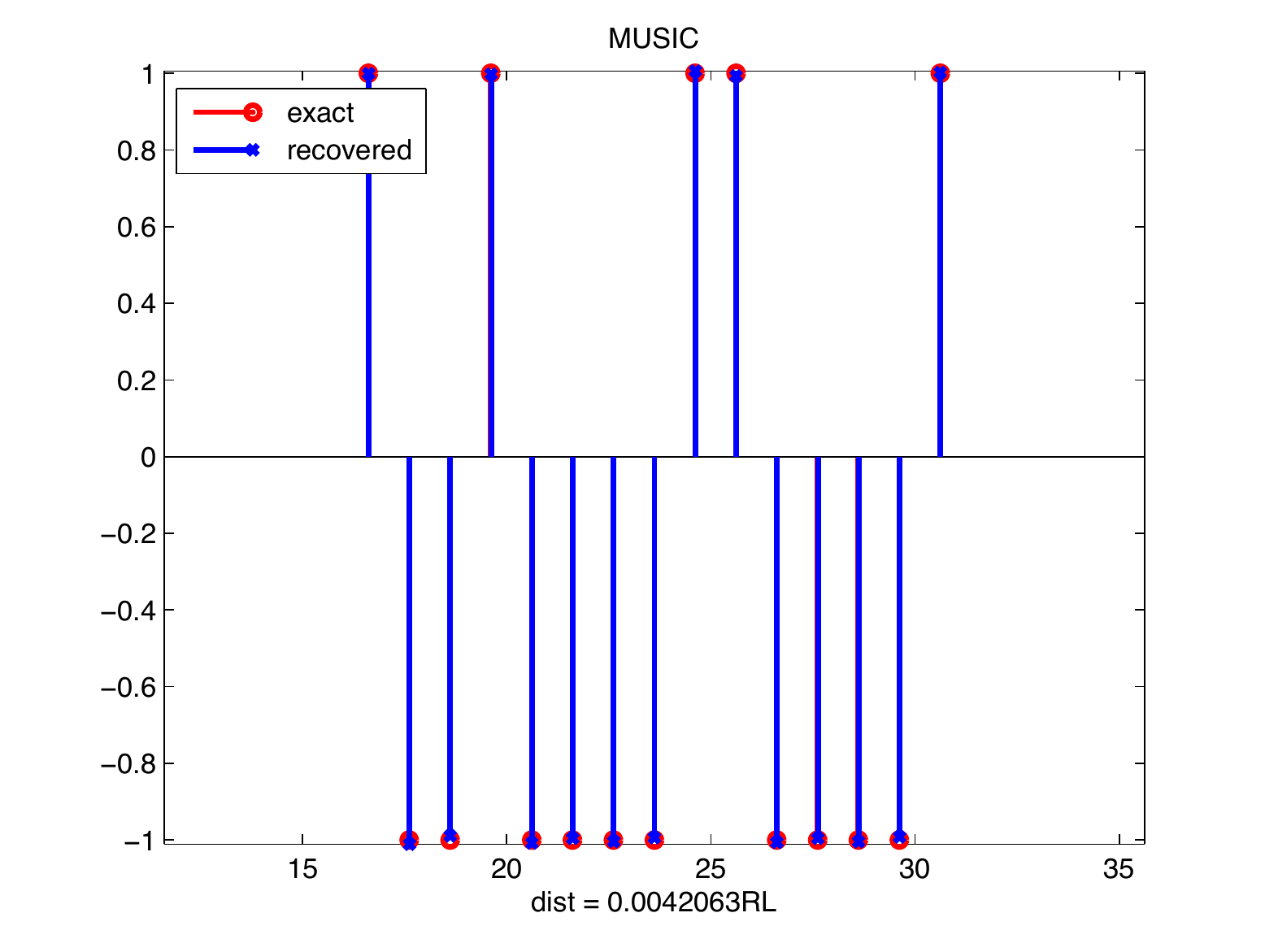}}
       \qquad
        \subfigure[BLOOMP. Red: exact; Blue: recovered. $d(\hat\supp ,\supp)\approx 1.81$ RL.]{\includegraphics[width=6.9cm,height=5.5cm]{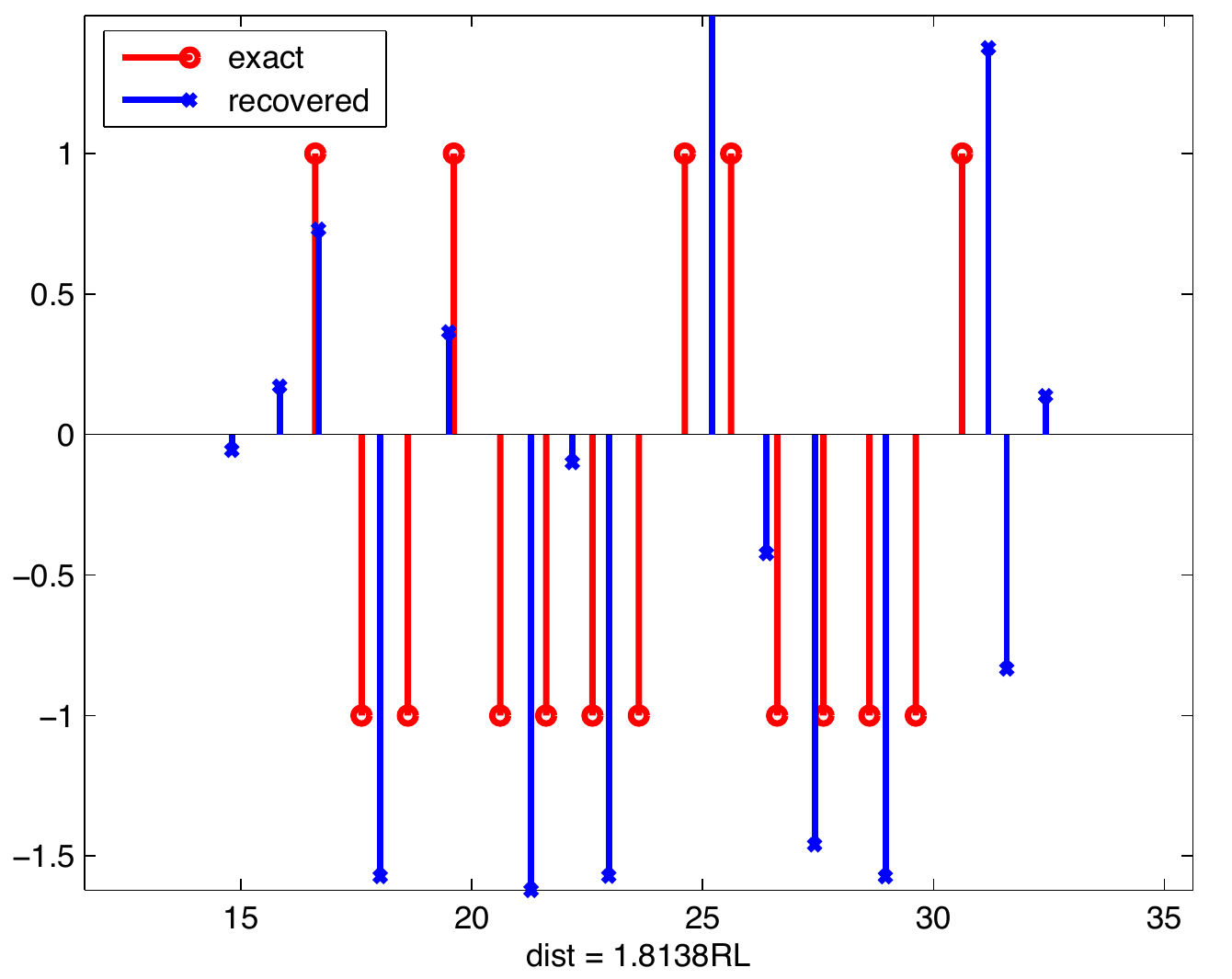}}
         \subfigure[SDP. Red: exact; Blue: recovered.  Hard thresholding yields $d(\hat\supp ,\supp)\approx 2.72$ RL.]{\includegraphics[width=7cm]{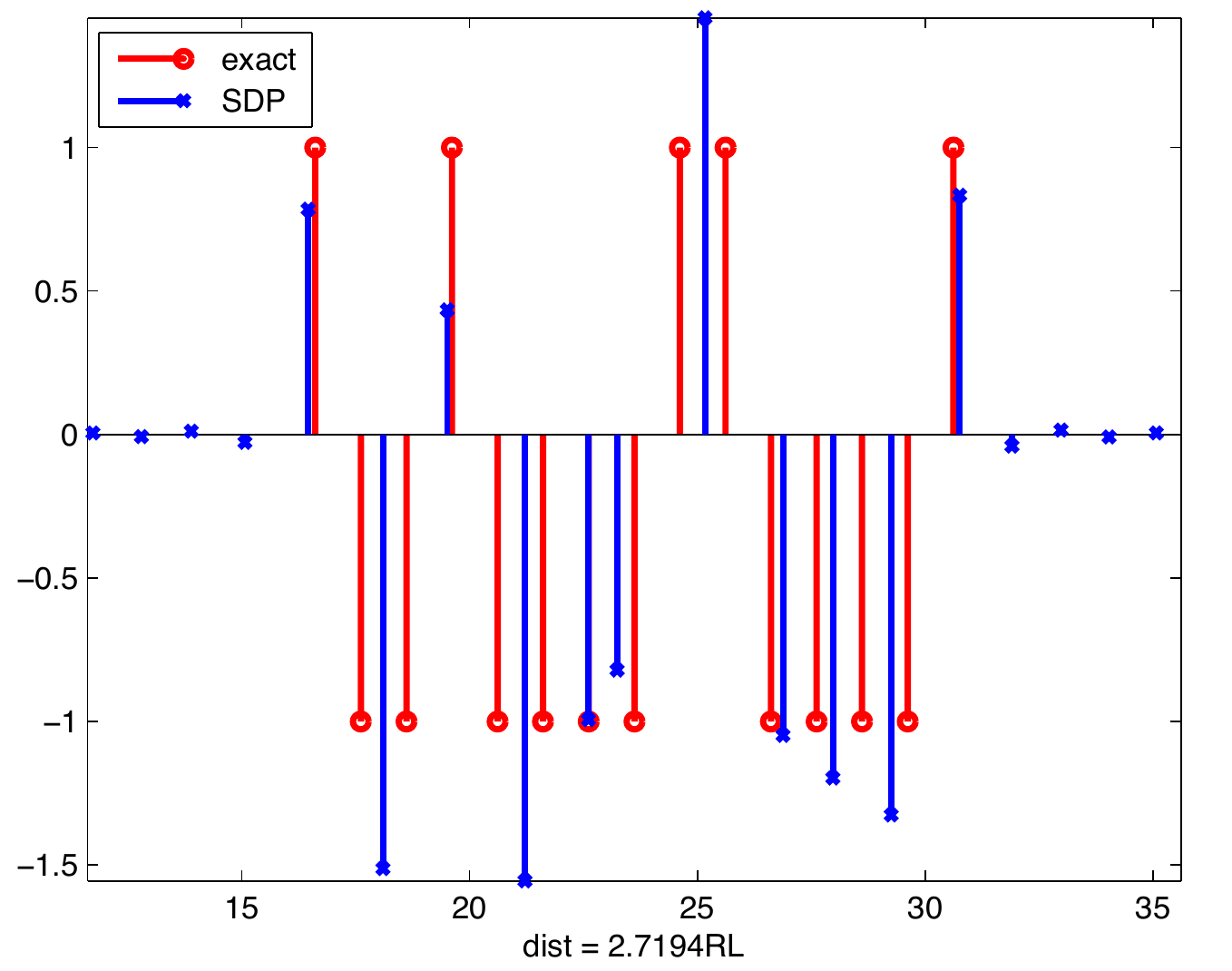}}
         \quad
         \subfigure[MF using prolates. Red: exact; Blue: inverse Fourier transform of $y$ windowed by the first DPSS sequence.]{\includegraphics[width=7cm]{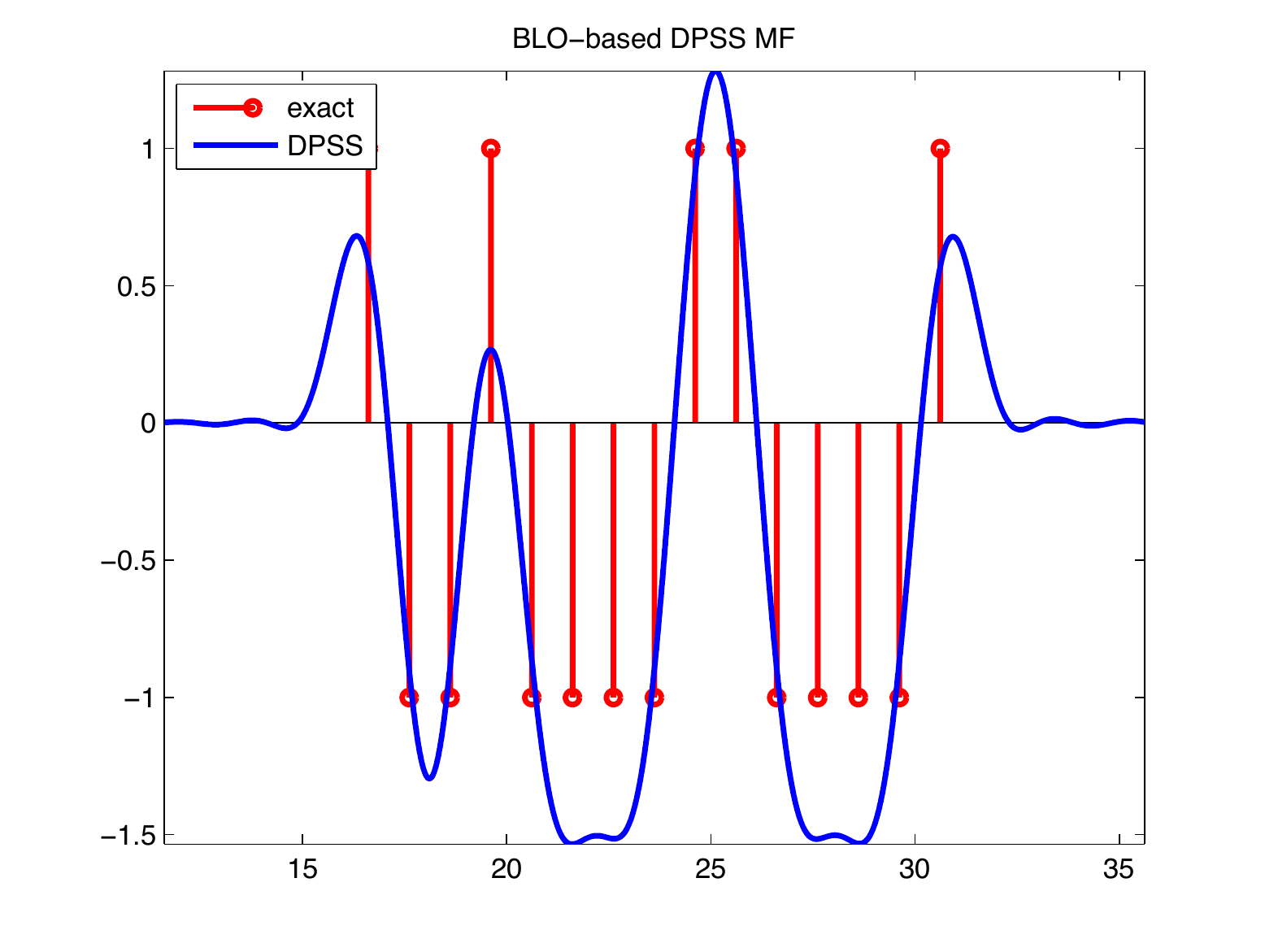}}
           \caption{Reconstruction of 15 real-valued frequencies separated by 1 RL. Dynamic range $= 1$ and NSR $ = 0\%$.           }
              \label{figc3}
\end{figure}

\subsection{Detection of well-separated frequencies}
\begin{figure}[hthp]
        \centering
       \subfigure[MUSIC. Red: exact; Blue: recovered. $d(\hat\supp ,\supp)\approx 0.06$ RL.]{\includegraphics[width=7cm]{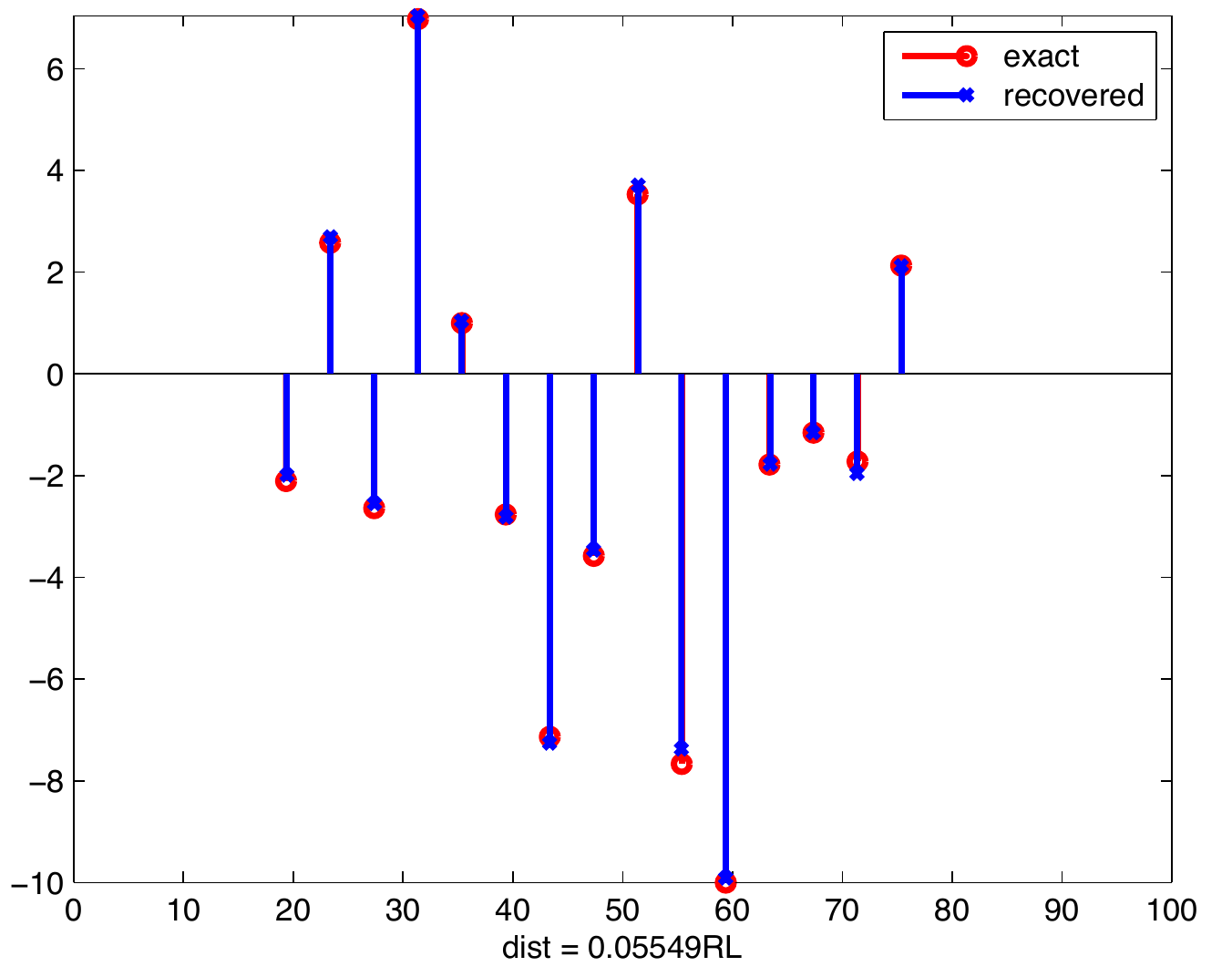}}
       \qquad
        \subfigure[BLOOMP. Red: exact; Blue: recovered. $d(\hat\supp ,\supp)\approx 0.05$ RL.]{\includegraphics[width=7cm]{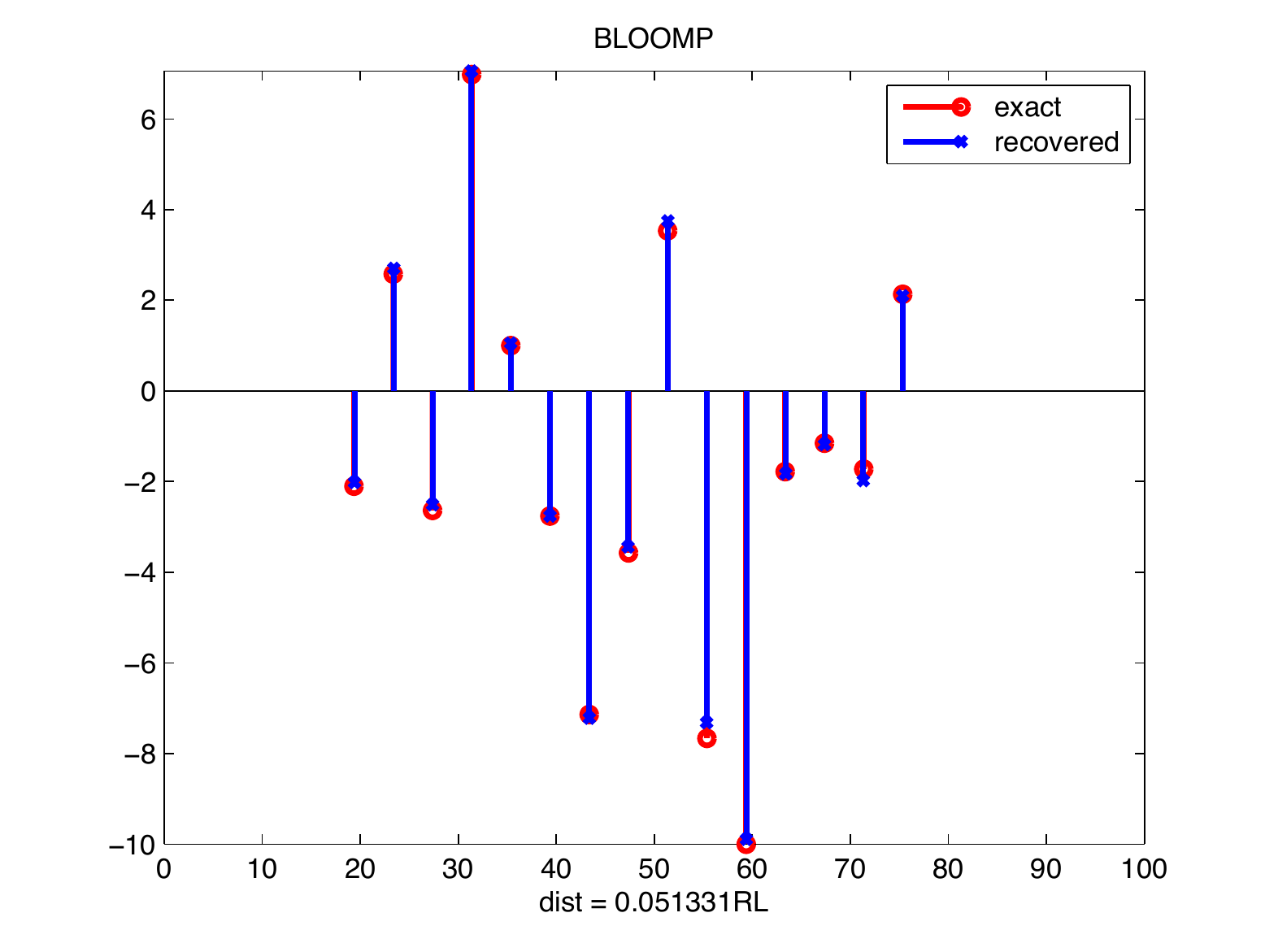}}
         \subfigure[SDP. Red: exact; Blue: Primal solution of SDP. Hard thresholding (green) yields $d(\hat S,\supp) \approx 3.94$ RL. The true amplitude  around $33$ RL is recovered as two amplitudes and the BET technique  can be used to eliminate the smaller one in the step of frequency selection. ]{\includegraphics[width=7.1cm,height=5.5cm]{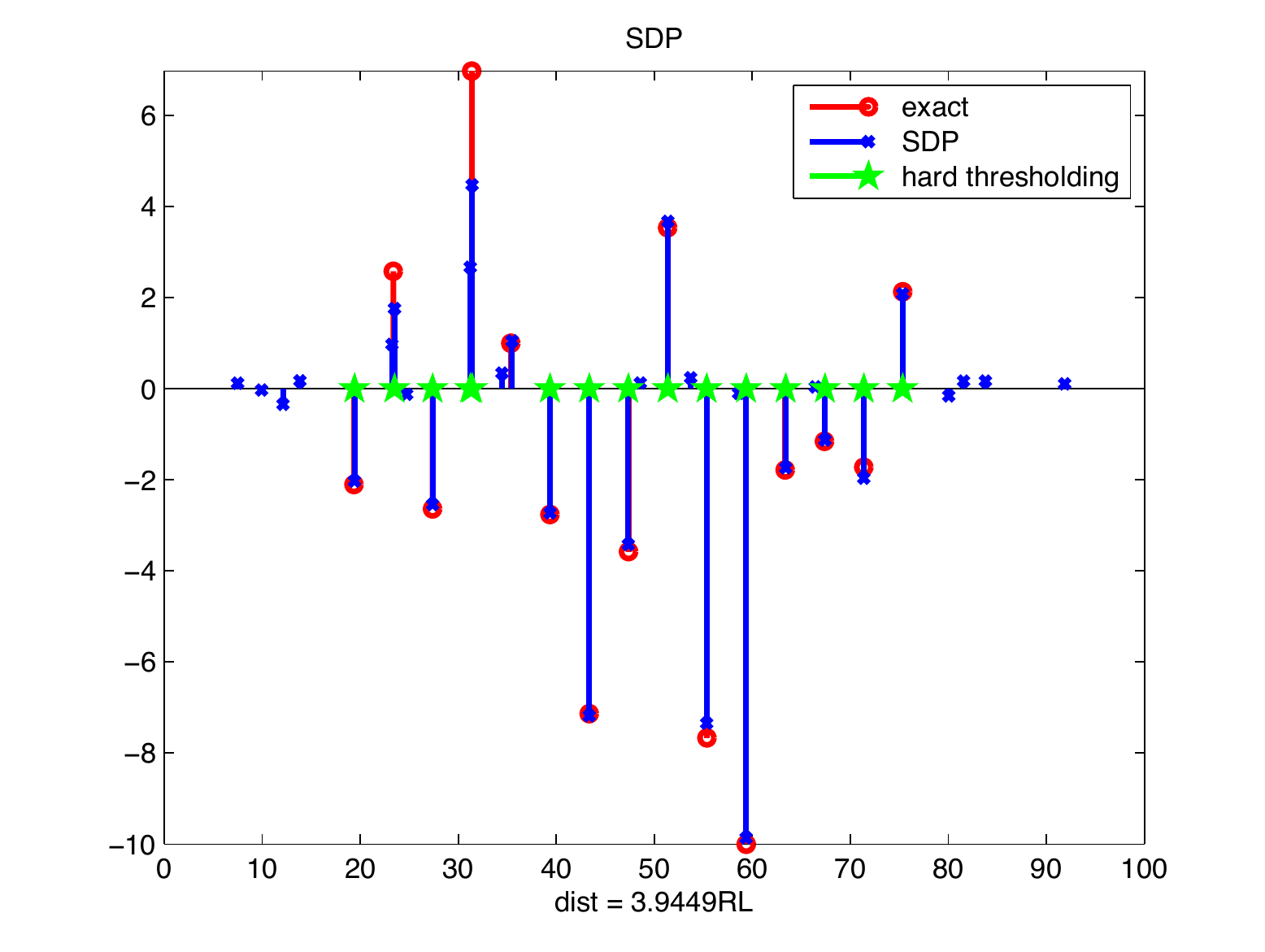}}
         \qquad
          \subfigure[MF using prolates. Red: exact; Blue: inverse Fourier transform of $y^\ep$ windowed by the first DPSS sequence; Green: frequencies selected by the BLO technique. $d(\hat\supp ,\supp)\approx 0.10$ RL.]{\includegraphics[width=7cm]{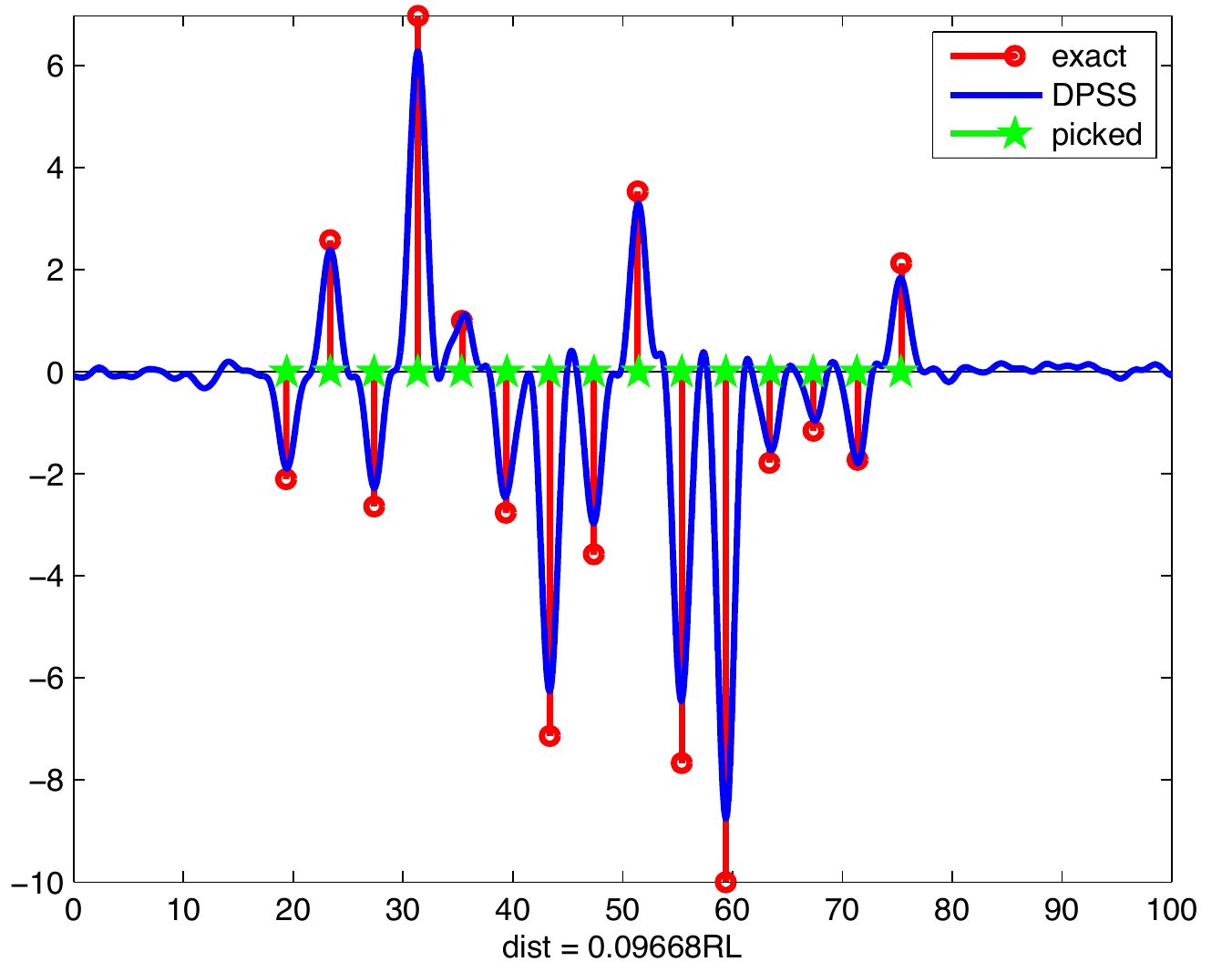}}      
          \caption{Reconstruction of $15$ real-valued amplitudes separated by $4$ RL. Dynamic range $= 10$ and NSR $ = 10\%$.
          }
              \label{figc1}
\end{figure}

Figure \ref{figc1} shows reconstructions of $15$ real-valued frequencies separated by $4$ RL. By extracting $15$ largest local maxima of the imaging function $J^\ep(\om)$, MUSIC yields a reconstruction distance of $0.06$ RL. As predicted by the theory in \cite{FL}, every recovered object of BLOOMP is within 1 RL distance from a true one.  Indeed, in this simulation BLOOMP achieves the best accuracy of $0.05$ RL among tested algorithms. The primal solution of SDP is  usually not $s$-sparse and the recovered frequencies tend to cluster around the true ones \cite{Csr2}
which degrades the accuracy.  The Hausdorff distance between the recovered spikes with the $s$ strongest amplitudes  and the true frequencies is 3.94 RL in this simulation. The  BET  technique can be used to enhance the accuracy of reconstruction and
achieve the accuracy of 0.13 RL.  Similarly  the BLO  technique introduced in \cite{FL}  can
be  applied to improve  the result of Matched filtering windowed by the DPSS sequence (the blue curve in Figure \ref{figc1}(d)) and achieve the accuracy of $0.10$ RL.

Figure \ref{figc2} shows the average errors of 100 trials by SDP with HT, BET-enhanced SDP, BLOOMP and MUSIC for complex-valued objects separated between 4 RL and 5 RL (Fig. \ref{figc2}(a)(b)) or separated between 2 RL and 3 RL (Fig. \ref{figc2}(c)(d)) versus NSR when dynamic range = 1 (Fig. \ref{figc2}(a)(c)) and when dynamic range = 10 (Fig. \ref{figc2}(b)(d)). In this simulation $[0,1)$ are fully occupied by frequencies satisfying the separation condition and amplitudes $x$ are complex-valued with random phases. Refinement factor $F$ in BLOOMP is adaptive according to the rule: $F = \max(5,\min(1/{\rm NSR},20))$.
Figure \ref{figc2} shows that BLOOMP is the stablest algorithm while frequencies are separated above 4 RL and MUSIC becomes the stablest one while frequencies are separated between 2 RL and 3 RL. Simply extracting $s$ largest amplitudes from the SDP solution (black curve) is not a good idea and the BET technique (green curve) can mitigate the problem with SDP. The average running time in Figure \ref{figc2} shows that MUSIC takes about 0.33s for one experiment and is the most efficient one among all methods being tested. SDP needs about 20.5s for one experiment and is computationally most expensive. Running time of BLOOMP is dependent on sparsity $s$ and refinement factor $F$. The running time of BLOOMP in Fig. \ref{figc2}(c)(d) is more than the time in Fig. \ref{figc2}(a)(b) as $s \in [33,50]$ in Fig. \ref{figc2}(c)(d) and $s \in [20,25]$ in Fig. \ref{figc2}(a)(b).

\commentout{
As an example in Fig. \ref{figc2}(d), we display the reconstruction of frequencies separated between 2 RL and 3 RL by MUSIC, BLOOMP, SDP and BET-enhanced SDP when dynamic range $= 10$ and NSR = $12\%$ in Figure \ref{figc20}. The reason why the SDP curve in Fig. \ref{figc2} is bad is demonstrated in Fig. \ref{figc20}. The primal solution of SDP usually contains more than $s$ spikes, and a frequency with large amplitude can split into two or more spikes around it. If hard thresholding is applied for frequency extraction, it is usually the case that two or more candidates are selected around that frequency with large amplitude and some others with small amplitudes are missed. 
}

\begin{figure}[hthp]
        \centering
          Frequencies separated between 4 RL and 5 RL\\
               \subfigure[Dynamic range = 1. Average running time for SDP and MUSIC in one experiment is 20.3583s and 0.3627s while the average running time for BLOOMP is 6.3420s ($F=20$), 3.2788s ($F=10$) and 1.7610($F = 5$).]{\includegraphics[width=7cm]{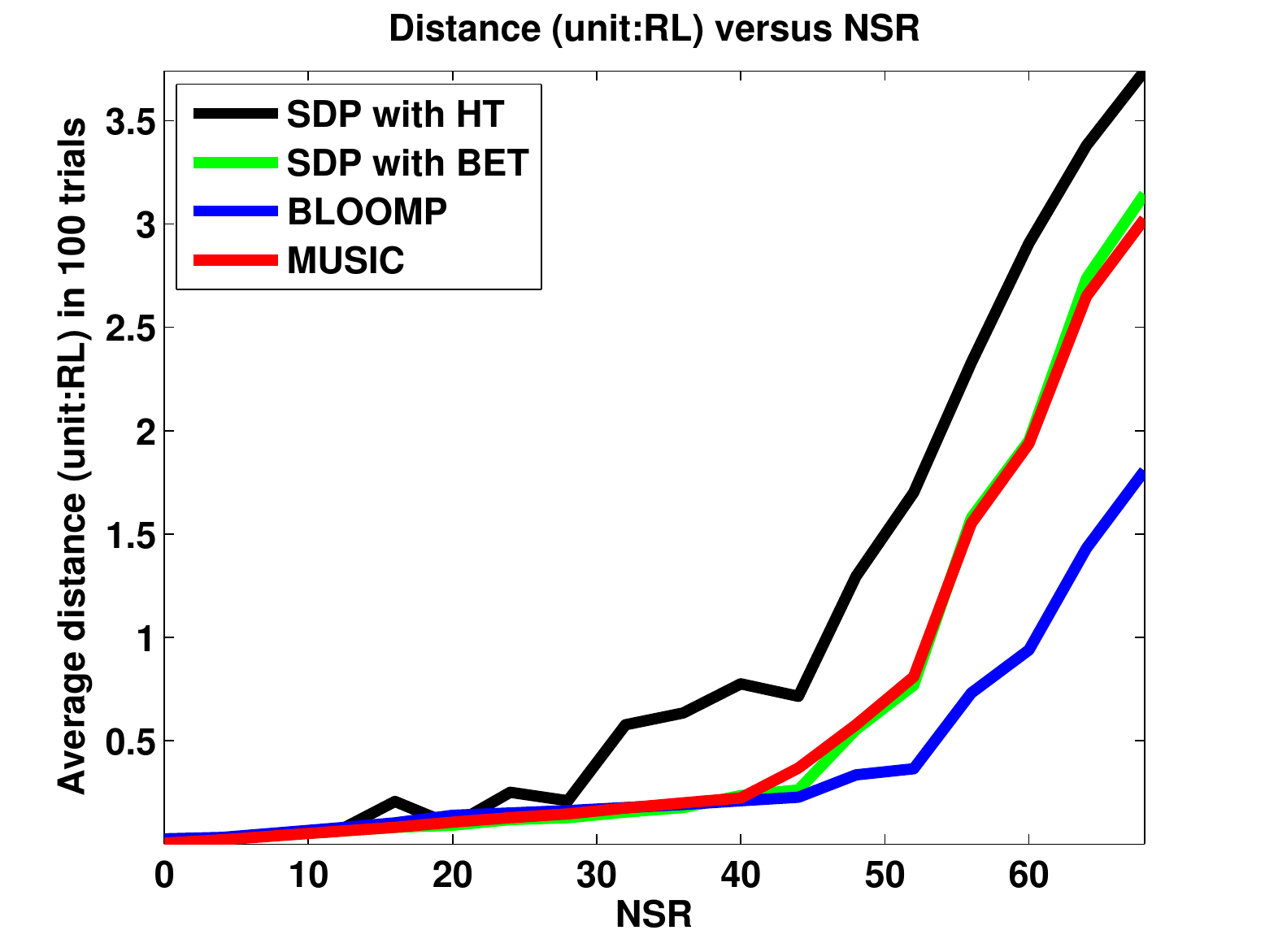}}
               \qquad
        \subfigure[Dynamic range = 10. Average running time for SDP and MUSIC in one experiment is 20.5913s and 0.3661s while the average running time for BLOOMP is 6.2623s ($F=20$), 3.3030s ($F=10$) and 1.7542s ($F=5$).]{\includegraphics[width=7cm]{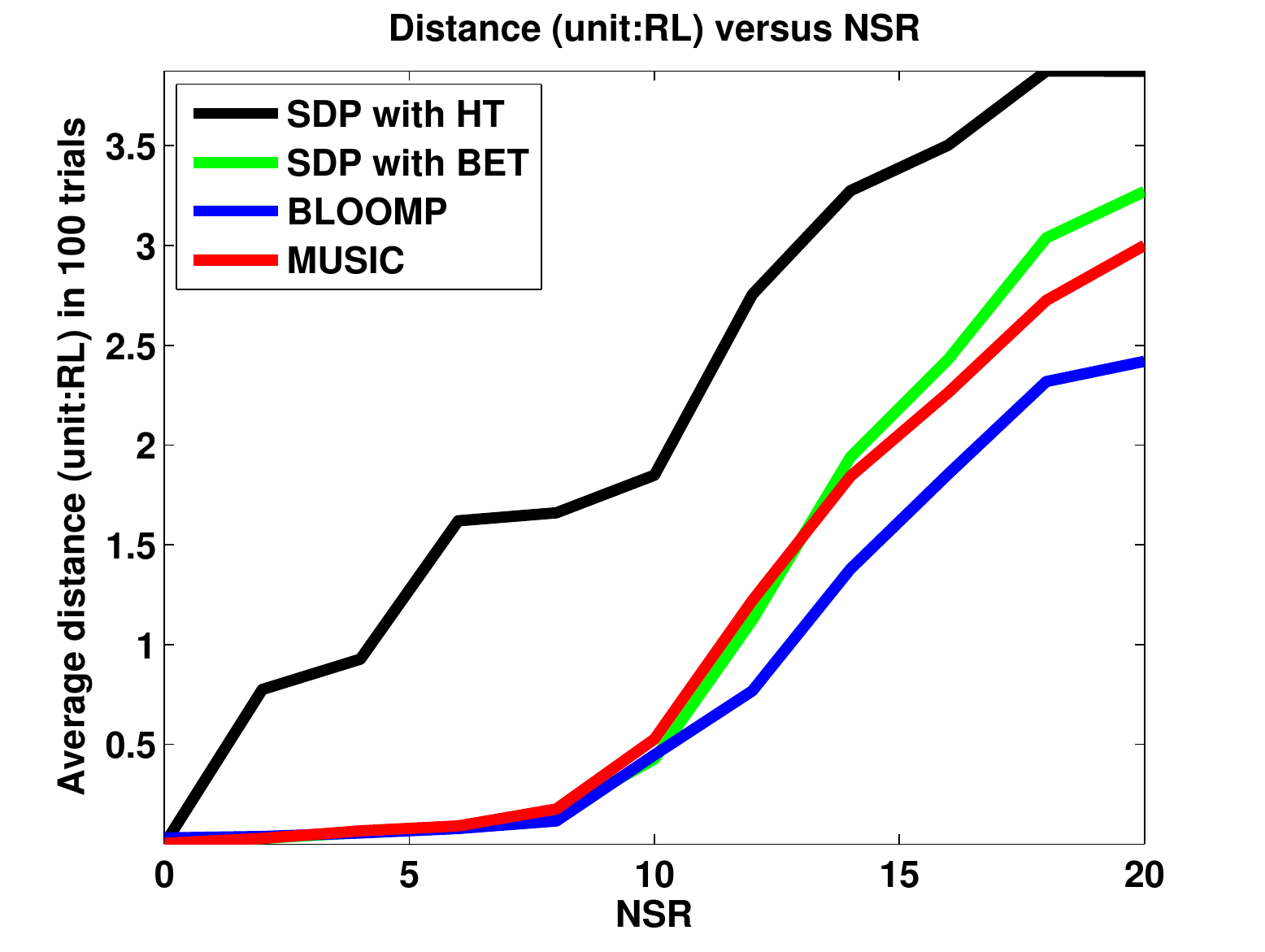}}
        \ \\
          Frequencies separated between 2 RL and 3 RL \\
               \subfigure[Dynamic range = 1. Average running time for SDP and MUSIC in one experiment is 20.6750s and 0.3334s while the average running time for BLOOMP is 19.8357s ($F=20$), 11.2349s ($F=9$) and 8.0850s ($F=5$).]{\includegraphics[width=7cm]{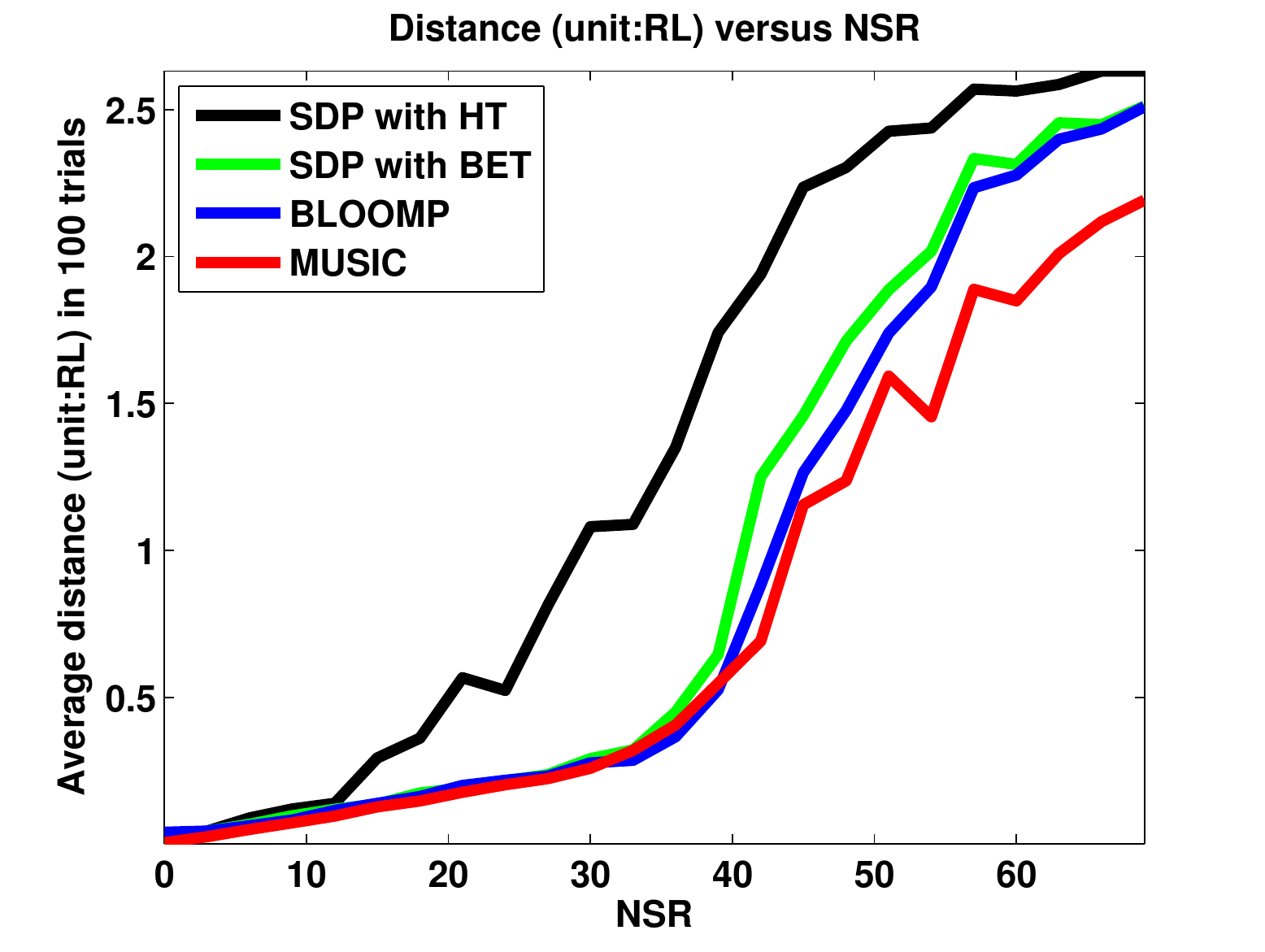}}
               \qquad
        \subfigure[Dynamic range = 10. Average running time for SDP and MUSIC in one experiment is 21.0572s and 0.3321s while the average running time for BLOOMP is 19.9233s ($F=20$), 11.5190s ($F=9$) and 8.1054s ($F=5$).]{\includegraphics[width=7cm]{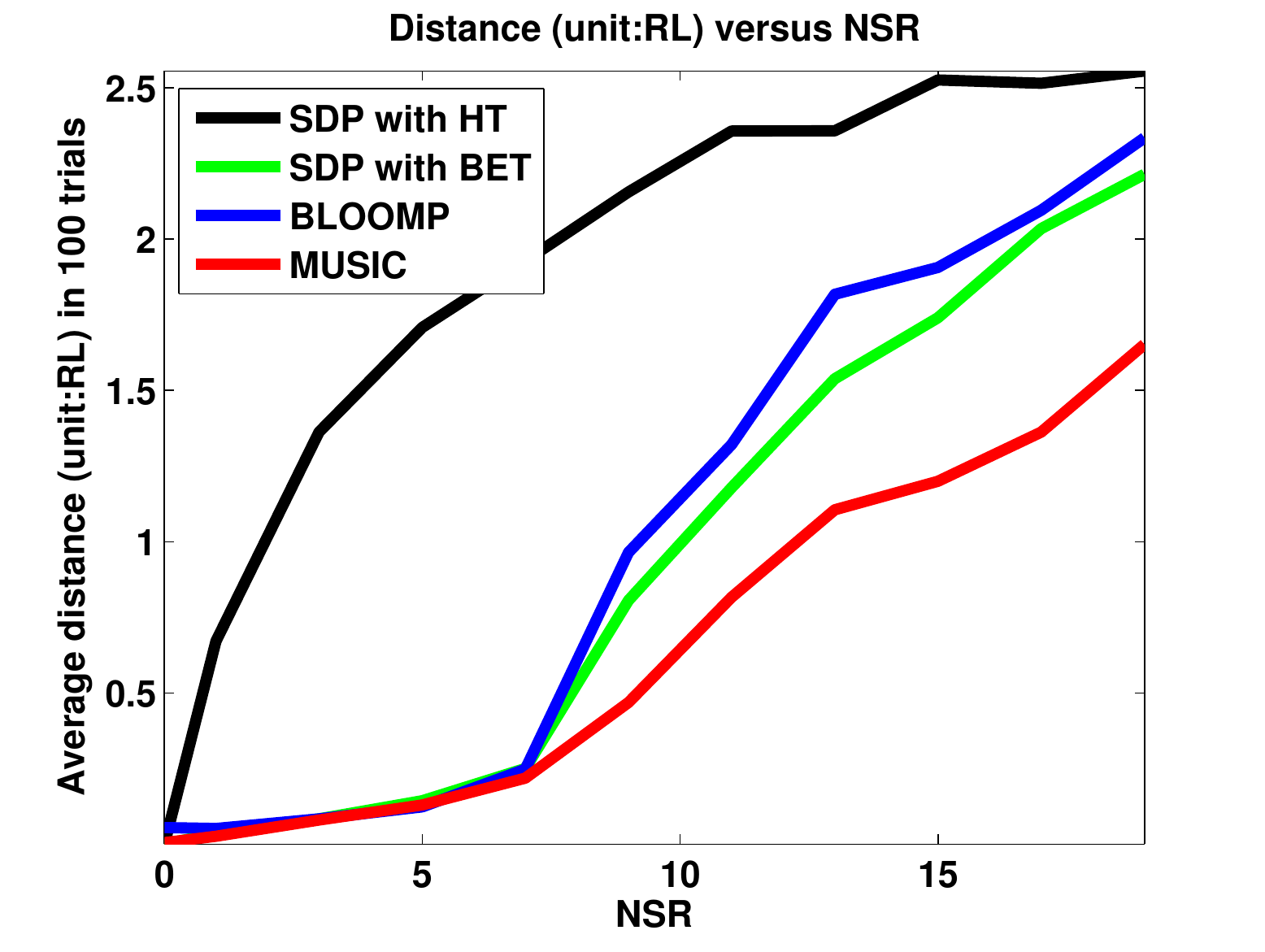}}
                 \caption{Average error by SDP with HT, BET-enhanced SDP, BLOOMP and MUSIC on complex-valued objects separated between 4 RL and 5 RL (a)(b) or separated between 2 RL and 3 RL (c)(d) versus NSR when dynamic range = 1 (a)(c) and when dynamic range = 10 (b)(d). MF using prolates is not included since it is not designed for complex amplitudes. 
          }
              \label{figc2}
\end{figure}

\commentout{
\begin{figure}[hthp]
        \centering
       \subfigure[MUSIC. Red: exact; Blue: recovered. $d(\hat\supp ,\supp)\approx 0.24$ RL.]{\includegraphics[width=7cm]{FigCompare/Sep2to3Range10Noise12/MUSIC.pdf}}
       \qquad
        \subfigure[BLOOMP. Red: exact; Blue: recovered. $d(\hat\supp ,\supp)\approx 0.23$ RL.]{\includegraphics[width=7cm]{FigCompare/Sep2to3Range10Noise12/BLOOMP.pdf}}
         \subfigure[SDP. Red: exact; Blue: Primal solution of SDP. Hard thresholding yields $d(\hat S,\supp) \approx 2.88$ RL. ]{\includegraphics[width=7cm]{FigCompare/Sep2to3Range10Noise12/SDP.pdf}}
         \qquad
          \subfigure[BET-enhanced SDP. Red: exact; Blue: recovered. $d(\hat\supp ,\supp)\approx 0.30$ RL.]{\includegraphics[width=7cm]{FigCompare/Sep2to3Range10Noise12/SDPBET.pdf}}      
          \caption{Reconstruction of frequencies separated between 2 RL and 3 RL when dynamic range $= 10$ and NSR $ = 12\%$. Amplitudes $x$ are complex-valued and their magnitudes are displayed here.
          }
              \label{figc20}
\end{figure}
}

\subsection{Super-resolution of MUSIC}
\label{secsupnum}

Theory in Section \ref{secsup} implies that MUSIC has super-resolution effect and moreover the noise level that MUSIC can handle follows a power law with respect to the minimum separation of the frequencies.
We numerically investigate the $2,3,4,5$-point resolution of MUSIC  here as numerical verification.


In Figure \ref{figPT}, we consider support set $\supp$ containing two, three, four and five equally spaced  frequencies.
We run MUSIC algorithm on reconstructions of randomly phased complex objects supported on $\supp$ with varied separation $q$ and varied NSR for 100 trials and record the average of $d(\supp,\hat\supp)/q$. Figure \ref{figPT} (a)-(d) displays the color plot of the logarithm to the base 2 of average $d(\supp,\hat\supp)/q$ with respect to NSR (y-axis) and $q$ (x-axis) in the unit of RL. Frequency localization is considered successful if $d(\supp,\hat\supp)/q <1/2$.

A  phase transition occurs in (a)-(d), manifesting  MUSIC's capability of resolving two, three, four and five closely spaced complex-valued objects if NSR is below certain level. 
Theory in Section \ref{secsup} indicates the noise level that MUSIC can handle scales at worst like $q^{4R_*+2}$ where $R_* =2$ in Figure \ref{figPT}(a), $R_* = 3$ in Figure \ref{figPT}(b), $R_* = 4$ in Figure \ref{figPT}(c) and $R_* = 5$ in Figure  \ref{figPT}(d).  
The borderline between successful recovery and failure
defined by $d(\supp,\hat \supp)/q=1/2$ are marked out in black in Figure \ref{figPT} (a)-(d).
The phase transition curves for $R_*=2,3,4,5$ are shown in (e) in the ordinary scale  and in (f) in log-log scale. 
It appears that the transition curves can be fitted to a constant times $q^{e(R_*)}$ with  $e(2) = 3.6691$, $e(3) = 6.0565$, $e(4) = 8.3861$ and $e(5) = 11.2392$, suggesting that a much smaller exponent 
$e(R_*)\approx 2.504 R_*-1.4262$ 
than $4R_*+2$.

\begin{figure}[hthp]
        \centering
     \subfigure[Two-point resolution of MUSIC, $R_*=2$]{\includegraphics[width=7cm]{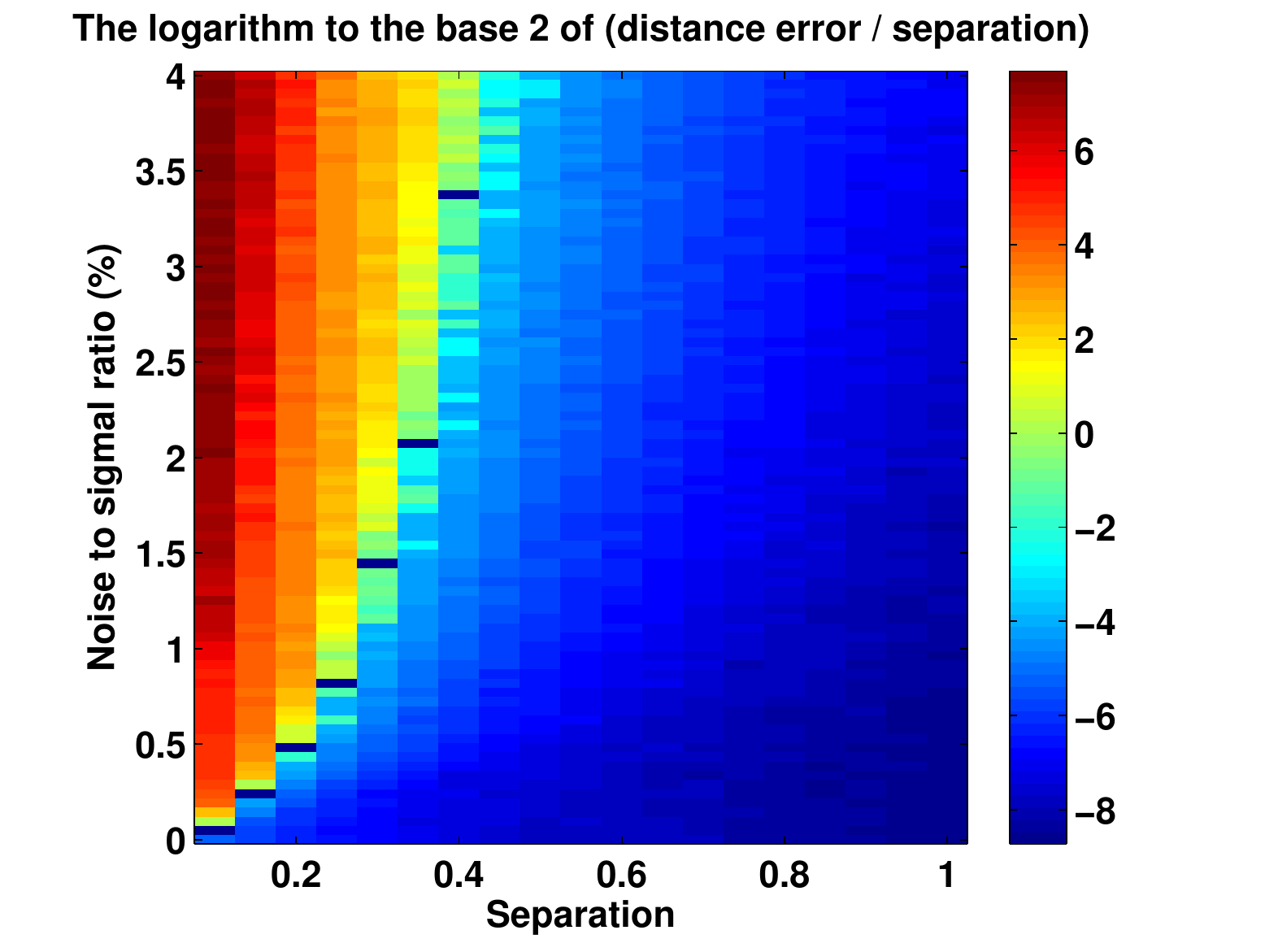}}
          \subfigure[Three-point resolution of MUSIC, $R_*=3$]{\includegraphics[width=7cm]{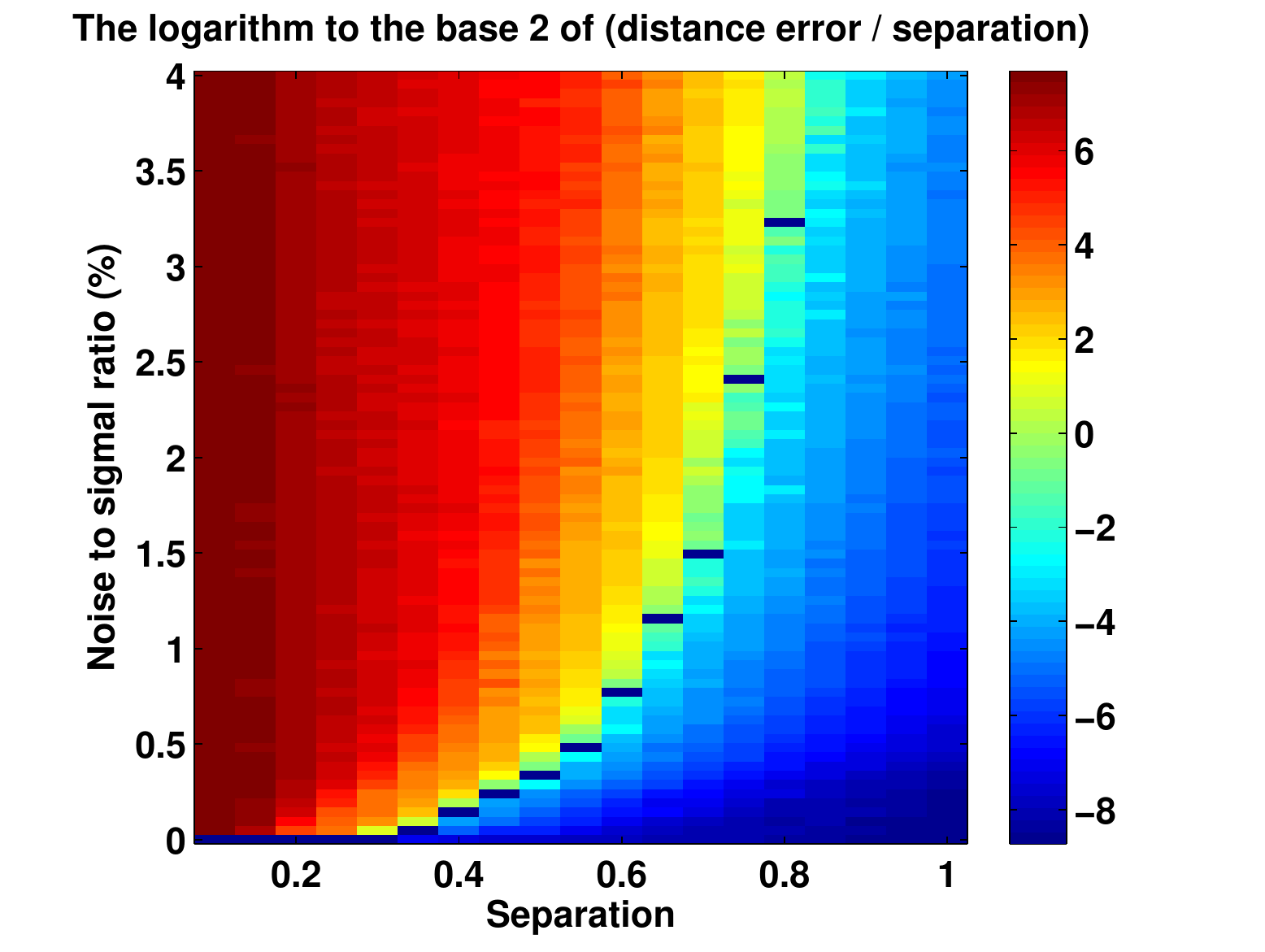}}
          \subfigure[Four-point resolution of MUSIC, $R_*=4$]{\includegraphics[width=7cm]{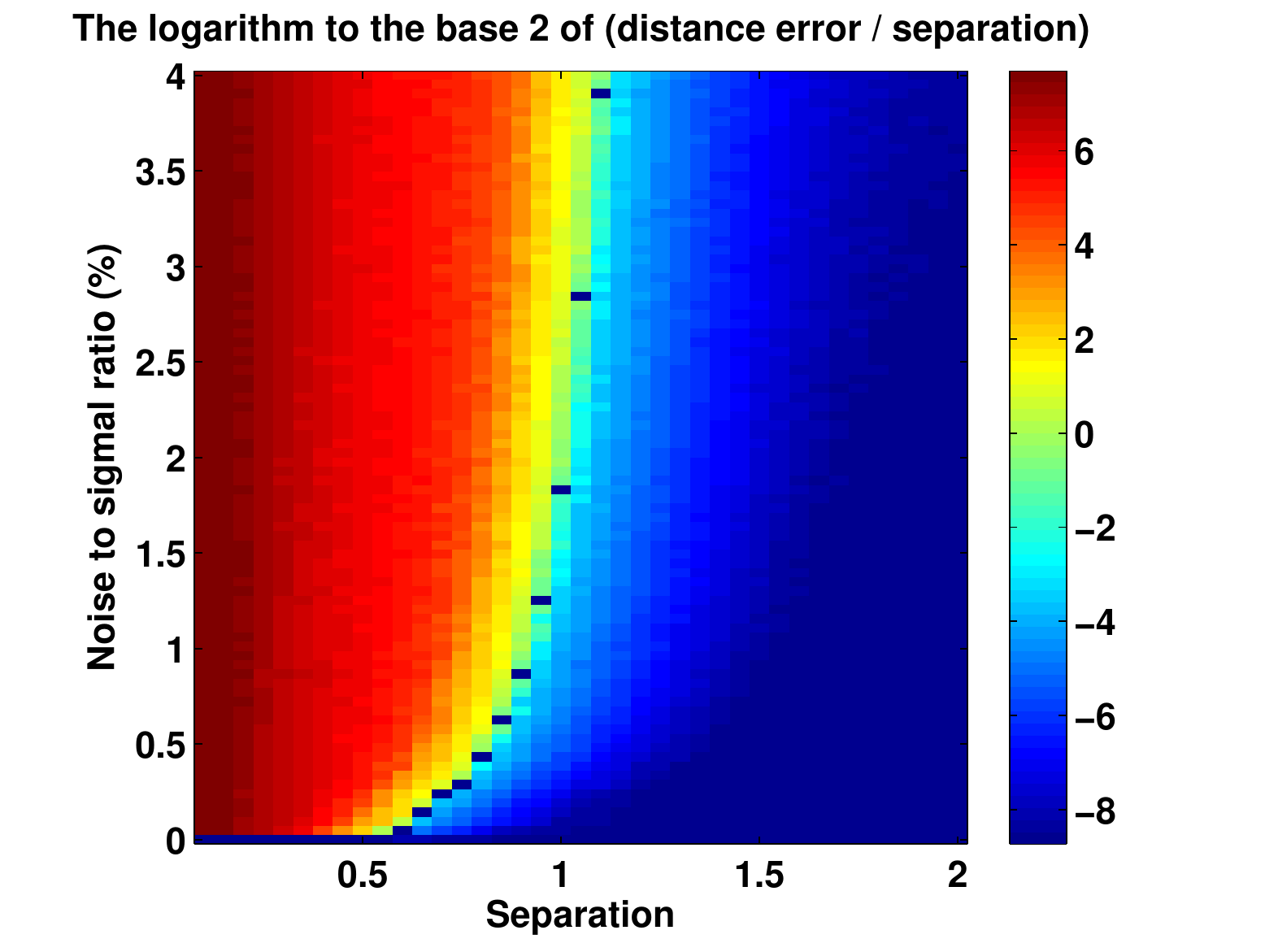}}
                    \subfigure[Five-point resolution of MUSIC, $R_*=5$]{\includegraphics[width=7cm]{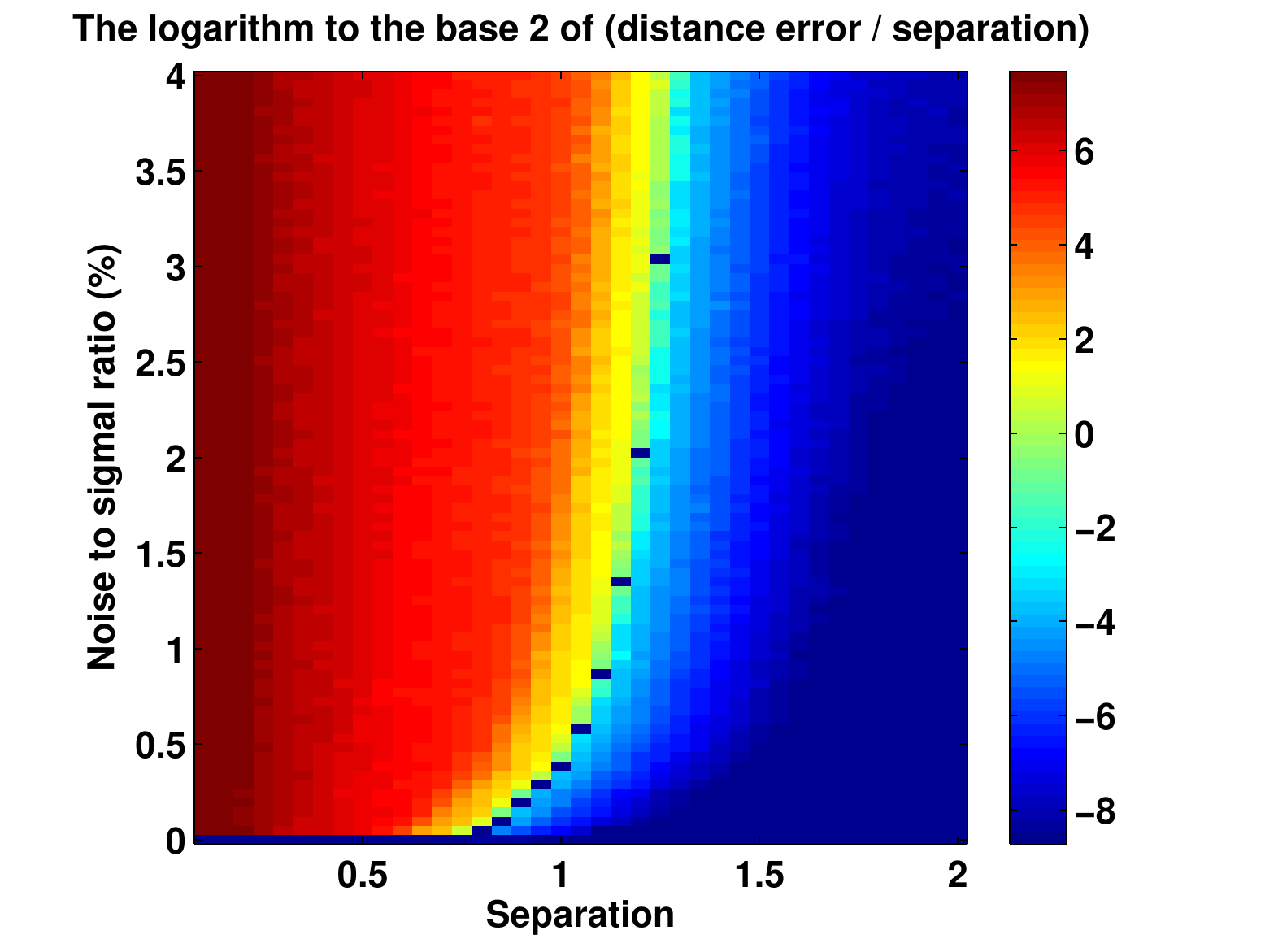}}
                     \subfigure[Phase transition curves]{\includegraphics[width=7cm]{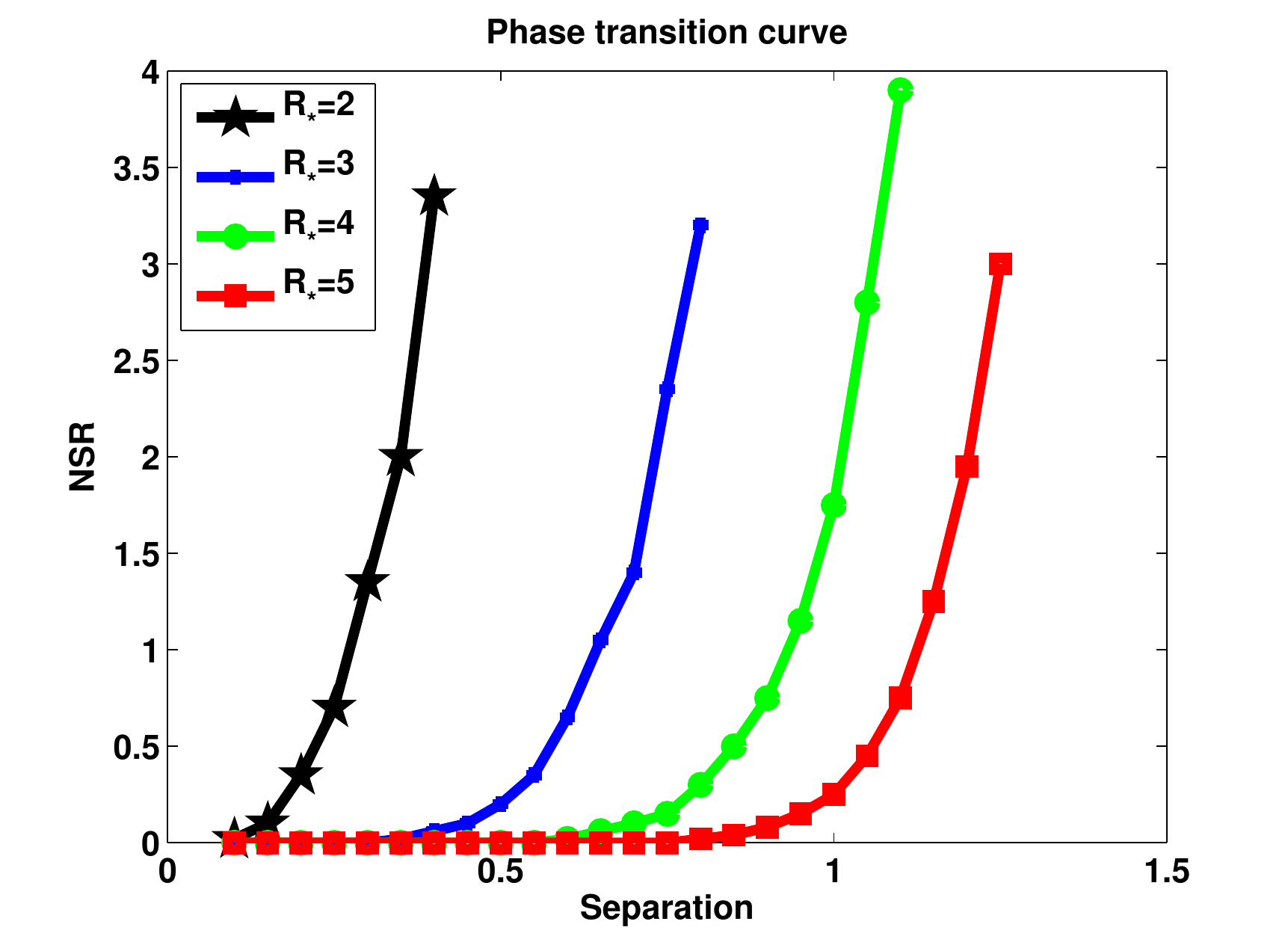}}
                           \subfigure[Log-log plot of phase transition curves]{\includegraphics[width=7cm]{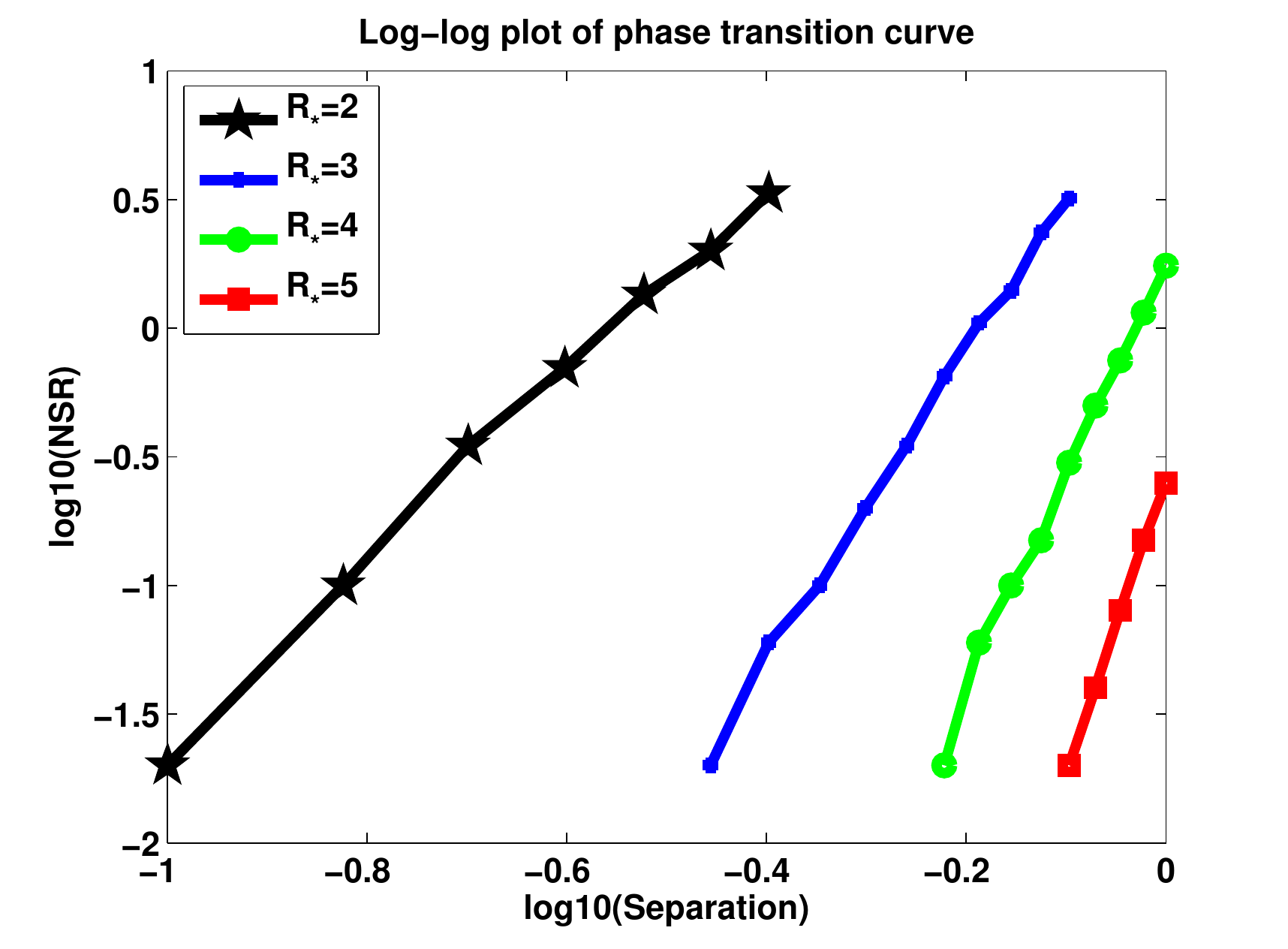}}
     \caption{Color plots in (a)-(d) shows the logarithm to the base 2 of average $d(\supp,\hat\supp)/q$ with respect to NSR (y-axis) and $q$ (x-axis) in the unit of RL. 
Reconstruction is considered successful if $d(\supp,\hat\supp)/q < 1/2$ (from green to black). A clear phase transition is observed. Transition points from which $d(\supp,\hat \supp)/q < 1/2$ are marked out by black bars in (a)-(d). Phase-transition curves connecting the black bars for $R_*=2,3,4,5$ are shown in (e) and in (f) in log-log scale. Least squares fitting for the slope of the log-log plot yields $e(2) = 3.6691$, $e(3) = 6.0565$, $e(4) = 8.3861$ and $e(5) = 11.2392$. 
  }
     \label{figPT}
\end{figure}

\commentout{     
\begin{figure}[hthp]
        \centering
     \subfigure[$\smin^2(\Phi^L)$ while $R=2$]{\includegraphics[width=8cm]{SingularValue1DR2.eps}}
           \subfigure[$\smin^2(\Phi^L)$ while $R=3$]{\includegraphics[width=8cm]{SingularValue1DR3.eps}}
      \subfigure[$\smin^2(\Phi^L)$ while $R=4$]{\includegraphics[width=8cm]{SingularValue1DR4.eps}}
       \subfigure[$\smin^2(\Phi^L)$ while $R=5$]{\includegraphics[width=8cm]{SingularValue1DR5.eps}}
     \caption{$L = 10$}
\end{figure}     
}


\commentout{
\begin{figure}[h]
        \centering
       \subfigure[MUSIC. Red: exact; Blue: imaging function $J^\ep(\om)$ normalized to ${[0,1]}$; Green: candidates picked from local maxima of $J^\ep(\om)$. $d(\hat\supp ,\supp)\approx 0.04$ RL.]{\includegraphics[width=4.8cm]{FigCompare/Sep03Range1Noise1/MUSIC1.pdf}}
       \quad
        \subfigure[BLOOMP. Red: exact; Blue: recovered. $d(\hat\supp ,\supp)\approx 0.71$ RL.]{\includegraphics[width=4.8cm]{FigCompare/Sep03Range1Noise1/BLOOMP1.pdf}}
        \quad
         \subfigure[SDP. Red: exact; Blue: recovered.  $d(\hat\supp ,\supp)\approx 0.02$ RL.]{\includegraphics[width=4.8cm]{FigCompare/Sep03Range1Noise1/SDP1.pdf}}
       \subfigure[MUSIC. $d(\hat\supp ,\supp)\approx 0.01$ RL.]{\includegraphics[width=4.8cm]{FigCompare/Sep03Range1Noise1/MUSIC2.pdf}}
       \quad
        \subfigure[BLOOMP. $d(\hat\supp ,\supp)\approx 0.38$ RL.]{\includegraphics[width=4.8cm]{FigCompare/Sep03Range1Noise1/BLOOMP2.pdf}}
        \quad
         \subfigure[SDP. Red: exact; Blue: recovered.  $d(\hat\supp ,\supp)\approx 0.35$ RL.]{\includegraphics[width=4.8cm]{FigCompare/Sep03Range1Noise1/SDP2.pdf}}
          \caption{Reconstruction of 2 real-valued frequencies of the same sign (a)(b)(c) and of opposite sign (d)(e)(f). Separation = $0.3$ RL. Dynamic range $= 1$ and NSR $ = 1\%$. BLO-based DPSS MF is not capable of resolving objects separated below 1 RL since the half width of DPSS is above 1 RL.
          }
              \label{figc4}
\end{figure}
}

\section{Conclusion and extension}
\label{secconc}

We have provided a stability analysis of the MUSIC algorithm for single-snapshot spectral estimation off the grid. We have proved that perturbation of the noise-space correlation by external noise is roughly proportional to the spectral norm of the noise Hankel matrix with a magnification factor given in terms of maximum and minimum nonzero singular values of the Hankel matrix constructed from the noiseless measurements. Under  the assumption of frequency 
 separation roughly $\geq$  2 RL, the magnification factor is explicitly estimated by means of a new version of discrete Ingham inequalities.

A systematic numerical study has shown that the MUSIC algorithm enjoys strong stability and low computation complexity for the reconstruction of well-separated frequencies. MUSIC is the only algorithm that can recover arbitrarily closely spaced frequencies as long as the noise is sufficiently small. And we have numerically documented the super-resolution effect of MUSIC in terms of the relationship among  the minimum separation, the Rayleigh index (the size of largest cluster) and the noise. The results
conform to the optimal bound conjectured (and partially proved) by Donoho \cite{Donoho92}.

Finally  we discuss a possible extension of the present work.
We became aware of the reference \cite{YY} after completing the first 
draft of this work ({\tt  arXiv:1404.1484}). In \cite{YY}, Chen and Chi used the matrix completion technique to obtain a stable approximation of $\{ y(k), k = 0,\ldots,M\}$ from its partial noisy samples. This can be used as the preprocessing  denoising step before invoking the single-snapshot MUSIC. 
Together \cite{YY} and the present work constitute a framework for single-snapshot spectral estimation with  compressive noisy measurements. 

Let $y$ and $y^\ep$, respectively, be the full set of noiseless and noisy data as before. Let 
the sampling set $\Lambda$ be a random subset of size $m$  from $\{0,\ldots,M\}$. Let $\calP_\Lambda(v)$ be the orthogonal projection of $v\in \CC^{M+1}$ onto the space of vectors supported on $\Lambda$. 

For the noisy compressive data $\calP_\Lambda y^\ep$ satisfying  $\|\calP_\Lambda (y^\ep - y)\|_2 \le \delta$,
\cite{YY} proposes the  following denoising strategy of  Hankel matrix completion
\[
{\hat y} =  {\rm arg}\min_{z} \|{\rm Hankel}(z)\|_\star, \  
 \text{s.t.}\quad  \|\calP_\Lambda (z-y^\ep) \|_2 \le\delta
 \]
where $\|\cdot\|_*$ denotes the nuclear norm.

The total procedure of compressive spectral estimation is given in the following table.
\begin{center}
   \begin{tabular}{|l|}\hline
    { \centerline{\bf Spectral estimation with compressive measurements}} \\ \hline
    {\bf Input:} $\calP_\Lambda y^\ep \in \CC^{M+1},\delta, s, L$. \\
    1) Matrix completion:
          ${\hat y} =  {\rm arg}\min_{z \in \CC^{M \times 1}} \|{\rm Hankel}(z)\|_\star, \  
 \text{subject to } \|\calP_\Lambda (z-y^\ep) \|_2 \le\delta$
    \\
     2) Form Hankel matrix $\hat H = {\rm Hankel}(\hat y) \in \CC^{(L+1)\times(M-L+1)}$.
     \\
     3) SVD: $\hat H = [\hat{U}_1\  \hat{U}_2] {\rm diag}(\hat{\si}_1 , \ldots , \hat{\si}_s ,\ldots) [\hat{V}_1\ \hat{V}_2]^\star $, where $\hat{U}_1 \in \CC^{(L+1)\times s}$.\\
     4) Compute imaging function $\hat{J}(\om) = \|\phi^{L}(\om)\|_2 /\|{\hat{U}_2}^\star \phi^L(\om)\|_2$. \\
   {\bf Output:} $\hat \supp =\{ \om \text{ corresponding to } s \text{ largest local maxima of } \hat{J}(\om) \} $.\\
    \hline
   \end{tabular}
\end{center}

The following  estimate on the difference  $\hat R(\om)-R(\om)$ is obtained by 
combining \cite[Theorem 2]{YY} and Theorem \ref{thmp1}.

\begin{theorem}
\label{thmYY}

Let $ R(\om) $ and  $\hat R(\om)$, respectively, be the noise-space correlation functions for the noiseless $y$ and  denoised data $\hat y$. 

Let $\Lambda$ of size $m$ be uniformly sampled at random from $\{0,\ldots,M\}$. Suppose
$\|\calP_\Lambda (y^\ep - y)\|_2 \le \delta$. Then there exists a universal constant $C>0$ such that 
\beq
\label{eqYY2}
\|{\rm Hankel}(\hat y) - {\rm Hankel}(y)\|_F \le 
\left(
2\sqrt{M+1} + 8(M+1) + \frac{8\sqrt 2 (M+1)^2}{m}
\right)
\delta
\eeq
and 
\beq
\label{eqYY3}
|\hat R(\om) - R(\om)| \le  \frac{4\si_1 + \left(
2\sqrt{M+1} + 8(M+1) + \frac{8\sqrt 2( M+1)^2}{m}
\right)
\delta
}{\left[\si_s-\left(
2\sqrt{M+1} + 8(M+1) + \frac{8\sqrt 2 (M+1)^2}{m}
\right)
\delta
\right]^2} 
 \left(
2\sqrt{M+1} + 8(M+1) + \frac{8\sqrt 2 (M+1)^2}{m}
\right)
\delta.
\eeq
with probability exceeding $1-(M+1)^{-2}$ provided that
\beq
\label{eqYY1}
m > C \mu \gamma s \log^3(M+1)
\eeq
where
$$\mu = \max\left(\frac{L+1}{\smin^2(\Phi^L)}, \frac{M-L+1}{\smin^2(\Phi^{M-L})}\right), \quad \gamma = \max\left(\frac{M+1}{L+1},\frac{M+1}{M-L+1} \right).$$
\end{theorem}

Theorem \ref{thmYY} implies that compressive spectral estimation is stable with matrix completion and MUSIC whenever $\Phi^L,\Phi^{M-L}$ are well-conditioned and the sample size is sufficiently large.

Furthermore with the discrete Ingham inequalities we can give
explicit estimate for the right hand side of (\ref{eqYY3}) and  (\ref{eqYY1}). 
\commentout{
Let $L$ and $M-L$ be even integers and $L = M/2$ so that $\gamma \approx 2$. Suppose $\supp$ satisfies the gap condition \eqref{eq9'}, then
\beq
\label{eqmu}
\mu < \max\left(
\frac{L+1}{L} 
\left(\frac 2 \pi - \frac{2}{\pi L^2 q^2} - \frac 4 L \right)^{-1} ,
\frac{M-L+1}{M-L}
\left(   \frac 2 \pi - \frac{2}{\pi (M-L)^2 q^2} -\frac{4}{M-L}
\right)^{-1}
\right).
\eeq
}
In particular,  with $L \approx M/2$ and well-separated  ($>$ 2 RL) frequencies,  $\mu$ and $\gamma$  scale like a constant and $m=\mathcal{O}( s \log^3 M)$ suffices for any sufficiently small $\delta$. 
In this case, the right hand side of (\ref{eqYY3}) is 
\[
\mathcal{O}\left( {M\over m}{\xmax \over \xmin}{\delta\over \xmin} \right) 
\]
showing  enhanced stability as $M/m \to 1$ where $M/m$ is the compression ratio, $\xmax/\xmin$ the object's peak-to-trough ratio and $\delta/\xmin$ the noise-to-object ratio. 

\commentout{Under the same gap condition \eqref{eq9'}, discrete Ingham inequalities yields an explicit perturbation bound of the noise-space correlation function as
$$
|\hat R(\om) - R(\om)| \le  \frac{4\alpha_1\sqrt{L(M-L)} + \left(
2\sqrt{M} + 8M + \frac{8\sqrt 2 M^2}{m}
\right)
\delta
}{\left[\alpha_s \sqrt{L(M-L)}-\left(
2\sqrt{M} + 8M + \frac{8\sqrt 2 M^2}{m}
\right)
\delta
\right]^2} 
\cdot
 \left(
2\sqrt{M} + 8M + \frac{8\sqrt 2 M^2}{m}
\right)
\delta,
$$
where $\alpha_1$ and $\alpha_2$ are defined in \eqref{eqalpha1} and \eqref{eqalpha2}. In this case $|\hat{R}(\om)-R(\om)| \rightarrow 0$ as $\delta\rightarrow 0$.
}

\section*{Acknowledgement}
Wenjing Liao would like to thank Armin Eftekhari for providing their codes and helpful discussions at SAMSI.

\appendixtitleon
\appendixtitletocon
\begin{appendices}

\section{Proof of Theorem \ref{thm1}}
\label{secmusic}

\commentout{
The MUSIC algorithm was introduced by Schmidt \cite{Sch,SchD} and many of its extensions, such as S-MUSIC\cite{SMUSIC}, IES-MUSIC\cite{IESMUSIC}, R-MUSIC\cite{RMUSIC} and RAP-MUSIC\cite{RAPMUSIC} were explored. These algorithms have been widely applied to various problems in imaging and sensing. According to Wikipedia, in a detailed evaluation based on thousands of simulations, M.I.T.'s Lincoln Laboratory concluded that, among currently accepted high-resolution algorithms, MUSIC was the most promising and a leading candidate for further study and actual hardware implementation. 

With the image vector defined in \eqref{imagingvector}, we cast the spectral estimation problem in the form 
\beq
y^\ep = \Phi^{M}  x +\ep
\label{linearsystem}
\eeq
where
$$\Phi^{M} := \Phi^{0\rightarrow M}= [\phi^{M}(\om_1)\  \phi^{M}(\om_2) \ \ldots \ \phi^{M}(\om_s)].$$

Fixing a positive integer $1< L < M$, we form the Hankel matrix 
\beq
\label{hankel}
H = {\rm Hankel}(y) =
\begin{bmatrix}
y_0 & y_1 & \ldots & y_{M-L}\\
y_1 & y_2 & \ldots & y_{M-L+1}\\
\vdots & \vdots & \vdots & \vdots\\
y_{L} & y_{L+1} & \ldots & y_{M}\\
\end{bmatrix}.
\eeq
With a single-snapshot measurement the MUSIC algorithm hinges a special property of Fourier measurements: a time translation corresponds to a frequency phase modulation. It is straightforward to verify that  ${\rm Hankel}(y)$ admits the decomposition  
\beqn
H = \Phi^L X (\Phi^{M-L})^T, \quad X = {\rm diag}(x_1,\ldots,x_s).
\eeqn
Let $H^{\ep} = {\rm Hankel}(\yep)$ and $E = {\rm Hankel}(\ep)$.
The noisy Hankel matrix $H^{\ep} $ can be written as 
\beqn
H^{\ep} = H + E = \Phi^L X (\Phi^{M-L})^T + E.
\eeqn
The core of the MUSIC algorithm lies in the identification of noise space from matrix $H^{\ep}$. 
}

In the noise-free case $\range(H)$ and $\range(\Phi^L)$ coincide if the matrix $X(\Phi^{M-L})^T$ has full row rank, i.e., $\rank(\Phi^{M-L}) = s$, which is guaranteed on the condition that $M-L +1\ge s$ and the frequencies in $\supp$ are pairwise distinct.

\begin{lemma}
\label{lemmarange}
If $\rank(\Phi^{M-L}) = s,$ then $\range(H) = \range(\Phi^L)$.
\end{lemma}

\begin{lemma}
$\rank(\Phi^L) = s$ if $L+1 \ge s$ and $\om_k \neq \om_l, \ \forall k \neq l$.
\label{lemmarank}
\end{lemma}

\begin{proof}
If $L+1\geq s$, then  $s\times s$ square submatrix $\Psi$ of $\Phi^L$ is a square Vandermonde matrix whose
determinant is given by
\[
\det{(\Psi)}=\prod_{1\leq i<j\leq s}(e^{-i2\pi \om_j} -e^{-i2\pi \om_i}).
\]
Clearly, $ \det{(\Psi)}\neq 0 $ if and only if $\om_i\neq \om_j, i\neq j$. 
Hence $\rank (\Psi)=s$ which implies $\rank(\Phi^L) = s$. 
\end{proof}

Similarly,  if $L +1\ge s+1$, the extended matrix $\Phi^L_\om = [\Phi^L \  \phi^L(\om)]$ has full column rank for any $\om \notin \mathcal{S}$. As a consequence, $\om \in \mathcal{S}$ if and only if $\phi^L(\om)$ belongs to $\range(\Phi^L)$. 

\section{Proof of Theorem \ref{thm3}}
\label{app2}

The proof of Theorem \ref{thm3} combines techniques used in \cite{Ingham} and \cite{KP}. We take 
$$g(t) = \cos \pi (t-0.5)$$
and let 
\beq
G(\om) = \sum_{k = 0}^{L} g(\frac k L)e^{2\pi i k \om}.
\label{eqG}
\eeq

\begin{figure}[h]
        \centering
       \subfigure[$g(t)$]{\includegraphics[width=5cm]{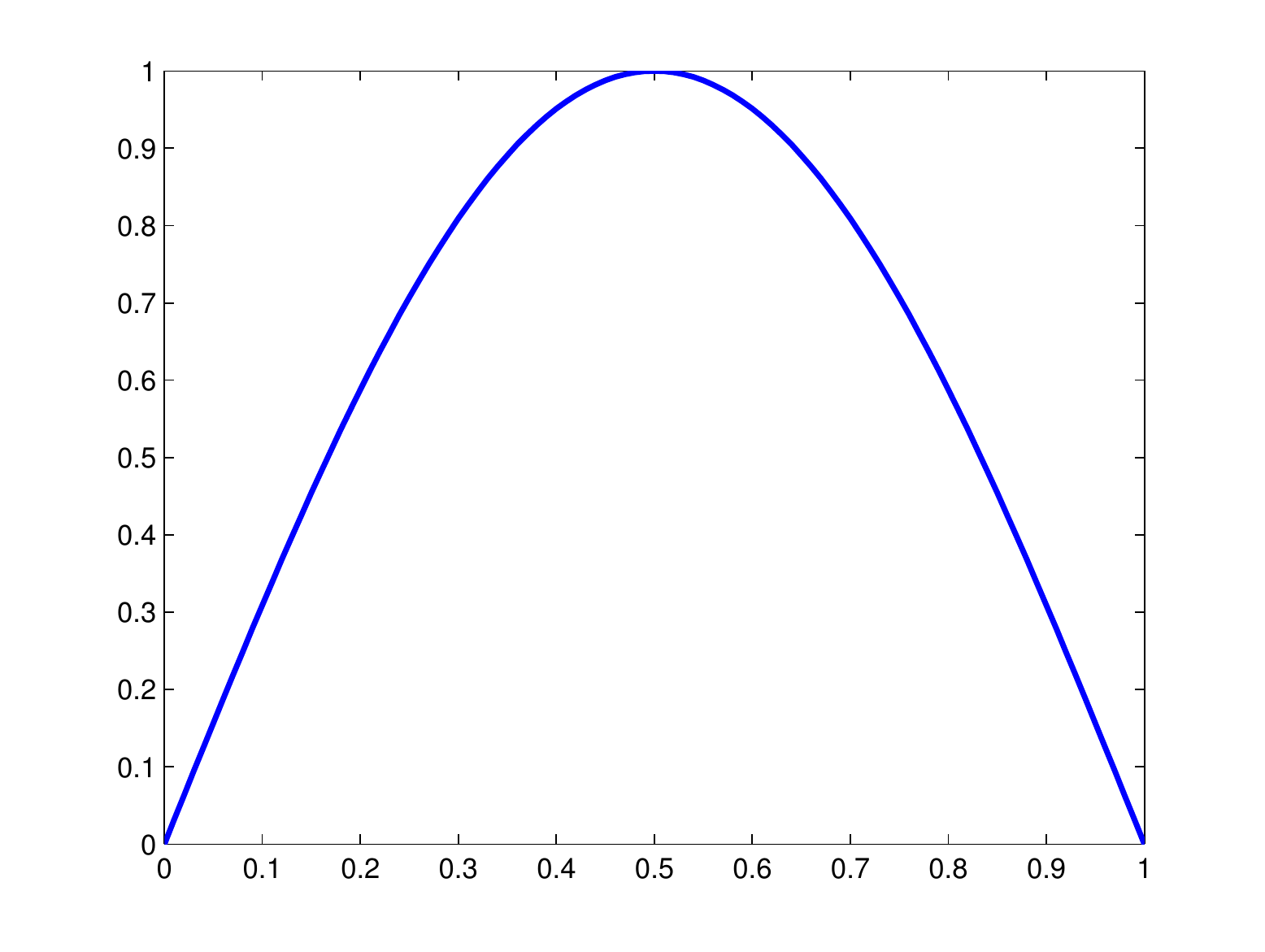}}
        \subfigure[$|G(\om)|/L$]{\includegraphics[width=5cm]{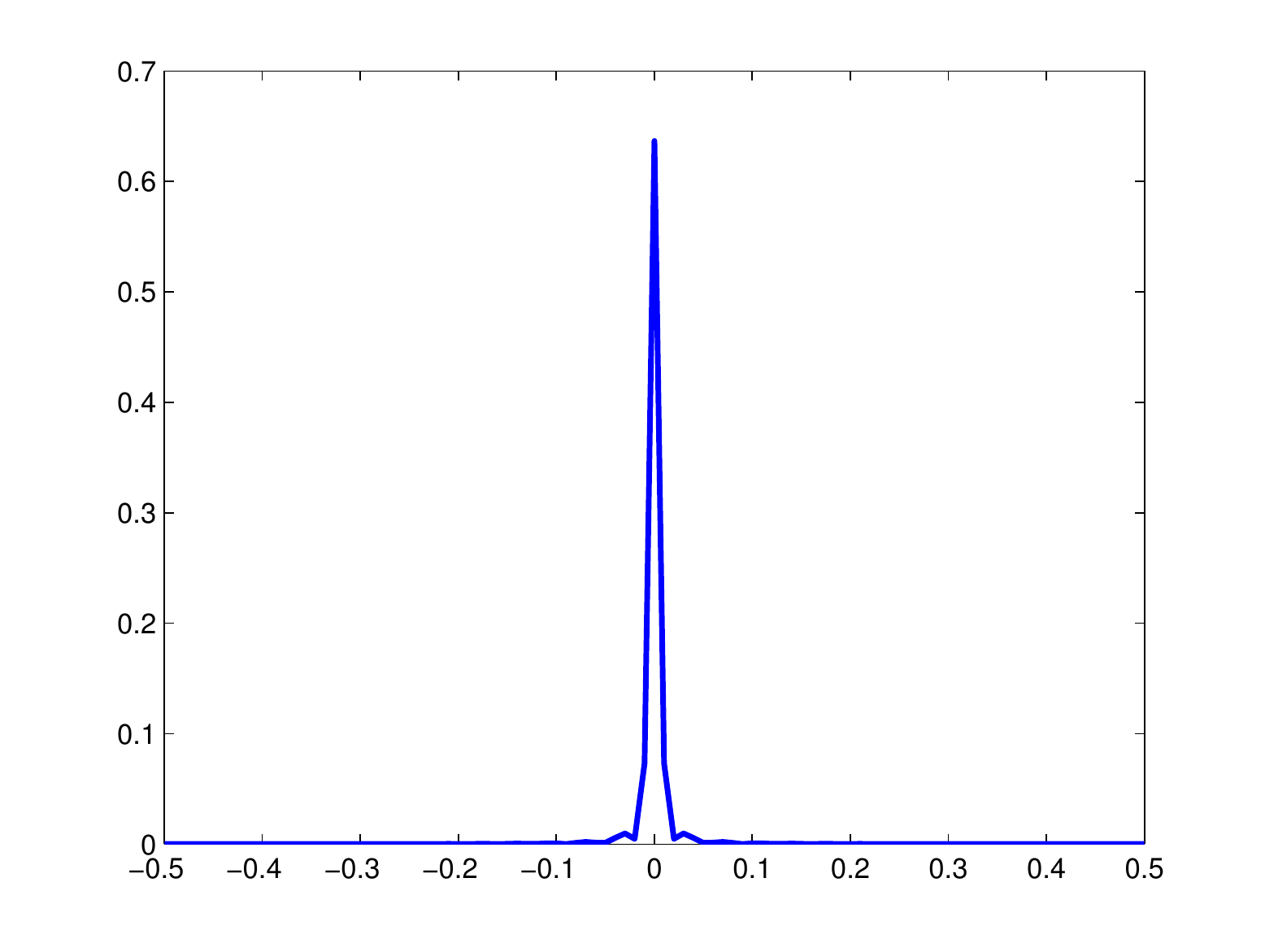}}
                \subfigure[Real part of $G(\om)/L$]{\includegraphics[width=5cm]{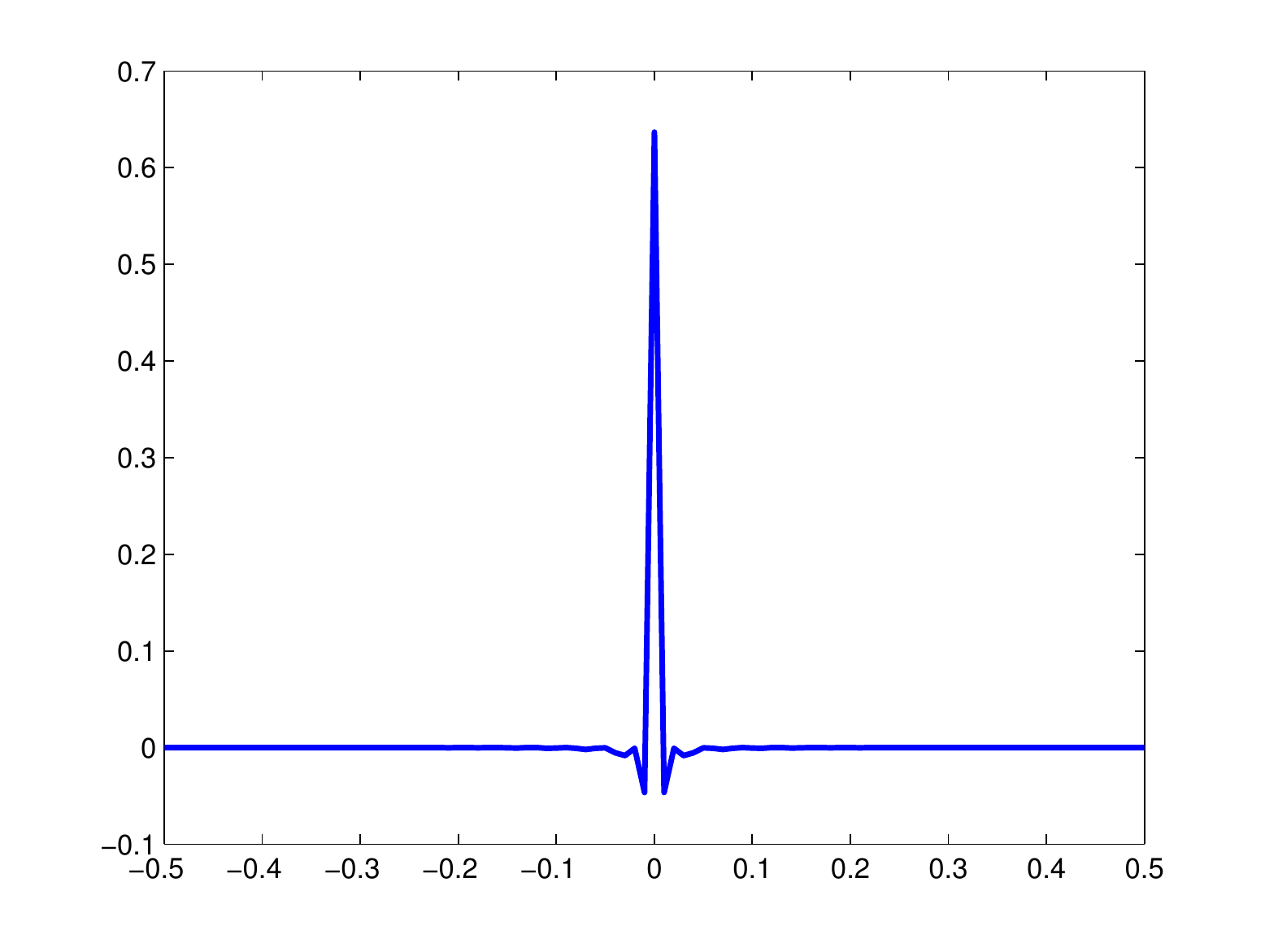}}
              \caption{Function $g(t)$ and $G(\om)/L$ when $L = 128$.}
              \label{fig1}
\end{figure}

Graphs of $g(t)$, $|G(\om)|/L$ and the real part of $G(\om)/L$ are shown in Figure \ref{fig1}. Function $G$ has the following properties.

\begin{lemma}
\label{lemma4}
\begin{enumerate}
   \setlength{\itemsep}{0.2cm}
\item $G(\om+n) = G(\om )$ for $n \in \ZZ$.

\item $G(-\om) = e^{-2\pi i L\om}G(\om)$ and $|G(-\om)| = |G(\om)|$.

\item $L(\frac 2 \pi - \frac{1}{L}) \le G(0) \le L(\frac 2 \pi+ \frac{1}{L}).$

\item $
|G(\om)| \le\frac{2}{\pi}\frac{L}{|1-4L^2\om^2|} + \frac{8}{\pi L}$ for $\om \in [0,1/2]$.
\end{enumerate}
\end{lemma}

\begin{proof}
\begin{enumerate}
\item 
$$G(\om+n) =  \sum_{k=0}^L g(\frac k L) e^{2\pi i k(\om+n)}= \sum_{k=0}^L g(\frac k L) e^{2\pi i k\om} = G(\om).$$

\item 
\begin{align*}
G(-\om) & = \sum_{k=0}^L g(\frac k L) e^{-2\pi i k \om} 
 = \sum_{l=0}^L g\left(\frac{L-l}{L}\right) e^{-2\pi i(L-l)\om} \text{ by letting } l = L-k \\
 & = \sum_{l=0}^L \left[\cos\pi (1-\frac l L - 0.5)\right] e^{-2\pi i(L-l)\om} 
 = \sum_{l=0}^L \left[\cos\pi (\frac l L - 0.5)\right] e^{-2\pi i(L-l)\om}
 \\
 & = \sum_{l=0}^L g(\frac l L) e^{-2\pi i(L-l)\om} = e^{-2\pi i L\om} \sum_{l=0}^L g(\frac l L) e^{2\pi i l \om} 
 = e^{-2\pi i L\om} G(\om).
\end{align*}

\item 
On the one hand, $$\|g\|_1 = \int_{0}^{1}\cos \pi (t-0.5)dt = \frac 2 \pi.$$ 
On the other hand,
$$G(0) = \displaystyle \sum_{0}^{L} g(\frac k L),$$ and $$\frac{1}{L}(G(0) - 1)\le \|g\|_1 \le  \frac{1}{L}(G(0) + 1).$$

\item 
According to the Poisson summation formula,
\begin{align}
& \frac 1 L G(\om) = \frac{1}{L} \sum_{k= 0}^{L} g(\frac{k}{L})e^{2\pi i k \om}
 = \sum_{r = -\infty}^{\infty} \int_{0}^{1} g(z)e^{2\pi i L(\om-r)z}dz
 = \frac{2}{\pi}\sum_{r=-\infty}^{\infty} \frac{\cos{\pi L(\om-r)}}{1-4L^2(\om-r)^2} e^{i\pi L(\om-r)}
\label{eqpoisson}.
\end{align}
Hence
\begin{align*}
& \Big|\frac 1 L G(\om)\Big| \le \frac{2}{\pi}\sum_{r=-\infty}^{\infty}\Big| \frac{\cos{\pi L(\om-r)}}{1-4L^2(\om-r)^2}\Big| \\
& \le \frac{2}{\pi}\frac{1}{|1-4L^2\om^2|} + \frac{2}{\pi} \sum_{r \neq 0} \frac{1}{4L^2(r-\om)^2-1} \\
& \le \frac{2}{\pi}\frac{1}{|1-4L^2\om^2|} + \frac{2}{\pi} \sum_{r \neq 0} \frac{2}{4L^2(r-\om)^2} \\
 & \le \frac{2}{\pi}\frac{1}{|1-4L^2\om^2|} + \frac{4}{\pi} \Big[\frac{1}{4L^2(\frac 1 2)^2}+\frac{1}{4L^2(1)^2} + \frac{1}{4L^2(\frac 3 2)^2} + \frac{1}{4L^2(2)^2} +\ldots\Big] \\
 & \hspace{5cm}\text{ as } \om \in [0,\frac 1 2],
  \\
  & \le  \frac{2}{\pi}\frac{1}{|1-4L^2\om^2|} + \frac{4}{\pi}  \frac{1}{L^2}\Big[\frac{1}{1^2}+\frac{1}{2^2} + \frac{1}{3^2} + \frac{1}{4^2} +\ldots\Big]  \\
  & \le  \frac{2}{\pi}\frac{1}{|1-4L^2\om^2|} + \frac{4}{\pi} \frac{1}{L^2} 2 =  \frac{2}{\pi}\frac{1}{|1-4L^2\om^2|} + \frac{8}{\pi L^2}.\\
\end{align*}

In \eqref{eqpoisson} the difference between the discrete and the continuous case lies in $$\frac{2}{\pi}\sum_{r \neq 0} \frac{\cos{\pi L(\om-r)}}{1-4L^2(\om-r)^2}$$
which is bounded above by $8/(\pi L^2)$, and is therefore negligible when $L$ is sufficiently large.
 
\end{enumerate}
\end{proof}

We start with the following lemma, which paves the way for the proof of Theorem \ref{thm3}.

\begin{lemma}
\label{lemma6}
Suppose objects in $\supp$ satisfy the gap condition 
\beq
\label{lemma60}
d(\om_j,\om_l) \ge q > \frac 1 L \sqrt{\frac 2 \pi}\left( \frac 2 \pi - \frac 1 L - \frac{8s}{\pi L^2}\right)^{-\frac 1 2}.
\eeq
Then
\beq
\Big(\frac{2}{\pi}-\frac 1 L - \frac{2}{\pi L^2 q^2} - \frac{8s}{\pi L^2}\Big)\|{{\bc}}\|_2^2
\le
\frac{1}{L}\sum_{k= 0}^{L} g(\frac{k}{L}) \Big|(\Phi^{0\rightarrow L}  {\bc})_k\Big|^2 
\le 
\Big(\frac{2}{\pi}+\frac 1 L + \frac{2}{\pi L^2 q^2} + \frac{8s}{\pi L^2}\Big)\|{\bc}\|_2^2\label{lemma61}
\eeq
for all ${\bc} \in \CC^s$.
\end{lemma}

\begin{proof}
\begin{align*}
&\sum_{k=0}^{L} g(\frac{k}{L}) \Big|(\Phi^{0\rightarrow L}  {\bc})_k\Big|^2  
= \sum_{k= 0}^{L} g(\frac{k}{L}) 
\overline{\sum_{j=1}^s {\bc}_j e^{-2\pi i k \om_j}}
\sum_{l=1}^s {\bc}_l e^{-2\pi i k \om_l}\\
& = \sum_{j=1}^s \sum_{l=1}^s \overline{{\bc}_j} {\bc}_l \sum_{k=0}^{L} g(\frac k L)e^{2\pi i k(\om_j -\om_l)} = \sum_{j=1}^s\sum_{l=1}^s G(\om_j-\om_l) \overline{{\bc}_j} {\bc}_l \\
& = G(0) \|{\bc}\|_2^2 + \sum_{j=1}^s \sum_{l\neq j} G(\om_j-\om_l)\overline{{\bc}_j}{\bc}_l.
\end{align*}
It follows from the triangle inequality that
\begin{align*}
 G(0) \|{\bc}\|_2^2 - \sum_{j=1}^s \sum_{l\neq j} |G(\om_j-\om_l)\overline{{\bc}_j}{\bc}_l|
\le
\sum_{k= 0}^{L} g(\frac{k}{L}) \Big|(\Phi^{0\rightarrow L }  {\bc})_k\Big|^2 \le G(0) \|{\bc}\|_2^2 + \sum_{j=1}^s \sum_{l\neq j} |G(\om_j-\om_l)\overline{{\bc}_j}{\bc}_l|, \\
\end{align*}
where $\displaystyle \sum_{j=1}^s \sum_{l\neq j} |G(\om_j-\om_l)\overline{{\bc}_j}{\bc}_l|$ can be estimated through Property 4 in Lemma \ref{lemma4}.
\begin{align}
&\sum_{j=1}^s \sum_{l\neq j} |G(\om_j-\om_l)\overline{{\bc}_j}{\bc}_l| 
 \le \sum_{j=1}^s \sum_{l\neq j} |G(\om_j-\om_l)|\frac{|{\bc}_j|^2+|{\bc}_l|^2}{2}  = \sum_{j=1}^s |{\bc}_j|^2 \sum_{l \neq j} |G(\om_j-\om_l)|
 \nonumber
\\
& = \sum_{j=1}^s |{\bc}_j|^2 \sum_{l \neq j} |G(d(\om_j,\om_l))|
\le \sum_{j=1}^s |{\bc}_j|^2 \sum_{l \neq j} \Big[\frac{2}{\pi} \frac{L}{|1-4L^2d^2(\om_j,\om_l)|} +\frac{8}{\pi L}\Big] 
\nonumber 
\\
& \le \sum_{j=1}^s |{\bc}_j|^2 \frac{4}{\pi} \sum_{n = 1}^{\lfloor \frac{s}{2}\rfloor} \Big[ \frac{L}{4 L^2 n^2 q^2 -1} +\frac{4}{ L} \Big]  \text{ as frequencies in } \supp \text{ are pairwise separated by } q >\frac{1}{L} 
\nonumber \\
& \le \sum_{j=1}^s |{\bc}_j|^2  \frac{4}{\pi} \Big[\frac{2 s}{ L} + \sum_{n=1}^\infty \frac{L}{4L^2 n^2 q^2 -1}\Big] 
\le \sum_{j=1}^s |{\bc}_j|^2  \frac{4}{\pi} \Big[\frac {2s}{ L} +  \frac{1}{L^2 q^2}\sum_{n=1}^\infty\frac{L}{4 n^2  -1}\Big]
\label{lemma62} \\
& \le \sum_{j=1}^s |{\bc}_j|^2  \frac{4}{\pi} \Big[\frac {2s}{ L} +  \frac{1}{L q^2} \frac{1}{2}\sum_{n=1}^\infty\Big(\frac{1}{2n-1}-\frac{1}{2n+1}\Big)\Big] 
=\sum_{j=1}^s |{\bc}_j|^2  \frac{4}{\pi} \Big[\frac {2s}{ L} +  \frac{1}{L q^2} \frac{1}{2}\Big] 
\nonumber\\
& = \|{\bc}\|_2^2  \frac{2}{\pi} \Big(\frac{1}{L q^2} + \frac{4s}{L}\Big),
\nonumber
\end{align}
where  $\lfloor {s\over 2} \rfloor$  denotes the nearest integer smaller than or equal to $s/2$. Kernel $g$ in \eqref{lemma61} is crucial for the convergence of the series in \eqref{lemma62}.

Therefore 
$$G(0) \|{\bc}\|_2^2 - \frac{2}{\pi} \Big(\frac{1}{L q^2} + \frac{4s}{L}\Big)\|{\bc}\|_2^2 
\le
\sum_{k= 0}^{L} g(\frac{k}{L}) \Big|(\Phi^{0\rightarrow L}  {\bc})_k\Big|^2 \le G(0) \|{\bc}\|_2^2 + \frac{2}{\pi} \Big(\frac{1}{L q^2} + \frac{4s}{L}\Big)\|{\bc}\|_2^2 . $$

The equation above along with Property 3 in Lemma \ref{lemma4} yields
$$
L\Big(\frac{2}{\pi}-\frac 1 L - \frac{2}{\pi L^2 q^2} - \frac{8s}{\pi L^2}\Big)\|{\bc}\|_2^2
\le
\sum_{k= 0}^{L} g(\frac{k}{L}) \Big|(\Phi^{0\rightarrow  L}  {\bc})_k\Big|^2 
\le 
L\Big(\frac{2}{\pi}+\frac 1 L + \frac{2}{\pi L^2 q^2} + \frac{8s}{\pi L^2}\Big)\|{\bc}\|_2^2. $$

The gap condition \eqref{lemma60} is derived from the positivity condition of the lower bound, i.e.,
$$
\frac{2}{\pi}-\frac 1 L - \frac{2}{\pi L^2 q^2} - \frac{8s}{\pi L^2} > 0.$$
\end{proof}

Proof of Theorem \ref{thm3} is given below.

\begin{proof}
Given that $\supp = \{\om_1,\ldots,\om_s\} \subset [0,1)$ and frequencies are separated above $1/L$, there are no more than $L$ frequencies in $\supp$, i.e., $s<L$.

The lower bound in Theorem \ref{thmp2} follows from Lemma \ref{lemma6} as  
\begin{align*} 
&\|\Phi^{0\rightarrow L}{\bc}\|_2^2 
=
\sum_{k = 0}^{L} \Big|(\Phi^{0\rightarrow L}  {\bc})_k\Big|^2 
=
\sum_{k= 0}^{L}  \Big|(\Phi^{0\rightarrow L}  {\bc})_k\Big|^2
\ge \sum_{k= 0}^{L} g(\frac{k}{L}) \Big|(\Phi^{0\rightarrow L}  {\bc})_k\Big|^2 \\
&
\ge L\Big(\frac{2}{\pi}-\frac 1 L - \frac{2}{\pi L^2 q^2} - \frac{8s}{\pi L^2}\Big) >
L\Big(\frac{2}{\pi}- \frac{2}{\pi L^2 q^2} -\frac 4 L\Big).
\end{align*}
The gap condition \eqref{sep} in Theorem \ref{thmp2} is derived from the positivity condition of the lower bound, i.e.,
$$\Big(\frac{2}{\pi}- \frac{2}{\pi L^2 q^2} -\frac 4 L\Big)>0.$$

We prove the upper bound in Theorem \ref{thmp2} in two cases: $L$ is even or $L$ is odd.

\begin{description}
\item [Case 1: $L$ is even.]
First we substitute $L$ with $2L$ in \eqref{lemma61} and obtain
\beq
\label{lemma63}
\sum_{k= 0}^{2L} g(\frac{k}{2L}) \Big|(\Phi^{0\rightarrow 2L}  {\bc})_k\Big|^2 
\le 
2L\Big(\frac{2}{\pi} + \frac{1}{2L} + \frac{2}{4\pi L^2 q^2} + \frac{8s}{4\pi L^2}\Big)\|{\bc}\|_2^2.
\eeq

Let $D^{\frac L 2} = {\rm diag}(e^{-2\pi i \om_1 \frac L 2}, e^{-2\pi i \om_2 \frac L 2},\ldots,e^{-2\pi i \om_s \frac L 2})$ and $D^{-\frac L 2} = (D^{\frac L 2})^{-1}$.
On the one hand,
\begin{align*}
& \sum_{k= 0}^{2L} g(\frac{k}{2L}) \Big|(\Phi^{0\rightarrow 2L}  D^{-\frac L 2}{\bc})_k\Big|^2 
 \ge 
\sum_{k= L/2}^{3L/2} g(\frac{k}{2L}) \Big|(\Phi^{0\rightarrow 2L}  D^{-\frac L 2}{\bc})_k\Big|^2 
\ge
\sum_{k= L/2}^{3L/2} g(\frac 1 4) \Big|(\Phi^{0\rightarrow 2L}  D^{-\frac L 2}{\bc})_k\Big|^2 \\
& =
\frac{1}{\sqrt 2}\sum_{k= L/2}^{3L/2}  \Big|(\Phi^{0\rightarrow 2L}  D^{-\frac L 2}{\bc})_k\Big|^2 
=
\frac{1}{\sqrt 2} \sum_{k= 0}^{L}  \Big|(\Phi^{\frac L 2 \rightarrow \frac{3L}{2}} D^{-\frac L 2} {\bc})_k\Big|^2 
=
\frac{1}{\sqrt 2} \sum_{k= 0}^{L}  \Big|(\Phi^{ 0 \rightarrow L} D^{\frac L 2} D^{-\frac L 2} {\bc})_k\Big|^2 
\\
& =
\frac{1}{\sqrt 2} \sum_{k= 0}^{L}  \Big|(\Phi^{ 0 \rightarrow L}  {\bc})_k\Big|^2. 
\end{align*}

On the other hand, \eqref{lemma63} implies 
\begin{align*}
\sum_{k= 0}^{2L} g(\frac{k}{2L}) \Big|(\Phi^{0\rightarrow 2L}  D^{-\frac L 2}{\bc})_k\Big|^2 
& \le 
2L\Big(\frac{2}{\pi} + \frac{1}{2L} + \frac{2}{4\pi L^2 q^2} + \frac{8s}{4\pi L^2}\Big)\|D^{-\frac L 2}{\bc}\|_2^2
\\
& =L\Big(\frac{4}{\pi} + \frac{1}{L} + \frac{1}{\pi L^2 q^2} + \frac{4s}{\pi L^2}\Big)\|{\bc}\|_2^2.
\end{align*}

As a result
$$\|\Phi^L {\bc}\|_2^2
= \sum_{k= 0}^{L}  \Big|(\Phi^{ 0 \rightarrow L}  {\bc})_k\Big|^2
\le
L\Big(\frac{4\sqrt 2}{\pi} + \frac{\sqrt 2}{L} + \frac{\sqrt 2}{\pi L^2 q^2} + \frac{4\sqrt 2 s}{\pi L^2}\Big)\|{\bc}\|_2^2 
<
L\Big(\frac{4\sqrt 2}{\pi}  + \frac{\sqrt 2}{\pi L^2 q^2}+ \frac{3\sqrt 2}{L} \Big)\|{\bc}\|_2^2. 
$$
When $L$ is an even integer,
$$
\Big(\frac{2}{\pi}- \frac{2}{\pi L^2 q^2} -\frac 4 L\Big) \|\bc\|_2^2
\le
\frac{1}{L}\|\Phi^L {\bc}\|_2^2
\le
\Big(\frac{4\sqrt 2}{\pi}  + \frac{\sqrt 2}{\pi L^2 q^2}+ \frac{3\sqrt 2}{L} \Big)\|{\bc}\|_2^2. 
$$

\item [Case 2: $L$ is odd.]
\begin{align}
& \|\Phi^L {\bc}\|_2^2
= \sum_{k= 0}^{L}  \Big|(\Phi^{ 0 \rightarrow L}  {\bc})_k\Big|^2
\le
\sum_{k= 0}^{L+1}  \Big|(\Phi^{ 0 \rightarrow L+1}  {\bc})_k\Big|^2
\nonumber \\
& <
(L+1)\Big(\frac{4\sqrt 2}{\pi}  + \frac{\sqrt 2}{\pi (L+1)^2 q^2}+ \frac{3\sqrt 2}{L+1} \Big)\|{\bc}\|_2^2.
\label{lemma65} 
\end{align}

Eq. \eqref{lemma65} above follows from Case 1 as $L+1$ is an even integer. In summary, when $L$ is an odd integer,
$$
\left(\frac{2}{\pi}- \frac{2}{\pi L^2 q^2} -\frac 4 L\right) \|\bc\|_2^2
\le
\frac{1}{L}\|\Phi^L {\bc}\|_2^2
\le
\left(1+\frac 1 L\right)\left(\frac{4\sqrt 2}{\pi}  + \frac{\sqrt 2}{\pi (L+1)^2 q^2}+ \frac{3\sqrt 2}{L+1} \right)\|{\bc}\|_2^2. 
$$

\commentout{
Let $D^{\frac {L-1}{2}} = {\rm diag}(e^{-2\pi i \om_1 \frac {L-1}{2}}, e^{-2\pi i \om_2 \frac{L-1}{2}},\ldots,e^{-2\pi i \om_s \frac{L-1}{2}})$ and $D^{-\frac{L-1}{2}} = (D^{\frac{L-1}{2}})^{-1}$.

First we derive a lower bound of $g(\frac 1 4 - \frac{1}{4L})$.
\begin{align}
& g(\frac 1 4 - \frac{1}{4L}) = \cos(\frac \pi 4 +\frac{\pi}{4L}) = \cos(\frac \pi 4)\cos(\frac{\pi}{4L}) - \sin(\frac \pi 4)\sin(\frac{\pi}{4L}) \nonumber
\\
&\ge \frac{1}{\sqrt 2} \left( \cos(\frac{\pi}{4L}) - \sin(\frac{\pi}{4L}) \right) > \frac{1}{\sqrt 2} (1-\frac{1}{2L}-\frac{\pi}{4L}) >  
\frac{1}{\sqrt 2} (1-\frac{3}{2L}).
\label{lemma64}
\end{align}
The estimation is \eqref{lemma64} is based on $\cos(\frac{\pi}{4L}) > 1-\frac{1}{2L}$ and $\sin(\frac{\pi}{4L}) < \frac{\pi}{4L}$.

On the one hand,
\begin{align*}
& \sum_{k= 0}^{2L} g(\frac{k}{2L}) \Big|(\Phi^{0\rightarrow 2L}  D^{-\frac{L-1}{2}}{\bc})_k\Big|^2 
 \ge 
\sum_{k= (L-1)/2}^{(3L-1)/2} g(\frac{k}{2L}) \Big|(\Phi^{0\rightarrow 2L}  D^{-\frac{L-1}{2}}{\bc})_k\Big|^2 
\\
& \ge
\sum_{k= (L-1)/2}^{(3L-1)/2} g(\frac 1 4 - \frac{1}{4L}) \Big|(\Phi^{0\rightarrow 2L}  D^{-\frac{L-1}{2}}{\bc})_k\Big|^2 
>
\frac{1}{\sqrt 2}(1-\frac{3}{2L})\sum_{k= (L-1)/2}^{(3L-1)/2}  \Big|(\Phi^{0\rightarrow 2L}  D^{-\frac{L-1}{2}}{\bc})_k\Big|^2
\\ 
&
=
\frac{1}{\sqrt 2}(1-\frac{3}{2L})\sum_{k= 0}^{L}  \Big|(\Phi^{\frac{L-1}{2} \rightarrow \frac{3L-1}{2}} D^{-\frac{L-1}{2}} {\bc})_k\Big|^2 
=
\frac{1}{\sqrt 2}(1-\frac{3}{2L}) \sum_{k= 0}^{L}  \Big|(\Phi^{ 0 \rightarrow L} D^{\frac{L-1}{2}} D^{-\frac{L-1}{2}} {\bc})_k\Big|^2 
\\
& =
\frac{1}{\sqrt 2}(1-\frac{3}{2L}) \sum_{k= 0}^{L}  \Big|(\Phi^{ 0 \rightarrow L}  {\bc})_k\Big|^2. 
\end{align*}
}

\end{description}
\end{proof}

\section{Proof of Theorems in Section \ref{secper}}

\subsection{Proof of Theorem \ref{thmp1}}
\label{app1}
\begin{proof}

Let $\HH = H H^\star$, $\HHE = H^\ep {H^\ep}^\star$ and $\EE = H E^\star + E H^\star + E E^\star$. Then
\beq
\label{p1e1}
\HHE = \HH + \EE,
\eeq
and
\beq
\label{p1e2}
\HH = \begin{bmatrix} U_1 & U_2 \end{bmatrix}
 \begin{bmatrix} \SI_1\SI_1^\star & 0 \\ 0 & 0 \end{bmatrix}
 \begin{bmatrix} U^\star_1 \\ U^\star_2 \end{bmatrix},
 \eeq

\beq
\label{p1e3}
\HHE = \begin{bmatrix} \UE_1 & U^\ep_2 \end{bmatrix}
 \begin{bmatrix} \SE_1{\SE_1}^\star & 0 \\ 0 & \SE_2{\SE_2}^\star\end{bmatrix}
 \begin{bmatrix} {\UE}^\star_1 \\ {\UE}^\star_2 \end{bmatrix},
\eeq
where $\SI_1 = {\rm diag}(\si_1,\ldots,\si_s), \SE_1 = {\rm diag}(\se_1,\ldots,\se_s),$ and $\SE_2 = {\rm diag}(\se_{s+1},\se_{s+2},\ldots).$

Combining \eqref{p1e1} and \eqref{p1e3} yields
\beq
\label{p1e4}
\begin{bmatrix} {U}^\star_1 \\  {U}^\star_2 \end{bmatrix}
(\HH + \EE)
 \begin{bmatrix} \UE_1 & U^\ep_2 \end{bmatrix}
 = \begin{bmatrix}
 U_1^\star \UE_1 \SE_1 {\SE_1}^\star & U_1^\star \UE_2 \SE_2 {\SE_2}^\star \\
 U_2^\star \UE_1 \SE_1{\SE_1}^\star & U_2^\star \UE_2 \SE_2 {\SE_2}^\star
 \end{bmatrix}.
\eeq

On the one hand, the (2,1) entries on both sides of \eqref{p1e4} are equal such that
 $$U_2^\star \HH U_1^\ep + U_2^\star \EE U_1^\ep = U_2^\star \UE_1 \SE_1 {\SE_1}^\star.$$
 Based on \eqref{p1e2}, we obtain $U_2^\star \HH  = \mathbf 0$ and therefore
 $$ U_2^\star \EE U_1^\ep ( \SE_1 {\SE_1}^\star)^{-1} = U_2^\star \UE_1$$
 which implies 
 \beq
 \label{p1e5}
 \| U_2^\star \UE_1\|_2 \le \frac{\|\EE\|_2}{(\se_s)^2}
 \eeq
 
 On the other hand, the (1,2) entries on both sides of \eqref{p1e4} are equal such that
 $$U_1^\star \HH U_2^\ep + U_1^\star \EE U_2^\ep = U_1^\star \UE_2 \SE_2 {\SE_2}^\star.$$
 Based on \eqref{p1e2}, we obtain $U_1^\star \HH  = \SI_1 \SI_1^\star U_1^\star$ and therefore
 $$\SI_1 \SI_1^\star U_1^\star \UE_2 +  U_1^\star \EE \UE_2  = U_1^\star \UE_2 \SE_2 {\SE_2}^\star.$$
For any $\phi \in \CC^{L+1-s}$,
\begin{align*}
 \si_s^2 \|U_1^\star \UE_2\phi\|_2
& \le
\|\SI_1 \SI_1^\star U_1^\star \UE_2\phi \|_2 
\le
\|U_1^\star \EE \UE_2\|_2\|\phi\|_2 + \|U_1^\star \UE_2 \SE_2 {\SE_2}^\star\phi\|_2 \\
& 
\le 
\|\EE\|_2 \|\phi\|_2 + \|U_1^\star \UE_2\|_2 (\se_{s+1})^2 \|\phi\|_2,
\end{align*}
so
$$\frac{\|U_1^\star \UE_2\phi\|_2}{\|\phi\|_2}
\le
\frac{\|\EE\|_2 +  (\se_{s+1})^2\|U_1^\star \UE_2\|_2 }{\si_s^2}.$$
By taking the supremum over $\phi\in\CC^{L+1-s}$ on the left, we obtain 
$$\|U_1^\star \UE_2 \|_2
\le
\frac{\|\EE\|_2 +  (\se_{s+1})^2\|U_1^\star \UE_2\|_2 }{\si_s^2},$$
and then
\beq
\label{p1e6}
\|U_1^\star \UE_2 \|_2 \le  \frac{\|\EE\|_2}{\si_s^2- (\se_{s+1})^2}.
\eeq

Let $\PO$ and $\PEO$ be orthogonal projects onto the subspace spanned by columns of $U_1$ and $\UE_1$ respectively. For any $\phi \in \CC^{L+1}$
\begin{align}
 & \frac{\|\PET\phi -\PT\phi\|_2}{\|
\phi\|_2} 
\nonumber
 = \frac{\|\PO\PET\phi +\PT\PET\phi -\PT\phi\|_2}{\|\phi\|_2} 
= \frac{\|\PO\PET\phi -\PT\PEO\phi\|_2}{\|\phi\|_2} 
\nonumber
\\
& = \frac{\|U_1 U_1^\star \UE_2 {\UE_2}^\star \phi 
- U_2 U_2^\star \UE_1 {\UE_1}^\star\phi
\|_2}{\|\phi\|_2}
 \le \|U_1^\star \UE_2\|_2 + \| U_2^\star \UE_1\|_2.
\label{p1e7}
\end{align}

Eq. \eqref{p1e7} together with  \eqref{p1e5} and  \eqref{p1e6} imply  
$$|\RE(\om) -R(\om)| 
\le \|\PET-\PT\|_2 = \sup_{\phi\in \CC^{L+1}} \frac{\|\PET\phi-\PT\phi\|_2}{\|\phi\|_2}
\le \left[\frac{1}{(\se_s)^2} + \frac{1}{\si_s^2-(\se_{s+1})^2}\right]\|\EE\|_2.$$
Meanwhile 
$$\|\EE\|_2 \le 2\|H\|_2\|E\|_2 + \|E\|_2^2 \le (2\si_1+\|E\|_2)\|E\|_2,$$
and therefore 
$$\|\PET-\PT\|_2 \le (2\si_1 + \|E\|_2)\left[\frac{1}{(\se_s)^2} + \frac{1}{\si_s^2-(\se_{s+1})^2}\right]\|E\|_2.$$

According to Proposition \ref{propweyl}, $\se_s \ge \si_s -\|E\|_2$ and $\se_{s+1}\le \|E\|_2$. Then 
$$\frac{1}{(\se_s)^2} + \frac{1}{\si_s^2-(\se_{s+1})^2} 
\le
\frac{2}{(\si_s -\|E\|_2)^2},
$$
which implies that 
$$\|\PET-\PT\|_2 \le \frac{2(2\si_1 + \|E\|_2)}{(\si_s-\|E\|_2)^2}
\|E\|_2.$$

In particular, while $\om $ is restricted on $\supp$, a sharper upper bound in \eqref{eqp2} is derived as follows:

Since $M-L+1\ge s$ and true frequencies are pairwise distinct, $\XP$ has full row rank. Denote $Y^\ep = \UE_1 \SE_1{\VE_1}^\star +\UE_2 \SE_2{\VE_2}^\star$ where $\Sigma_1^\ep = \text{diag}(\se_1,\ldots,\se_s)$ and $\Sigma_2^\ep = \text{diag}(\se_{s+1},\se_{s+2},\ldots)$.  Multiplying ${\UE_2}^\star$ on the left and the pseudo-inverse of $\XP$ on the right of
$$\UE_1 \SE_1{\VE_1}^\star +\UE_2 \SE_2{\VE_2}^\star = \Phi^L X (\Phi^{M-L})^T + E$$
yields
$$
\SE_2{\VE_2}^\star[\XP]^\dagger ={\UE_2}^\star\Phi^L  +{\UE_2}^\star E[\XP]^\dagger.
$$
Then
$$
{\UE_2}^\star \phi^L(\om_j) ={\UE_2}^\star\Phi^L e_j ={\UE_2}^\star E[\XP]^\dagger e_j -\SE_2{\VE_2}^\star[\XP]^\dagger e_j,
$$
and
$$
\|\PET \phi^L(\om_j)\|_2 =\|{\UE_2}^\star \phi^L(\om_j)\|_2 \le \frac{\|E\|_2 + \si^\ep_{s+1}}{\smin(X(\Phi^{M-L})^T)} 
\le \frac{2\|E\|_2 }{\smin(X(\Phi^{M-L})^T)}
\le  \frac{2\|E\|_2 }{\xmin\smin((\Phi^{M-L})^T)}.
$$
Therefore
$$
\RE(\om_j) 
= \frac{\|\PET \phi^L(\om_j)\|_2}{\|\phi^L(\om_j)\|_2} 
\le  \frac{2\|E\|_2 }{\xmin\smin((\Phi^{M-L})^T)\|\phi^L(\om_j)\|_2}.
$$

\end{proof}

\commentout{
\subsection{Proof of Corollary \ref{cor2}}

\begin{proof}
\eqref{ce1} follows from Theorem \ref{thm3} and Theorem \ref{thmp1}. If external noise is bounded such that $\ep_j \le \sigma, \ j =0,\ldots, M,$

$$\|E\|_2 \le \|E\|_F \le \sigma\sqrt{(L+1)(M-L+1)}. $$

\end{proof}

}

\subsection{Proof of Theorem \ref{thmp2}}
\label{app4}

\begin{proof}
Let 
$$Q(\om) = R^2(\om)= \frac{{\plo}^\star U_2 U_2^\star U_2 U_2^\star \plo}{\|\plo\|_2^2} = \frac{{\plo}^\star U_2 U_2^\star U_2 U_2^\star \plo}{L+1}$$
and
$$Q^\ep(\om) = [{\RE}(\om)]^2 = \frac{{\plo}^\star \UE_2 {\UE}_2^\star \UE_2 {\UE}_2^\star \plo}{L+1}.$$

Both $Q(\om)$ and $Q^\ep(\om)$ are smooth functions and
\begin{align*}
Q'(\om) 
& = \frac{{\plo}^\star U_2 U_2^\star U_2 U_2^\star [\plo]' +
[\plo]'^\star U_2 U_2^\star U_2 U_2^\star {\phi^L}(\om)}{L+1}  \\
& = \frac{\langle \PT\plo, \PT [\plo]' \rangle + \langle \PT [\plo]', \PT \plo\rangle}{L+1},
\end{align*}

\begin{align}
Q''(\om) 
& = \frac{{\plo}^\star U_2 U_2^\star U_2 U_2^\star [\plo]'' 
+ 2{[\plo]' }^\star U_2 U_2^\star U_2 U_2^\star [\plo]' 
+ [\plo]''^\star U_2 U_2^\star U_2 U_2^\star \plo
}{L+1} \nonumber \\
& = \frac{\langle \PT \plo, \PT[\plo]'' \rangle + 2\|\PT [\plo]'\|_2^2 + \langle \PT[\plo]'', \PT\plo\rangle }{L+1}
\label{Qpp}
.
\end{align}

Let $D(\om) = \QE(\om) -Q(\om)$ and then 
$$[\QE(\om)]' = Q'(\om) + D'(\om),$$
$$[\QE(\om)]'' = Q''(\om) + D''(\om).$$

First, we derive an upper bound of $|D'(\om)|$ and $|D''(\om)|$ in terms of $\alpha,L$ and $\|E\|_2$.
\begin{align*}
(L+1)|D'(\om)|
&
= |\lan \PET\plo,\PET[\plo]'\ran -\lan \PT\plo,\PT[\plo]'\ran  
\\
& \quad 
+ \lan \PET[\plo]',\PET\plo\ran  - \lan \PT[\plo]',\PT\plo\ran |
\\
& 
= |\lan \PET\plo,\PET[\plo]'\ran   - \lan \PET\plo,\PT[\plo]'\ran  
\\
& \quad
 + \lan \PET\plo,\PT[\plo]'\ran  
-\lan \PT\plo,\PT[\plo]'\ran  
\\
& \quad 
+ \lan \PET[\plo]',\PET\plo\ran  -\lan \PET[\plo]',\PT\plo\ran 
\\
& \quad 
+ \lan \PET[\plo]',\PT\plo\ran 
 - \lan \PT[\plo]',\PT\plo\ran |
 \\
 &
 \le
  \|\PET\plo\|_2 \|\PET-\PT\|_2 \|[\plo]'\|_2
  + \|\PT[\plo]'\|_2 \|\PET-\PT\|_2 \|\plo\|_2
  \\
  & \quad
  +  \|\PET[\plo]'\|_2 \|\PET-\PT\|_2 \|\plo\|_2
  +  \|\PT\plo\|_2 \|\PET-\PT\|_2 \|[\plo]'\|_2
  \\
  & \le 4\|\PET-\PT\|_2 \|\plo\|_2 \|[\plo]'\|_2. 
\end{align*} 
Since $[\plo]' = -2\pi i[0  \ e^{-2\pi i \om} \ 2e^{-2\pi i 2\om} \ 3 e^{-2\pi i3 \om} \ \ldots \ L e^{-2\pi i L\om}]^T$,  we have
$\|[\plo]'\|_2 = \eta(L)\sqrt{L+1} $ where $\eta(L) = 2\pi \sqrt{1^2+2^2+\ldots+L^2}/\sqrt{L+1}$ and then
\beq
\label{p4e1}
|D'(\om)| \le 4 \alpha\eta(L){\|E\|_2},
\eeq
by Theorem \ref{thmp1}. 

Applying the same technique to $D''(\om)$ yields
\beq
\label{p4e2}
|D''(\om)| \le  4\alpha \left[\eta^2(L)+\zeta(L)\right] {\|E\|_2},
\eeq
where $\zeta(L) = (2\pi)^2\sqrt{1^4 + 2^4 +\ldots + L^4}/\sqrt{L+1}$.

Next we prove that for each $\om_j \in \supp$, there exists a strict local minimizer $\hat\om_j$ of $\RE$
near $\om_j$ satisfying \eqref{eq23}.

Since $\om_j$ is a strict local minimizer of $Q(\om)$ and 
 \beq
 \label{p4e20}
 \left. \PT [\plo]' \right|_{\om = \om_j} \neq \mathbf{0}, \quad \forall \om_j \in \supp
 \eeq
 we have 
 $$Q'(\om_j) = 0 \ \text{ and }\ Q''(\om_j) > 0.$$  
 
 We break the following argument into several steps:
 
 \begin{description}
 
 \item[Step 1]
 Let $Q''(\om_j) = 2m_j$. $m_j >0$ due to assumption \eqref{p4e20}. Thanks to the smoothness of $Q''$, there exists $\delta_j > 0$ such that  
 $$Q''(\om) > m_j , \ \forall \om \in (\om_j-\delta_j,\om_j+\delta_j).$$
 For sufficiently small noise satisfying 
 \beq
 \label{p4e3}
 4\alpha[\eta^2(L)+\zeta(L)] {\|E\|_2} < m_j/2,
\eeq
we have
 $$[\QE(\om)]'' > m_j/2, \ \forall \om \in(\om_j-\delta_j,\om_j+\delta_j).$$
 
 
 \item[Step 2] 
 Consider $Q'(\om_j - \delta_j)$ and $Q'(\om_j+\delta_j)$. There exist $\kappa_1,\kappa_2\in (0,\delta_j)$ such that
 \begin{align*}
 Q'(\om_j-\delta_j) & = Q'(\om_j) + Q''(\om_j-\kappa_1)(-\delta_j) \\
 Q'(\om_j+\delta_j) &= Q'(\om_j) + Q''(\om_j+\kappa_1)\delta_j.
 \end{align*}
 From Step 1, $Q''(\om_j-\kappa_1)>m_j$ and $Q''(\om_j+\kappa_1)>m_j$, so $Q'(\om_j-\delta_j) < -m_j\delta_j$ and $Q'(\om_j+\delta_j) > m_j\delta_j$.
 According to \eqref{p4e1}, when noise is sufficiently small such that 
\beq
\label{p4e4}
4\alpha \eta(L) {\|E\|_2}< m_j\delta_j/2,
\eeq
we have
$$[\QE(\om_j-\delta_2)]'  < -m_j\delta_j/2 < 0 \text{ and } [\QE(\om_j+\delta_3)]' > m_j\delta_j/2 > 0.$$
$[{\QE}(\om)]'$ is a smooth function so there exists $\hat\om_j \in (\om_j-\delta_j,\om_j+\delta_j)$ such that $[{\QE}(\hat\om_j)]' = 0$ by intermediate value theorem.

 \commentout{
As $\om_j$ is a strict local minimizer of the smooth function $Q(\om)$, there exists $\delta_2,\delta_3 \in (0,\delta_j]$ such that 
$$Q'(\om_j-\delta_2) < 0 \text{ and } Q'(\om_j+\delta_3) > 0.$$
Let $Q'(\om_j-\delta_2) =-p_j$ and $Q'(\om_j+\delta_3) = q_j$. 
}

\item [Step 3] From Step 2 and Step 3 we have obtained an open interval containing $\om_j$: $(\om_j-\delta_j,\om_j+\delta_j)$ such that 
$$Q''(\om) > m_j \text{  and  } [\QE(\om)]''>m_j/2, \ \forall \om\in(\om_j-\delta_j,\om_j+\delta_j).$$
Also there exists $\hat\om_j \in (\om_j-\delta_j,\om_j+\delta_j)$ such that
$$[\QE(\hat\om_j)]' = 0 \text{  and  }[\QE(\hat\om_j)]'' > 0,$$ 
so $\hat\om_j$ is a strict local minimizer of $\QE(\om)$ and $\RE(\om)$.

\item[Step 4] Furthermore 
\begin{align*}
 0 
 &= [{\QE}(\hat\om_j)]' = Q'(\hat\om_j) +D'(\hat\om_j) 
\\
& = Q'(\om_j) + {Q''}(\xi_j)(\hat\om_j -\om_j) + D'(\hat\om_j), \text{ for some } \xi_j \in (\om_j,\hat\om_j),
\end{align*}
through Taylor expansion of $Q'(\om)$ at $\om = \om_j$.
Since $Q'(\om_j) = 0$,
\beq
\label{p4e5}
|\hat\om_j-\om_j|\displaystyle\min_{\xi \in (\om_j,\hat\om_j)} |Q''(\xi)|
\le
|\hat\om_j-\om_j|  {|Q''(\xi_j)|}\le 
|D'(\hat\om_j)|
\le
4\alpha  \eta(L) {\|E\|_2} 
\eeq

\end{description}

Proof of Theorem \ref{thmp2} ends here. Next we provide the argument for the validation of Remark \ref{rm41} and Remark \ref{rm43}.

\begin{description}
\item [Remark \ref{rm41}] Suppose noise vector $\ep$ contains i.i.d. random variables of variance $\si^2$. For fixed $M$, $\|E\|_2 \le \|E\|_F^2 = \mathcal{O}(\sigma)$. Since $\min_{\xi \in (\om_j,\hat\om_j)}|Q''(\xi)| > m_j$, we have $|\hat\om_j-\om_j|\rightarrow 0$ as $\sigma \rightarrow 0$ and 
$$|\hat\om_j - \om_j| = \mathcal{O}(\sigma).$$

\item [Remark \ref{rm43}] The asymptotic rate of $\hat\om_j \rightarrow \om_j$ as $M = 2L \rightarrow \infty$ in the case of $q \ge$ 4 RL is discussed in Remark \ref{rm43}. Here we show that condition \eqref{p4e3} and \eqref{p4e4} hold as $M=2L \rightarrow \infty$ under assumption \eqref{eqremark6}. 

The left hand side (l.h.s.) of \eqref{p4e3} scales like $M\sqrt{M\log M}$. On the right hand side (r.h.s.), 
$$m_j = \frac{2\|\PT[\phi^L(\om_j)]'\|_2^2}{L+1} \ge \frac{2C_1}{L+1}\|[\phi^L(\om_j)]'\|_2^2 \sim M^2, \text{ as } M=2L \rightarrow \infty.$$
Therefore the r.h.s. of \eqref{p4e3} grows faster than the l.h.s. and \eqref{p4e3} holds as $M = 2L \rightarrow \infty$.

Next we show that \eqref{p4e4} is also valid as $M \rightarrow \infty$. First 
\begin{align}
\label{p4e6}
Q'''(\om) =&\frac{\langle \PT \plo, \PT[\plo]''' \rangle +  \langle \PT[\plo]''', \PT\plo\rangle }{L+1} \nonumber\\
& +\frac{3\langle \PT [\plo]', \PT[\plo]'' \rangle + 3 \langle \PT[\plo]'', \PT[\plo]'\rangle }{L+1}\nonumber
\end{align}
and then 
\begin{align*}
|Q'''(\om)| & \le \frac{2\|\plo\|_2\|[\plo]'''\|_2 + 6\|[\plo]'\|_2 \|[\plo]''\|_2}{L+1} \\
& = \frac{2\sqrt{L+1}\sqrt{1^6+2^6+\ldots+L^6}+6\sqrt{1^2+2^2+\ldots+L^2}\sqrt{1^4+2^4+\ldots+L^4}}{L+1}
\\
& = \mathcal{O}(L^3) \text{ as } L \rightarrow \infty\\
& = \mathcal{O}(M^3) \text{ as } M=2L \rightarrow \infty.
\end{align*}
In Step 1, for all $\om \in (\om_j-\delta_j,\om_j+\delta_j)$, $|Q''(\om)-Q''(\om_j)|$ can be as large as $m_j = \mathcal{O}(M^2)$. Meanwhile $Q''(\om) = Q''(\om_j) + (\om-\om_j)Q'''(\kappa)$ for some $\kappa \in (\om_j,\om)$ and then $|Q''(\om)-Q''(\om_j)| = |\om-\om_j||Q'''(\kappa)|$. Since $|Q''(\om)-Q''(\om_j)|$ can be as large as $\mathcal{O}(M^2)$ and $Q'''(\kappa) \le \mathcal{O}(M^3)$, $|\om-\om_j|$ can be as large as $\mathcal{O}(1/M)$. In other words, $\delta_j = \mathcal{O}(1/M)$ in Step 1.

In \eqref{p4e4}, the l.h.s. scales like $\sqrt{M\log M}$ and the r.h.s. $= m_j\delta_j/2=\mathcal{O}(M^2/M)= \mathcal{O}(M)$ as $M = 2L \rightarrow \infty$. As a result, \eqref{p4e4} holds while $M = 2L \rightarrow \infty$ under assumption \eqref{eqremark6}.
\end{description}
\end{proof}

\commentout{
\subsection{Proof of Theorem \ref{lemma3}}
\label{app3}
\begin{proof}
Let $\phi^L(\om) = \Phi^L c + \PT\phi^L(\om)$. Due to the fact that $\lan \Phi^L c,\PT\phi^L(\om)\ran  = 0$,
\begin{align*}
\|\Phi^L c\|_2^2 & = \re\lan \Phi^L c, \phi^L(\om)\ran  \\
&= \frac{\|\Phi^L c + \phi^L(\om)\|_2^2-\|\Phi^L c - \phi^L(\om)\|_2^2}{4} \\
& \le \frac{\smax^2(\Phi^L_\om)-\smin^2(\Phi^L_\om)}{4}(1+\|c\|_2^2) \\
& \le  \frac{\smax^2(\Phi^L_\om)-\smin^2(\Phi^L_\om)}{4}\Big(1+\frac{\|\Phi^L c\|_2^2}{\smin^2(\Phi^L)}\Big).
\end{align*}
On the condition of \eqref{lemma31} we have
\begin{align*}
\Big(1-\frac{\smax^2(\Phi^L_\om)-\smin^2(\Phi^L_\om)}{4\smin^2(\Phi^L)}\Big) \|\Phi^L c\|_2^2 
\le
  \frac{\smax^2(\Phi^L_\om)-\smin^2(\Phi^L_\om)}{4}
\cdot 
\frac{\|\phi^L(\om)\|_2^2}{L+1}.
\end{align*}
Furthermore, if \eqref{lemma32} is satisfied,
\begin{align*}
\|\Phi^L c\|_2^2
\le
 \frac{\smax^2(\Phi^L_\om)-\smin^2(\Phi^L_\om)}{4(L+1)}
\Big(1-\frac{\smax^2(\Phi^L_\om)-\smin^2(\Phi^L_\om)}{4\smin^2(\Phi^L)}\Big)^{-1} 
\|\phi^L(\om)\|_2^2.
\end{align*}
Therefore
$$\frac{\|\PT\plo\|_2}{\|\plo\|_2} 
\ge
 \sqrt{1-\frac{\smax^2(\Phi^L_\om)-\smin^2(\Phi^L_\om)}{4(L+1)}
\Big(1-\frac{\smax^2(\Phi^L_\om)-\smin^2(\Phi^L_\om)}{4\smin^2(\Phi^L)}\Big)^{-1} 
}.$$
\end{proof}
}

\section{Proof of Theorem \ref{thmd1}}
\label{appd}

\begin{proof}
We partition  $\supp$ into the subsets  
$
\supp_m  = \{\om_{m+kR}: k\in \ZZ\}, m=1,\dots, R,$
each of which  satisfies the gap condition: 
$d(\om_j,\om_l) > R \rho \ \text{ for all } \om_j,\om_l \in \supp_m \text{ with } j\neq l.$ According to Theorem \ref{thm3},
$$\frac{1}{L} \sum_{k=0}^L \left| \sum_{\om_j\in\supp_m\cap [0,1)} \bc_j e^{-2\pi i k\om_j}\right|^2 
\le
B(R\rho,L) \sum_{\om_j\in\supp_m\cap [0,1)} |\bc_j|^2, \ m = 1,\ldots,R.$$
As $|z_1+\ldots+z_R|^2 \le R(|z_1|^2+\ldots+|z_R|^2)$,
\begin{align*}
& \frac{1}{L} \sum_{k=0}^L \left| \sum_{\om_j\in\supp\cap [0,1)} \bc_j e^{-2\pi i k\om_j}\right|^2 
 = \frac{1}{L} \sum_{k=0}^L \left| \sum_{m=1}^R \sum_{\om_j\in\supp_m\cap [0,1)} \bc_j e^{-2\pi i k\om_j}\right|^2 
 \\
 &
\leq {R \over L} \sum_{m=1}^R  \sum_{k=0}^L \left|\sum_{\om_j\in\supp_m\cap [0,1)} \bc_j e^{-2\pi i k\om_j}\right|^2 
\leq R \sum_{m=1}^R  \sum_{\om_j\in\supp_m\cap [0,1)} B(R\rho,L) |\bc_j|^2
\\
&
\leq  B(R\rho,L) R\sum_{m=1}^R  \sum_{\om_j\in\supp_m\cap [0,1)}  |\bc_j|^2
 =  B(R\rho,L) R\|\bc\|_2^2.
\end{align*}
\end{proof}

\end{appendices}

\bibliographystyle{plain}	
\bibliography{myref}		
 
\end{document}